\documentclass[10pt]{article}

\usepackage{my-details}

\usepackage[utf8]{inputenc} % allow utf-8 input
\usepackage[T1]{fontenc}    % use 8-bit T1 fonts
\usepackage{hyperref}       % hyperlink
\hypersetup{
  colorlinks   = true, %Colours links instead of ugly boxes
  urlcolor     = blue, %Colour for external hyperlinks
  linkcolor    = blue, %Colour of internal links
  citecolor   = blue %Colour of citations, could be ``red''
}
\usepackage{amsfonts}       % blackboard math symbols
\usepackage{microtype}      % microtypography
\usepackage{graphicx}  % for inserting graphics
\usepackage[usenames]{color} %used for font color
\usepackage[UKenglish]{babel}
\usepackage{amssymb} %maths
\usepackage{amsmath} %maths
\usepackage{amsthm}
\usepackage{tcolorbox}
\theoremstyle{plain}
\newtheorem{theorem}{Theorem}[section]
\newtheorem{corollary}[theorem]{Corollary}
\newtheorem{prop}[theorem]{Proposition}
\newtheorem{lemma}[theorem]{Lemma}

\newtheorem{definition}[theorem]{Definition}

\DeclareMathOperator{\Tr}{Tr}

\usepackage{color,hyperref}
\usepackage{physics}

\newcommand{\mc}{\mathcal}

\newcommand{\Riem}{\mathrm{Riem}}
\newcommand{\supp}{\mathrm{supp}}
\newcommand{\std}{\mathrm{std}}

\DeclareFontFamily{U}{mathx}{}
\DeclareFontShape{U}{mathx}{m}{n}{<-> mathx10}{}
\DeclareSymbolFont{mathx}{U}{mathx}{m}{n}
\DeclareMathAccent{\widehat}{0}{mathx}{"70}
\DeclareMathAccent{\widecheck}{0}{mathx}{"71}

\usepackage{caption}
\usepackage{physics}
\usepackage{mathtools}
\usepackage{enumitem}
\usepackage{bbm} %identity operator
\usepackage{float} %to keep figures in place
\usepackage{xcolor} %so that I can edit color of Refs
\usepackage{authblk}

\newcommand{\paula}[1]{{\color{blue}[PB: #1]}}

\title{Quantum $f$-divergences and Their Local Behaviour: \\An Analysis via Relative Expansion Coefficients}

\begin{document}

%\author{Shreyas Iyer\footnote{Corresponding Author}, Peixue Wu, Paula Belzig, Graeme Smith}

\author{Shreyas Iyer\thanks{Corresponding author: s9natarajan@uwaterloo.ca}}
\author{Peixue Wu}
\author{Paula Belzig}
\author{Graeme Smith}

\affil{Department of Applied Mathematics, University of Waterloo, Waterloo, Ontario, Canada \\ Institute for Quantum Computing, University of Waterloo, Waterloo, Ontario, Canada}

% \author[1,2]{Paula Belzig}
% \author[1,2]{Shreyas Iyer \thanks{Corresponding author}}
% \author[1,2]{Graeme Smith}
% \author[1,2]{Peixue Wu}

% \affil[1]{Department of Applied Mathematics, University of Waterloo, Waterloo, Ontario, Canada}
% \affil[2]{Institute for Quantum Computing, University of Waterloo, Waterloo, Ontario, Canada}

% \author{Shreyas Iyer}
%  % \email{s9natara@uwaterloo.ca}
%  % \altaffiliation[Also at ]{Physics Department, XYZ University.}%Lines break automatically or can be forced with \\
% \author{Peixue Wu}%
%  % \email{p33wu@uwaterloo.ca}
% \author{Paula Belzig}
%  % \email{pbelzig@uwaterloo.ca}
% \author{Graeme Smith}
 % \email{graeme.smith@uwaterloo.ca}
% \affiliation{%
%  Institute for Quantum Computing\\
%  University of Waterloo\\
%  Waterloo, ON, Canada
% }%

%\date{12th September 2025}

\maketitle

\DeclareRobustCommand\Compactcdots{\mathinner{\cdotp\mkern-1mu\cdotp\mkern-1mu\cdotp}}

\begin{abstract}
%Any reasonable measure of quantum information or distinguishability of quantum states must satisfy a data processing inequality, that is, it must not increase under the action of a quantum channel. In this work, w

Any reasonable measure of distinguishability of quantum states must satisfy a data-processing inequality, that is, it must not increase under the action of a quantum channel. We can ask about the proportion of information lost or preserved and this leads us to study contraction and expansion coefficients respectively, which can be combined into a single \emph{relative expansion coefficient}. 
We focus on two prominent families: (i) standard quantum $f$-divergences and (ii) their local (second-order) behaviour, which induces a monotone Riemannian semi-norm (that is linked to the $\chi^2$–divergence). Building on prior work, we identify new families of $f$ for which the global ($f$-divergence) and local (Riemannian) relative expansion coefficients coincide for every pair of channels, and we clarify how exceptional such exact coincidences are. Beyond equality, we introduce an \emph{equivalence} framework that transfers qualitative properties such as strict positivity uniformly across different relative expansion coefficients.

Leveraging the link between equality in the data processing inequality (DPI) and channel reversibility, we apply our framework of relative expansion coefficients to approximate recoverability of quantum information. Using our relative expansion results for primitive channels, we prove a reverse quantum Markov convergence theorem, converting positive expansion coefficients into quantitative lower bounds on the convergence rate.

\end{abstract}

\tableofcontents
\section{Introduction}

\noindent 

This work offers a detailed comparison of distinguishability measures of quantum states that are closely related to the quantum relative entropy. Namely, we focus on standard $f$-divergences and quantum $\chi^2$-divergences (the latter we often refer to via Riemannian semi-norms). It is of central importance in quantum information theory to analyse various ways of quantifying the difference between two quantum states, because they can be used to examine the reliability of protocols, such as the transmission of information across a quantum channel. Other examples include the fidelity \cite{jozsa1994fidelity,fuchs2002cryptographic}, trace distance \cite{holevo1973statistical,helstrom1969quantum,audenaert2007discriminating}, and quantum Rényi divergences \cite{nussbaum2009chernoff,nussbaum2011asymptotic,muller2013quantum}, which each have their own specific use cases. The quantum relative entropy not only expresses the Type-II error exponent of state discrimination in the quantum Stein's lemma \cite{hiai1991proper,ogawa2005strong}, but it is also used to capture the typical behaviour for compression in channel coding theorems \cite{preskill1998lecture,Wildebook}. Specifically, it defines the coherent information and mutual information seen throughout quantum Shannon theory \cite{schumacher1996quantum,lloyd1997capacity}. As a consequence of its operational meaning, the quantum relative entropy serves as a guide for the properties one should expect from a good distinguishability measure. 

Classical and quantum $f$-divergences are families of nicely-behaving distinguishability measures, inspired from the relative entropy \cite{csiszar1963informationstheoretische,Petz84,Petz86}. We will consider them in the finite-dimensional setting. Classical $f$-divergences are a family of functionals \[D_f^\mathrm{cl}:\mathbb{R}^n\times\mathbb{R}^n\to [0,\infty)\] that measure the distinguishability of two probability vectors $P,Q\in\mathbb{R}^n$, where $f:(0,\infty)\to\mathbb{R}$ is a convex function, $f(1)=0,f''(1)>0$. They were used to generalise the Shannon-type quantities defined using the classical relative entropy (KL-divergence), corresponding to $f(x)=x\log x$, and have had great success in classical statistics \cite{csiszar1963informationstheoretische,csiszar1967topological,liese1987convex,liese2006divergences,SV16}. Quantum $f$-divergences, on the other hand, are classes of functionals 
\[D_f:\mathcal{B}(\mathcal{H})_+\times\mathcal{B}(\mathcal{H})_+\to [0,\infty),\] such that $D_f(\rho\|\gamma)$ reduces to a classical $f$-divergence $D_f^\mathrm{cl}$ over commuting states $\rho,\gamma$ \cite{HM17}. 
Quantum $f$-divergences often share many desirable properties with the quantum relative entropy, for example: positivity, continuity over all pairs of states when finite, and the data-processing inequality under completely-positive trace-preserving (CPTP) maps $\mathcal{N}$ \cite{schumacher2000relative,lindblad1975,Hiai_2011,HM17,HT24}. There are many variants of quantum $f$-divergences, and this includes the standard $f$-divergences ($D_f^\mathrm{std}$) \cite{Petz86,OP93,Hiai_2011} and quantum $\chi^2$-divergences ($\chi^2_\kappa$) \cite{GT25} that this work is focused on. Other examples have also attracted a lot of attention over the years, such as the measured ($D_f^\mathrm{meas}$), maximal ($\widehat D_f$) \cite{PR98,Matsumoto15,HM17}, and Hirche-Tomamichel quantum $f$-divergences ($\widetilde{D}_f$) \cite{HT24,H2024,beigi2025,liu2025}. Further, we can reinterpret Petz-Rényi divergences $D_\alpha$ and sandwiched Rényi divergences $\widetilde{D}_\alpha$ through these quantum $f$-divergences (see, for example \cite{HT24}). 

The data-processing inequality is at core of our discussion. It states that
\[
D_f(\mathcal{N}(\rho)\|\mathcal{N}(\gamma))\leq D_f(\rho\|\gamma)\text{ for all states }\rho,\gamma, \text{and CPTP maps }\mathcal{N}:\mathcal{B}(\mathcal{H}_A)\to\mathcal{B}(\mathcal{H}_B).
\]

This is a vital property of any valid distinguishability measure as it ensures consistency with the inherent loss of information induced by a noisy quantum channel; in this case, the states become harder to distinguish. This allows us to study the proportion of information preserved by a quantum channel with respect to any of these distinguishability measures, via the following contraction and expansion coefficients generalised/modified from prior works \cite{LR99,HR15,HRS22,GZB24,RISB21,BGSW24}:
\[
\eta_f(\mathcal{N};\mathcal{C}):= \sup_{\substack{\rho\neq \gamma \in \mc C}} \frac{D_f(\mc N(\rho) \|\mc N(\gamma))}{D_f(\mc \rho \|\mc \gamma)},\ \quad
\widecheck\eta_f(\mathcal{N};\mathcal{C}):= \inf_{\substack{\rho\neq \gamma \in \mc C}} \frac{D_f(\mc N(\rho) \|\mc N(\gamma))}{D_f(\mc \rho \|\mc \gamma)},
\]
where $\mc C\subseteq\mathcal{D}(\mathcal{H}_A)$ is a set of density operators, and $\mathcal{N}$ is a quantum channel/CPTP map with domain $\mathcal{B}(\mathcal{H}_A)$. The titular relative expansion coefficients $\widecheck\eta_f(\mathcal{N},\mathcal{M})$ compare the actions of two quantum channels $\mathcal{N},\mathcal{M}$, and effectively merge the above two coefficients into a single definition, at least when $\mathcal{C}$ is the image of another quantum channel, allowing us to study them simultaneously.

Contraction and expansion coefficients have been considered by many prior works, again under less generality. Contraction coefficients in the setting $\mathcal{C}=\mathcal{D}(\mathcal{H}_A)$ capture the proportion of distinguishability that is lost under $\mathcal{N}$ over all pairs of full-rank quantum states. They were extensively studied in both classical and quantum contexts to compare different distinguishability measures \cite{Choi1994,Cohen1993,LR99,Temme_2010,HR15,HM17,HT24,HRS22}. After all, the variation in these unrestricted contraction coefficients does not appear in the classical setting, which means that any differences we find are a consequence of how these distinguishability measures behave on non-commuting pairs of quantum states. %Comparing two arbitrary channels in the relative expansion coefficient has allowed us to develop a refined comparison of these distinguishability measures, as we could clearly see the impact of restricting the classes of both channels. 

One of the first findings in this paper is that if $\mc C=\mathcal{D}(\mc H)$, i.e. we optimise over all states, then for a large class of channels, no proportion of information is preserved with respect to any monotone quantum $f$-divergence (i.e. a quantum $f$-divergence satisfying the data-processing inequality):
\[
\widecheck\eta_f(\mathcal{N};\mathcal{D}(\mc H))=0\quad\text{for all non-unitary channels }\mc N:\mathcal{B}(\mathcal{H}_A)\to\mathcal{B}(\mathcal{H}_B),\;d_A\geq d_B.
\]
This resolves an open problem by \cite{BGSW24} (which proved this result for the quantum relative entropy) about whether this is common to all monotone quantum $f$-divergences, under the reasonable --- though sometimes unnecessary --- condition that $f$ is operator convex. We infer from this that for such non-unitary channels, it only makes sense to consider expansion coefficients over a strict subset $\mc C\subset\mathcal{D}(\mathcal{H}_A)$. 

We are largely interested in the specific class of standard (quantum) $f$-divergences \cite{OP93}, parametrised by functions $f \equiv f(x)$ from the subset $\mathcal{F}$ of operator convex functions, defined for density operators $\rho,\gamma$ with full rank (or equal support) by:
\begin{equation*}
    D_f(\rho\|\gamma):=\langle\gamma^{1/2},f(\Delta_{\rho,\gamma})(\gamma^{1/2})\rangle_\mathrm{HS}
\end{equation*}
where $\Delta_{\rho,\gamma}:=L_\rho R_\gamma^{-1}$ is the superoperator that left-multiplies by $\rho$ and right-multiplies by the generalised inverse $\gamma$. Note that the choice $f(x)=x\log x$ corresponds to the quantum relative entropy. This class is interesting because the standard $f$-divergences demonstrate different local behaviours. This is in contrast to the maximal and Hirche-Tomamichel $f$-divergences, each of which have uniform local behaviour over all choices of $f$. In fact, the local second order behaviour of a standard $f$-divergence $D_f^\mathrm{std}(\rho\|\gamma)$ is the $\chi^2$-divergence, $\chi_{\kappa_f}^2(\rho\|\gamma)$, for some $\kappa_f$ in another subset of operator convex, decreasing functions $\mathcal{K}$. This work is partially motivated from understanding the connection between the standard $f$-divergences and their local behaviour. We denote their respective divergence and Riemannian contraction and expansion coefficients by $\eta_f^\mathrm{std}(\mathcal{N};\mathcal{C}),\widecheck \eta_f^\mathrm{std}(\mathcal{N};\mathcal{C})$ and $\eta_\kappa^\mathrm{Riem}(\mathcal{N};\mathcal{C}),\widecheck \eta_\kappa^\mathrm{Riem}(\mathcal{N};\mathcal{C})$. Conditions of the form $\eta_f^\mathrm{std}(\mathcal{N};\mathcal{C})\equiv\eta_{\kappa_f}^\mathrm{Riem}(\mathcal{N};\mathcal{C})$, $\widecheck \eta_f^\mathrm{std}(\mathcal{N};\mathcal{C})\equiv\widecheck \eta_{\kappa_f}^\mathrm{Riem}(\mathcal{N};\mathcal{C})$ for all channels $\mathcal{N}$ and open subsets $\mathcal{C}$ imply that the best- and worst- case preserved proportions of information with respect to $D_f^\mathrm{std}$ are totally determined by its local behaviour. It is surprising that such cases exist; some valid choices for $f$ were provided in \cite{HR15}, but it was noted that this relies on a rather special relationship studied in \cite{LR99, AE11,HR15,GR22,WBCDT24} between $D_f^\mathrm{std}(\rho\|\gamma)$ and $\chi_{\kappa_f}^2(\rho\|\gamma)$, and it was left open whether $f(x)=x\log x$ is also a valid choice. This open problem has been resolved in \cite{GR22,BGSW24}, but still only a finite number of choices for $f$ for which equality holds were identified. In this work, we extended these results to two infinite families of valid $f$ for which this equality of coefficients generically holds. Generic equality, even under restricted classes of channels or other classes of quantum $f$-divergences appears to rely on the existence of an integral relationship between the quantum $f$-divergence and their induced Riemannian semi-norm. 
%\paula{We also noted the importance of the unboundedness of $\kappa_f(x)$ for these integral relationships to be possible, which we propose as a direction for further investigation.}

Furthermore, by introducing a notion of \emph{equivalence} between distinguishability measures via their relative expansion coefficients, we establish equivalence classes in which different relative expansion coefficients (hence also the contraction/expansion coefficients) can always be used to upper and lower bound each other's values, despite not usually having a generic equality between coefficients. 

Conversely, we have identified an example where the unrestricted-domain contraction coefficients can be separated by arbitrarily many orders of magnitude, which is of interest when trying to bound the convergence rate of a quantum Markov chain. Previous work has  largely concentrated on the \emph{strong data processing inequality} (SDPI) for a quantum channel $\mathcal{N}$. Concretely, one asks whether there exists a contraction coefficient $\eta<1$ such that 
\[
    D_f^\mathrm{std}(\mathcal{N}(\rho)\,\|\,\mathcal{N}(\gamma)) \leq \eta \, D_f^\mathrm{std}(\rho\|\gamma),
    \qquad \forall \rho\neq \gamma,
\]
for a standard $f$-divergence $D_f^\mathrm{std}$. An analogous question can be asked for the quantum $\chi^2$-divergence induced by $D_f^\mathrm{std}$. Contraction coefficents have been particularly useful for bounding the mixing times of quantum Markov chains, especially in the case of primitive channels (i.e., those with a unique full-rank fixed point) \cite{Choi1994, Cohen1993, LR99, GZB24, GT25, HR15, HT24, Temme_2010}. When the channel admits a unique fixed point, the SDPI (with the reference state given by this fixed point) is closely connected to the modified logarithmic Sobolev inequality (MLSI) \cite{GJLL22,LS23,araiza2025transportationcostcontractioncoefficient,GR22}, which often yields sharper mixing time estimates than spectral gap methods. Similarly, one can investigate a form of reverse data processing inequality (RDPI) for a quantum channel $\mathcal{N}$: does there exist an expansion coefficient $\widecheck\eta>0$ and a quantum channel $\mathcal{M}$ such that
\begin{equation*}
    D_f^\mathrm{std}(\mathcal{N}(\rho)\,\|\,\mathcal{N}(\gamma)) \;\geq\; \widecheck\eta \, D_f^\mathrm{std}(\rho\|\gamma),
    \ \forall \rho\neq \gamma \in \mathrm{Im}\,\mathcal{M}\,? 
\end{equation*}
We will see later that quantum Markov chains based on primitive channels again provide a natural application of expansion coefficients as a lower bound on the convergence rate (see Section~\ref{subsec:primitive}). 
\color{black}

In addition, we propose another fresh perspective on expansion coefficients in terms of recoverability. We can reformulate a fundamental theorem that links channel reversibility to standard $f$-divergences, \cite[Theorem 7.1]{Hiai_2011}, in the language of expansion coefficients; it states that the following condition holds
\[
    \widecheck\eta_f(\mathcal{N};\mathcal{C})=1\text{ for a closed, convex subset } \mathcal{C}\subset\mathcal{D}(\mathcal{H}_A)\text{ for all }f\in\mathcal{F}
\]
iff all of the states in $\mathcal{C}$ can be perfectly recovered by a Petz recovery map. There has also been significant progress towards a similar result for $\chi^2$-divergences \cite{gao2023sufficient}. Hence, we sought to develop this connection between expansion coefficients and recoverability. There is another remarkable result, by \cite{junge2018universal}, that has important implications to approximate quantum error correction. It states that the output of a quantum channel can be restored to the initial state at an accuracy depending on the decrease in its quantum relative entropy with a reference state. We make an observation from this that showing positive expansion coefficients $\widecheck\eta_{x\log x}(\mathcal{N};\mathcal{D})$ for a quantum channel $\mathcal{N}$ amounts to upper bounding the accuracy of the universal recovery map considered in the latter result. A method developed by \cite{gao2023sufficient} analogously connects the extent of recoverability to the decrease in the $\chi^2$-divergence with a reference state, and so positive Riemannian expansion coefficients find a similar application to recovery bounds; this is exciting because it is often easier to demonstrate that Riemannian expansion coefficients are strictly positive (see Section~\ref{sec:explicitcoeffs}).

The remainder of the manuscript is organized as follows. 
\begin{itemize}
    \item Section~\ref{sec:prelims} reviews quantum $f$-divergences and operator convex functions, and introduces \emph{relative expansion coefficients}, with particular attention to standard $f$-divergences and their induced Riemannian semi-norms. 
    \item Section~\ref{sec:nordpi} establishes our first global obstruction: for any quantum $f$-divergence, the expansion coefficient over all states vanishes whenever the channel’s output dimension is no larger than its input dimension. 
    \item In Section~\ref{sec:equalityofcoeffs} we show that, in general, the (standard $f$-divergence) relative expansion is not determined by its induced Riemannian counterpart; nevertheless, we identify integral relationships between $f$-divergences and their Riemannian geometry that do guarantee generic equality of the corresponding relative expansion coefficients. 
    \item Section~\ref{sec:equivalence} then relaxes equality to an \emph{equivalence} notion between coefficients, which both propagates known cases to new ones and cleanly separates the bounded versus unbounded regimes for Riemannian relative expansion.
    \item In Section~\ref{sec:applications}, we discuss how demonstrations of positive expansion coefficients can be incorporated into approximate recovery bounds, which strengthens their appeal as a measure of preserved information. Further, primitive quantum channels also have a positive relative expansion coefficient, that can be used to lower bound the convergence rate. 
\end{itemize}

\section{Preliminaries}\label{sec:prelims}
\subsection{Notation}
We will consider finite dimensional Hilbert spaces, typically $\mathcal{H}_A,\mathcal{H}_B$, with respective dimensions $d_A:=\dim\mathcal{H}_A,d_B:=\dim\mathcal{H}_B$. $\mathcal{B}(\mathcal{H}_A,\mathcal{H}_B)$ denotes the Hilbert space of bounded linear operators from $\mathcal{H}_A$ to $\mathcal{H}_B$, equipped with the Hilbert-Schmidt inner product: 
\begin{equation*}
    \langle X, Y \rangle_\mathrm{HS} := \Tr(X^* Y)
\end{equation*}
for operators $X,Y\in\mathcal{B}(\mathcal{H}_A,\mathcal{H}_B)$, where $(\cdot)^*$ is the adjoint. We denote the image space of an operator $A\in\mathcal{B}(\mathcal{H}_A,\mathcal{H}_B)$ by $\mathrm{Im}\:A$, which has dimension $\mathrm{rk}(A):=\dim(\mathrm{Im}\:A)$, and the kernel of $A$ by $\mathrm{ker}\:A$, whose orthogonal space (i.e. $A$'s \textit{support}) is $\mathrm{supp\:A}:=(\mathrm{ker}\:A)^\perp$. For a Hilbert space $\mathcal{H}$, $\mathcal{B}(\mathcal{H},\mathcal{H})\equiv\mathcal{B}(\mathcal{H})$. The real subspace of self-adjoint linear operators in $\mathcal{B}(\mathcal{H})$ is denoted $\mathcal{B}(\mathcal{H})_\mathrm{sa}$ and the subspace of positive (semi-definite) linear operators is denoted $\mathcal{B}(\mathcal{H})_+$. The set of density operators on a Hilbert space $\mathcal{H}$ is denoted as $\mathcal{D}(\mathcal{H})$, and the subset of positive definite density operators is denoted $\mathcal{D}^+(\mathcal{H})$. When $\mathrm{dim}\;\mathcal{H}=d$, we use $\mc D_d$ and $\mc D_d^+$, respectively. Given a state $\rho\in\mathcal{D}(\mathcal{H})$, we define $T_\rho\mathcal{D}(\mathcal{H}):=\{X\in\mathcal{B}(\mathcal{H}):\ \Tr X=0,\ X=X^*,\ \mathrm{supp}\:X\leq\mathrm{supp}\:\rho\}$ as an appropriate set of \textit{tangent vectors} at $\rho$. We will often use $\mathcal{N},\mathcal{M},\Phi$ to denote quantum channels, i.e. completely positive trace-preserving (CPTP) maps, $\mathcal{B}(\mathcal{H}_A)\to\mathcal{B}(\mathcal{H}_B)$ for Hilbert spaces $\mathcal{H}_A,\mathcal{H}_B$. $\widehat\Phi$ will denote the adjoint map $\mathcal{B}(\mathcal{H}_B)\to\mathcal{B}(\mathcal{H}_A)$, which is positive and unital. We will sometimes also consider classical channels, which are linear maps (or rather, maps that have a well-defined linear extension), $\Phi:P_{d_A}\to P_{d_B}$, where $P_d:=\{x\in \mathbb{R}^d:x_i\geq0\:\forall i,\:\sum_i x_i=1\}$ is the set of $d$-dimensional probability vectors. 

For any $A\in\mathcal{B}(\mathcal{H})_\mathrm{sa}$, $A^{-1}$ will mean the \textit{generalised inverse} of $A$, where only the non-zero eigenvalues of $A$ are inverted. We will often consider the spectral decomposition of a density operator $\rho\in\mathcal{D}(\mathcal{H})$, which we write as $\rho=\sum_{a\in\mathrm{spec}\:\rho}aP_a$, where $\mathrm{spec}\:\rho$ is the spectrum of $\rho$ and the $P_a$ are projections onto the $a$-eigenspace of $\rho$. The (super-)operators $L_A,R_A\in\mathcal{B}(\mathcal{B}(\mathcal{H}))$ are defined as $L_A(X):=AX,R_A(X):=XA$, for arbitrary $A,B\in\mathcal{B}(\mathcal{H})_\mathrm{sa}$. Note that, e.g. $L_A^{-1}(X):=A^{-1}X$. 
% In particular, the \textit{relative modular operator} is defined as $\Delta_{\rho,\gamma}:= L_\rho R_\gamma^{-1}$ for any states $\rho,\gamma\in\mathcal{D}(\mathcal{H})$. 

\subsection{An Introduction to $f$-divergences}\label{sec:introfdivs}
The classical relative entropy (i.e. the \textit{Kullback-Leibler divergence}) is a measure of distinguishability of two probability distributions that can be used to define the entropic quantities in classical Shannon theory.
The classical $f$-divergences \cite{AliSilvey1966,csiszar1963informationstheoretische,csiszar1967information,csiszar1967topological,morimoto1963markov}, parametrised by convex functions $f:(0,\infty)\to\mathbb{R}$ s.t. $f(1)=0,f''(1)>0$, generalise this relative entropy and are used to define information-theoretic quantities in a similar way. The classical $f$-divergences share a lot of convenient properties with the classical relative entropy.  discrete probability distributions over a finite set $\mathcal{Y}$, $P(y),Q(y),y\in\mathcal{Y}$ with full support ($\mathcal{Y}$ a finite set), the classical $f$-divergence is:
\begin{equation*}
    D_f^\mathrm{cl}(P\|Q):=\sum_{y\in\mathcal{Y}} Q(y)f\left(\frac{P(y)}{Q(y)}\right).
\end{equation*}
%\paula{$x\in\mc X$ vs the $x$ in $f(x)$}\textcolor{purple}{What other variable could I use to write a function?}\paula{like, y}
The classical relative entropy corresponds to the special case where $f(x)=x\log x$. 

In the quantum setting, in order to distinguish two quantum states $\rho,\gamma\in\mathcal{D}_d$ rather than probability distributions, the classical $f$-divergences are extended to \emph{quantum $f$-divergences}. Quantum $f$-divergences, $D_f(\rho\|\gamma)$, are distinguishability measures on quantum states that reduce to $D_f^\mathrm{cl}(P\|Q)$ if $\rho,\gamma$ commute and have respective eigenvalues $P(y),Q(y),\ y\in\{1,...,d\}$ on their joint eigenbasis. Many classes of quantum $f$-divergences have been studied for their interesting properties and operational interpretations~\cite{HM17}, e.g., measured, maximal \cite{PR98,Matsumoto15,HM17}, Hirche-Tomamichel \cite{HT24,H2024,beigi2025,liu2025},  standard \cite{Petz86,OP93,Hiai_2011} and $\chi^2$ \cite{GT25} quantum $f$-divergences. We will mostly be interested in the standard (Petz) $f$-divergence. For these quantum $f$-divergences to be considered as valid distinguishability measures, we want them to satisfy monotonicity/the data-processing inequality, i.e. for any quantum channel $\mathcal{N}:\mathcal{B}(\mathcal{H}_A)\to\mathcal{B}(\mathcal{H}_B)$, the $f$-divergence of any two states cannot increase:
\begin{equation*}
    D_f(\mathcal{N}(\rho)\|\mathcal{N}(\gamma))\leq D_f(\rho\|\gamma)\quad\forall\rho,\gamma\in\mathcal{D}(\mathcal{H}_A)
\end{equation*}
To ensure this, we will have to consider a restricted class of functions fulfilling \cite{Hiai_2011,HM17} %\paula{i wouldn't use see later here if the later is on the next page, and also just stating the exact same definition. maybe it can be reformulated so that it does not have to be restated on the next page}\textcolor{purple}{uniformly used see later, to be easily found and replaced with section names later... it will vanish soon. This ordering is already the most natural without being overwhelming.} 
\[
\mathcal{F}:=\{f:(0,\infty)\to\mathbb{R}, f\text{ operator convex},f(1)=0,f''(1)>0\}.
\]

The maximal and standard quantum $f$-divergences arise from alternative notions of Radon-Nikodym derivative in the non-commutative setting; to replace the random variable taking values $\frac{P(x)}{Q(x)}$ with outcome probabilities $P(x)$, which is the unique Radon-Nikodym derivative in the commutative setting, one can introduce different notions of $\rho/\gamma$ \cite{LR99}.

%\paula{The maximal and standard $f$-divergences in quantum theory arise from alternative notions of a random variable taking values $\frac{P(x)}{Q(x)}$ with probability $P(x)$ (Radon-Nikodym derivative) in the non-commutative setting \cite{LR99}.}\textcolor{purple}{[if it's not clear upon re-reading, let's discuss.]}

The maximal $f$-divergence has the following expression for states $\rho,\gamma\in\mathcal{D}(\mathcal{H})$ with equal support:
\begin{equation}\label{eqn:maxfdiv}
    \widehat D_f(\rho\|\gamma):=\langle\gamma^{1/2},f(\widehat\Delta_{\rho,\gamma})(\gamma^{1/2})\rangle_\mathrm{HS}\equiv \Tr\gamma f(\gamma^{-1/2}\rho\gamma^{-1/2})
\end{equation}
where $\widehat\Delta_{\rho,\gamma}:=R_{\gamma^{-1/2}\rho\gamma^{-1/2}}$ is the commutant Radon-Nikodym derivative.

Whereas the standard $f$-divergence is defined as:
\begin{equation}\label{eqn:stdfdiv}
     D_f^\mathrm{std}(\rho\|\gamma):=\langle\gamma^{1/2},f(\Delta_{\rho,\gamma})(\gamma^{1/2})\rangle_\mathrm{HS}
\end{equation}
where $\Delta_{\rho,\gamma}:=L_\rho R_\gamma^{-1}$ is the \emph{relative modular operator} (another Radon-Nikodym derivative).
%\paula{this is already in Notation}\textcolor{purple}{keep in notation or move entirely here?}\paula{depends how much we want to mention the RN derivative. my feeling is: mention it only in Notation?}. 
In this case, $f(x)=x\log x$ corresponds to the relative entropy, $D(\rho\|\gamma)\equiv D_{x\log x}^\mathrm{std}(\rho\|\gamma)= \Tr\rho (\log\rho-\log\gamma)$. 

%Notice how both of these definitions reduce to $D_f^\mathrm{cl}$ over commuting states. \paula{i would remove this sentence i think}

The maximal $f$-divergence has some significance as the largest-valued of all monotone quantum $f$-divergences $D_f$. i.e. for any $f\in\mathcal{F}$ and $\rho,\gamma\in\mathcal{D}(\mathcal{H})$ with equal support \cite{HM17,Matsumoto15}:
\begin{equation}\label{eqn:maxmax}
    D_f(\rho\|\gamma)\leq\widehat D_f(\rho\|\gamma)
\end{equation} %\paula{What is an operational meaning? usually it means that it "magically appears" in some task, like the success probability of state discrimination, i believe}\textcolor{purple}{this is a phrasing taken from Hiai-Mosonyi, 2017. I can refer to it I suppose.}\paula{What do they mean by "an operational meaning"? What is the operation?} \paula{a largest-valued functional?}\textcolor{purple}{is this really not clear?}\paula{i think it's more that $D_f$ alone without rho, sigma is not really said before what it is}\textcolor{purple}{eh. you can always define a function $f$(x) and then refer to $f$. It's fairly consistent. besides, I do elaborate in literally the next sentence.} 

Finally, quantum $\chi^2$-divergences are monotone quantum $f$-divergences that reduce to the classical $D_{(x-1)^2}^\mathrm{cl}$, and are parametrised by a function $\kappa:(0,\infty)\to(0,\infty)$ belonging to a subset $\mathcal{K}$ of operator monotone functions (see later); they have the following expression:
%\paula{we have to decide whether or not to rename them...?}\textcolor{purple}{true! problem is I want to keep X scalable... it's only useful when comparing to norm-based coefficients I suppose.} 

\begin{equation*}
    \chi_\kappa^2(\rho\|\gamma):=\|\rho-\gamma\|_{\kappa,\rho}^2:=\langle\rho-\gamma,R_\rho^{-1}\kappa(\Delta_{\rho,\rho})(\rho-\gamma)\rangle_\mathrm{HS}
\end{equation*}
where, for $X\in T_{\rho}\mathcal{D}(\mathcal{H})$, we define the Riemannian semi-norm: 
\begin{equation}\label{eqn:riemdiv}
    \|X\|_{\kappa,\rho}^2:=\langle X,R_\rho^{-1}\kappa(\Delta_{\rho,\rho})(X)\rangle_\mathrm{HS}.
\end{equation}
We study standard $f$-divergences and Riemannian semi-norms together because they satisfy the following relationship for $\rho\in\mathcal{D}(\mathcal{H}),X\in T_{\rho}\mathcal{D}(\mathcal{H})$ with equal support and $\gamma_\varepsilon:=\rho+\epsilon X$:
\begin{equation}\label{eqn:stddivlocallyRiem}
    \left.\frac{d^2}{d\varepsilon^2}\right|_{\varepsilon=0} {D_f^\mathrm{std}( \rho \| \gamma_\varepsilon)}\equiv \lim_{\epsilon\to 0^+}\frac{2}{\epsilon^2}D_f^\mathrm{std}( \rho \| \gamma_\varepsilon)=f''(1)\|X\|_{\kappa_f,\rho}^2
\end{equation}
where $f(x)\in\mathcal{F},\kappa_f(x):=\frac{f(x)+\widetilde f(x)}{f''(1)(x-1)^2}$, $\widetilde f(x):=xf(x^{-1})$. That is, \\ \\
\emph{The second-order local behaviour of a standard $f$-divergence is determined by an associated Riemannian semi-norm.}

This variation in the second order behaviour makes the standard $f$-divergences convenient to study the significance of the relationship between the relative entropy and its local behaviour, $\|X\|_{\kappa_\mathrm{BKM},\rho}^2$, $\kappa_\mathrm{BKM}:=\kappa_{x\log x}$, which is quite special. This is in contrast to the maximal $f$-divergences and quantum Hockey-stick divergences, which are families with only local  $\|X\|_{\kappa_{(y-1)^2}, \rho}^2,\|X\|_{\kappa_\mathrm{BKM},\rho}^2$ behaviour respectively \cite{PR98,HT24}. 

\subsection{Operator Convex Functions}
The standard $f$-divergences and Riemannian semi-norms, which are the focus of this paper, are parametrised by functions belonging to the following classes:
\begin{align*}
    &\mathcal{F} := \{f:(0,\infty)\to \mathbb{R}, \text{ operator convex, } f(1)=0,f''(1)>0\},\\
    &\mathcal{F}_\mathrm{sym} := \{f:(0,\infty)\to \mathbb{R}, \text{ operator convex, } xf(x^{-1})=f(x)\text{ for }x>0, f(1)=0,f''(1)=2\},\\
    &\mathcal{K}:=\{\kappa:(0,\infty)\to \mathbb{R}, \text{ operator convex, } x\kappa(x)=\kappa(x^{-1})\text{ 
    for }x>0,\kappa(1)=1\}.
\end{align*}
Note that the \textit{transpose} $\widetilde f(x):=xf(x^{-1})\in\mathcal{F}$ of $f(x)\in\mathcal{F}$ satisfies for all $\rho,\gamma\in\mathcal{D}(\mathcal{H})$:
\begin{equation*}
    D_{\widetilde{f}}^\mathrm{cl}(\rho\|\gamma)\equiv D_f^\mathrm{cl}(\gamma\|\rho),\quad D_{\widetilde{f}}^\mathrm{std}(\rho\|\gamma)\equiv D_f^\mathrm{std}(\gamma\|\rho),\quad \widehat D_{\widetilde{f}}(\rho\|\gamma)\equiv \widehat D_f(\gamma\|\rho)
\end{equation*}
which is the reason for considering $\mathcal{F}_\mathrm{sym}$ - this subclass gives symmetric quantum $f$-divergences.

An \textit{operator convex function} is a convex, real analytic function $f:(0,\infty)\to\mathbb{R}$ defined by the condition:
\begin{equation*}
    f(\lambda A+(1-\lambda)B)\leq\lambda f(A)+(1-\lambda)f(B)\;\forall\lambda\in(0,1),\forall A,B\in\mathbb{P}_d,\forall d\in\mathbb{N}
\end{equation*}
The functions in $\mathcal{K}$ are also \textit{operator monotone decreasing} \cite{HKPR13}:

An \textit{operator monotone function} is a positive increasing function $f:(0,\infty)\to(0,\infty)$ defined by the condition:
\begin{equation*}
    f(A)\leq f(B)\;\forall A,B\in\mathbb{P}_d,A\preccurlyeq B,\forall d\in\mathbb{N}
\end{equation*}
An \textit{operator monotone decreasing} function $\kappa$ is a positive decreasing function $\kappa(x)=1/f(x)$ where $f$ is an operator monotone function.

In particular, the operators $A,B$ may be superoperators and all of the above definitions still apply. 

Any $f \in \mc F$ (note: $\mc F_\mathrm{sym}\subset \mc F$), $\kappa\in\mathcal{K}$ can be written in the following integral representations \cite{HKPR13,HR15}:
\begin{enumerate}
    \item If $f\in\mathcal{F}$, there exists a unique constant $c\geq0$ and a unique positive measure $\mu$ on $[0,\infty)$ with $\int_{[0,\infty)} \frac{1}{1+s} d\mu (s) < \infty$ such that
    \begin{equation}\label{opcon}
    f(x) = f'(1)(x-1) + c(x-1)^2 + \mu(0)\frac{(x-1)^2}{x} + \int_0^\infty \frac{(x-1)^2}{x+s} d\mu(s), 
    \end{equation}
    \item If $\kappa\in\mathcal{K}$, there exists a unique probability measure $m$ on [0,1] such that \begin{equation}\label{opmon}
    \kappa(x) = \int_{[0,1]} \kappa_s(x) \, dm(s), \quad x \in (0,\infty). 
    \end{equation}
    where $\kappa_s(x):=\frac{1 + s}{2}\left( \frac{1}{x + s} + \frac{1}{sx + 1} \right)$ is decreasing in $s\in[0,1]$ for all $x\in(0,\infty)$.
\end{enumerate}

We will be especially interested in $\kappa_\mathrm{max}:=\kappa_0,\kappa_\mathrm{min}:=\kappa_1\in\mathcal{K}$, which satisfy $\kappa_\mathrm{min}(x)\leq\kappa(x)\leq\kappa_\mathrm{max}(x)$ for all $\kappa\in\mathcal{K},x\in(0,\infty)$. Since the definitions of the standard $f$-divergence, maximal $f$-divergence and Riemannian semi-norms are linear in $f\in\mathcal{F},\kappa\in\mathcal{K}$ respectively, these integral representations can be used to decompose the corresponding quantum $f$-divergences. This is a major tool in our disposal, that greatly simplifies the problem of dealing with these quantum $f$-divergences in generality, and this principle is at the core of many results in this paper.

% Then via functional calculus,  
% \begin{equation}\label{eqn:f-divergence decomposition}
%     D^{\std}_f(\rho \|\sigma) = b \langle \rho - \sigma, \sigma^{-1}(\rho-\sigma) \rangle + c \langle \rho - \sigma, \rho^{-1}(\rho-\sigma) \rangle + \int_0^\infty \langle \rho - \sigma, \frac{1}{L_\rho + s R_\sigma}(\rho-\sigma) \rangle d\mu(s), \quad \rho,\sigma \in \mc D_d.
% \end{equation}

\subsection{Relative Expansion Coefficients}
In this work, we study the following quantities, for a pair of quantum channels $\mathcal{N}:\mathcal{B}(\mathcal{H}_A)\to\mathcal{B}(\mathcal{H}_B),\mathcal{M}:\mathcal{B}(\mathcal{H}_A)\to\mathcal{B}(\mathcal{H}_B')$, and $f\in\mathcal{F},\kappa\in\mathcal{K}$:
\begin{align}
    % & \eta^{\std}_f(\mc N,\mc M):= \sup_{\substack{\rho\neq \gamma \in \mc D(\mathcal{H}_A),\\ \supp(\rho) = \supp (\gamma)}} \frac{D^{\std}_f(\mc N(\rho) \|\mc N(\gamma))}{D^{\std}_f(\mc M(\rho) \|\mc M(\gamma))}, \\
    & \widecheck{\eta}^{\std}_f(\mc N,\mc M):= \inf_{\substack{\rho\neq \gamma \in \mc D(\mathcal{H}_A),\\\supp(\rho) = \supp (\gamma)}} \frac{D^{\std}_f(\mc N(\rho) \|\mc N(\gamma))}{D^{\std}_f(\mc M(\rho) \|\mc M(\gamma))}, \\
    % & \eta^{\Riem}_{\kappa}(\mc N,\mc M):= \sup_{\rho \in \mc D(\mathcal{H}_A)} \sup_{X \in T_\rho \mc D(\mathcal{H}_A)} \frac{\|\mathcal{N}(X)\|_{\kappa,\:\mathcal{N}(\rho)}^2}{\|\mathcal{M}(X)\|_{\kappa,\:\mathcal{M}(\rho)}^2}, \\
    & \widecheck{\eta}^{\Riem}_{\kappa}(\mc N,\mc M):= \inf_{\rho \in \mc D(\mathcal{H}_A)} \inf_{X \in T_\rho \mc D(\mathcal{H}_A)} \frac{\|\mathcal{N}(X)\|_{\kappa,\:\mathcal{N}(\rho)}^2}{\|\mathcal{M}(X)\|_{\kappa,\:\mathcal{M}(\rho)}^2}. 
\end{align}
The above quantities can be similarly defined for other quantum $f$-divergences. We call $\widecheck{\eta}^{\std}_f(\mathcal{N},\mathcal{M})$ and $\widecheck{\eta}^{\Riem}_{\kappa}(\mathcal{N},\mathcal{M})$ the divergence and Riemannian \textit{relative expansion coefficients}, respectively. We will often be interested in cases where $\mathcal{N}=\mathcal{D}\circ\mathcal{M}$, which we call (plain) contraction and expansion coefficients over restricted subspaces:
\begin{align*}
    &\eta^\mathrm{std}(\mathcal{D};\mathrm{Im}\:\mathcal{M}):=\widecheck\eta^\mathrm{std}(\mathcal{M},\mathcal{D}\circ\mathcal{M})^{-1},\eta^\mathrm{Riem}(\mathcal{D};\mathrm{Im}\:\mathcal{M}):=\widecheck\eta^\mathrm{Riem}(\mathcal{M},\mathcal{D}\circ\mathcal{M})^{-1}\\
    &\widecheck\eta^\mathrm{std}(\mathcal{D};\mathrm{Im}\:\mathcal{M}):=\widecheck\eta^\mathrm{std}(\mathcal{D}\circ\mathcal{M},\mathcal{M}),\widecheck\eta^\mathrm{Riem}(\mathcal{D};\mathrm{Im}\:\mathcal{M}):=\widecheck\eta^\mathrm{Riem}(\mathcal{D}\circ\mathcal{M},\mathcal{M})
\end{align*}
These contraction and expansion coefficients are all contained within $[0,1]$ by the data processing inequality; they are respectively understood via the best-case or worst-case proportion of distinguishability %\paula{understood via? They "can be understood as" ?}\textcolor{purple}{via... because technically contraction coefficients tell you the most information preserved} 
preserved by the action of a quantum channel $\mathcal{D}$ on all states in $\mathrm{Im}\:\mathcal{M}$. Thus, the relative expansion coefficients are used to extend this conversation from the cases of perfect recoverability to a looser discussion about whether or not a similar proportion of distinguishability information is preserved across different distinguishability measures, and to provide explicit examples of reverse data processing inequalities. 
%\paula{what do we mean by reverse DPI? We have not been very explicit about it}
%\paula{From previous paper: this can be interpreted as a reverse data processing inequality for $\mc N$ (in the spirit of reverse Doeblin coefficients \cite{H2024}, reverse Pinsker inequality \cite{Sason15,SV16}, and reverse log-Sobolev inequality \cite{Sontz99}).}

In the case when the second channel $\mc M$ is the identity channel, we obtain the contraction and expansion coefficients \emph{over all states},
\begin{equation*}
    \eta^{\std}_f(\mc N):= \widecheck\eta^{\std}_f(id,\mc N)^{-1},{\eta}^{\Riem}_{\kappa}(\mc N):=\widecheck{\eta}^{\Riem}_{\kappa}(id,\mc N)^{-1}, \quad \widecheck \eta^{\std}_f(\mc N):= \widecheck\eta^\mathrm{std}_f(\mc N,id),\widecheck{\eta}^{\Riem}_{\kappa}(\mc N):=  \widecheck{\eta}^{\Riem}_{\kappa}(\mc N,id).
\end{equation*}
These above contraction coefficients have been the focus in many prior works, e.g. because they can be used to give upper bounds on the convergence rate of a discrete quantum Markov chain \cite{Choi1994, Cohen1993, LR99, GZB24, GT25, HR15, HT24, Temme_2010}.\\

This work establishes connections between different relative expansion coefficients, for example, for all pairs of channels $\mathcal{N},\mathcal{M}$, for all functions $f\in\mathcal{F}$:
\begin{equation}\label{eqn:stddivgeneral}
    \widecheck\eta_f^\mathrm{std}(\mathcal{N},\mathcal{M})=\widecheck\eta_{\widetilde{f}}^\mathrm{std}(\mathcal{N},\mathcal{M})\leq\widecheck\eta_{f_\mathrm{sym}}^\mathrm{std}(\mathcal{N},\mathcal{M})
\end{equation}
where $f_\mathrm{sym}(x):=\frac{f(x)+\widetilde{f}(x)}{f''(1)}\in\mathcal{F}_\mathrm{sym}$
The inequality is comes from the fact that for all states $\rho,\gamma$ with equal support and channels $\mathcal{N},\mathcal{M}$ \cite{HR15}:
\begin{equation*}
    \frac{D_f^\mathrm{std}(\mathcal{N}(\rho)\|\mathcal{N}(\gamma))+D_f^\mathrm{std}(\mathcal{N}(\gamma)\|\mathcal{N}(\rho))}{D_f^\mathrm{std}(\mathcal{M}(\rho)\|\mathcal{M}(\gamma))+D_f^\mathrm{std}(\mathcal{M}(\gamma)\|\mathcal{M}(\rho))}\geq \min\left\{{\frac{D_f^\mathrm{std}(\mathcal{N}(\rho)\|\mathcal{N}(\gamma))}{D_f^\mathrm{std}(\mathcal{M}(\rho)\|\mathcal{M}(\gamma))},\frac{D_f^\mathrm{std}(\mathcal{N}(\gamma)\|\mathcal{N}(\rho))}{D_f^\mathrm{std}(\mathcal{M}(\gamma)\|\mathcal{M}(\rho))}}\right\}\geq \widecheck\eta_f^\mathrm{std}(\mathcal{N},\mathcal{M})
\end{equation*}

% A complete answer to the following question remains open: 

% \begin{center}
%     \emph{
%     Under which conditions on $f\in \mc F$, do we have 
%     $\eta^{\std}_f(\mc N,\mc M) = \eta^{\Riem}_{\kappa_f}(\mc N,\mc M)$ for any pair of channels $\mc N, \mc M$?
% }
% \end{center}

% Note that in the classical setting, we always have equality \cite{Lesniewski_1999}. However, in the quantum setting, we know that the equality may not hold \cite[Section 6]{Hiai_2015}. This intricate phenonenon motivates us to study when or why the equality may or may not hold. 

% It is known that equality holds for $f(x) = x\log x$ and $f(x) = (x-1)^2$ as well as for their symmetrised variants (see Section~\ref{sec:equivalence}).
% Whether the equality extends to other $f$ remains an open problem.

%\include{prelim_old}

\section{No Reverse Data Processing Inequality over All States}\label{sec:nordpi}

We begin our investigation into relative expansion coefficients with a rather surprising result. It turns out that for \emph{all quantum $f$-divergences satisfying monotonicity}, no proportion of distinguishability is necessarily preserved by a quantum channel over all pairs of states, even if the quantum channel is injective. In other words, the expansion coefficient for an unrestricted domain is null is many cases, $\widecheck\eta_f(\mathcal{N})=0$. 

This behaviour of expansion coefficients is in stark contrast to an expansion coefficient based on e.g. the trace distance, or any choice of norms on the output spaces of $\mathcal{N},\mathcal{M}$. Such norm-based relative expansion coefficients $\widecheck\eta(\mathcal{N},\mathcal{M})$ are strictly positive for all injective quantum channels $\mathcal{N}$, regardless of $\mathcal{M}$. We obtain this non-trivial result that $\widecheck\eta_f(\mathcal{N})=0$, by exploiting the cases where $D_f$ is first and second order in e.g. the trace distance. The trick is to find a family of pairs of states $\rho_\varepsilon,\gamma_\varepsilon \in \mathcal{D}(\mathcal{H}_A)$, continuously parametrised by $\varepsilon \in \mathbb{R^+}$, $\varepsilon <<1$, such that $D_f(\rho_\varepsilon\|\gamma_\varepsilon) = \Theta(\varepsilon) $ while $D(\mathcal{N}(\rho_\varepsilon) || \mathcal{N}(\gamma_\varepsilon)) = o(\varepsilon)$, as $\varepsilon\to0$, in the non-trivial cases where $\mathcal{N}$ is neither a unitary channel nor a replacer channel. To do this, we use the following two lemmas:

\begin{lemma}[{\cite[Theorem 3.1]{davies1976quantum}}]\label{lem1}
If a quantum channel $\mathcal{N}: \mathcal{B}(\mathcal{H}_A) \to \mathcal{B}(\mathcal{H}_B)$ is purity-preserving, i.e. it maps any pure state into a pure state, then $\mathcal{N}$ must be either an isometric embedding $\mathcal{N}(\rho) = V \rho V^\dagger$, $V^\dagger V = \mathbbm{1}_A$, or a replacer channel $\mathcal{N}(\rho) = \mathrm{Tr}\:\rho \, \ket{\varphi}\bra{\varphi}$ for some state $\ket{\varphi}$.
\end{lemma}
    
\begin{lemma}[{\cite[Theorem 3.1]{BGSW24}}]\label{lem2}
    If a quantum channel $\mathcal{N}: \mathcal{B}(\mathcal{H}_A) \to \mathcal{B}(\mathcal{H}_B)$, $\dim \mathcal{H}_B =: d_B \le \dim \mathcal{H}_A =: d_A$, is not purity-preserving, then there exists an orthogonal projection $P_A \in \mathcal{B}(\mathcal{H}_A)$, $\mathrm{rk}\,P_A \leq d_A - 1$, and a pure state $|\psi \rangle$ such that
    \[
    \mathrm{supp} \, \mathcal{N}(\ket{\psi}\bra{\psi}) \leq \mathrm{supp} \, \mathcal{N}(P_A).
    \]
\end{lemma}
% \begin{proof}
%     See Appendix~\ref{sec:proofoflemma}.
% \end{proof} 

\noindent The following theorem tells us that, provided any quantum $f$-divergence $D_f$, $\widecheck\eta_f(\mathcal{N})=0$ for any non-unitary quantum channel $\mathcal{N}:\mathcal{B}(\mathcal{H_A})\to\mathcal{B}(\mathcal{H_B})$ whose output dimension $d_B^2$ is no more than the input dimension $d_A^2$. Lemma~\ref{lem2} requires that $d_B\leq d_A$ for its pigeonhole principle-based proof. We are extending this result from the case of the relative entropy, which has been established in \cite{BGSW24}, because this work explores connections with the other quantum $f$-divergences (particularly, standard $f$-divergences and Riemannian semi-norms). In fact, this problem reduces to considering the maximal $f$-divergences $\widehat D_f$, which is a key observation that allows us to completely generalise the result from \cite{BGSW24}. 

\begin{theorem}[No Divergence-based Reverse DPI on All States]\label{nordpi}
    $ $\newline
    Suppose we are given a quantum channel $\mathcal{N}: \mathcal{B}(\mathcal{H}_A) \to \mathcal{B}(\mathcal{H}_B)$, with $d_B\leq d_A$. Any monotone quantum $f$-divergence $D_f$, for operator convex $f\in\mathcal{F}$, satisfies:
    % \begin{enumerate}
    %     \item $f:(0, \infty) \to \mathbb{R}$, $f(1) = 0,f''(1)>$, is twice differentiable and strictly convex at 1 ($f''(1) > 0$).
    %     \item $f$ is convex (hence continuous) over $(0,1)$ with $f(0^+)
    %     \in (-\infty,\infty]$.
    %     % \item $\frac{d}{d\varepsilon}\big|_{\varepsilon = 0} D_f(\rho \| \rho + \varepsilon X) = 0$ for all $\rho \in \mathcal{D}(\mathcal{H}_B), X \in T_\rho \mathcal{D}(\mathcal{H}_B)$.
    %     \item $D_f(\rho \| \gamma)$ is differentiable at $\rho=\gamma$ for small perturbations $(X,Y),X,Y\in T_\rho \mathcal{D}(\mathcal{H}_B)$
    % \end{enumerate}
    % \noindent Then,
    \begin{align*}
        \widecheck{\eta}_f(\mathcal{N})
        &= \mathbbm{1}\{\mathcal{N} \text{ is a unitary channel}\} \\
        &=
        \begin{cases}
            1 & \text{if } d := d_A = d_B \text{ and }
            \mathcal{N}(\rho) = U\rho U^\dagger,\, U \in U(d) \\
            0 & \text{otherwise}
        \end{cases}
    \end{align*}
\end{theorem}

\begin{proof}
    $ $\newline
    \textit{Case 1} ($\mathcal{N}$ is purity-preserving):\\
    If $\mathcal{N}(\rho)=U\rho U^\dagger$ for some unitary $U\in U(d)$, $\widecheck{\eta}_f(\mathcal{N}) = 1$ because $\mathcal{N}$ is invertible. Otherwise, if $\mathcal{N}(\rho) = \mathrm{Tr}\:\rho \, \ket{\varphi}\bra{\varphi}$ for some pure state $\ket{\varphi}\bra{\varphi} \in \mathcal{D}(\mathcal{H}_B)$, $\widecheck{\eta}_f(\mathcal{N}) = 0$, because $\mathcal{N}$ is not injective. See \cite[Theorem 3.1]{BGSW24}.
    % \textit{Case 1} ($\mathcal{N}(\rho) = U \rho U^\dagger$ for some unitary $U$):
    
    % \noindent For any $\rho \neq \gamma \in \mathcal{D}(\mathcal{H}_A)$, $D_f(\mathcal{N}(\rho)||\mathcal{N}(\gamma)) = D_f(\rho||\gamma) > 0$ by unitary invariance of $D_f$.
    
    % \noindent As a result, $\widecheck{\eta}_f(\mathcal{N}) = 1$.
    
    % \noindent \textit{Case 2} ($\mathcal{N}(\rho) = \mathrm{Tr}\:\rho \, \ket{\varphi}\bra{\varphi}$ for some pure state $\ket{\varphi}\bra{\varphi} \in \mathcal{D}(\mathcal{H}_B)$):
    
    % \noindent For any $\rho \neq \gamma \in \mathcal{D}(\mathcal{H}_A)$, $D_f(\mathcal{N}(\rho)||\mathcal{N}(\gamma)) = D_f(\ket{\varphi}\bra{\varphi}\: \|\: \ket{\varphi}\bra{\varphi})= 0$ and $D_f(\rho||\gamma) > 0$.
    
    % \noindent As a result, $\widecheck{\eta}_f(\mathcal{N}) = 0$.
    
    \noindent \textit{Case 2} ($\mathcal{N}$ is not purity-preserving):\\ 
    \noindent Following the strategy mentioned earlier, we slightly modify the construction in the proof of \cite[Theorem 3.1]{BGSW24} and use Lemma~\ref{lem2} to choose two states $\rho_\varepsilon,\gamma_\varepsilon$ of equal support (to address quantum $f$-divergences generally): \[
    \bar{\rho}:= \frac{P_A}{\mathrm{rk} P_A},\quad\rho_\varepsilon := (1-o(\varepsilon))\bar{\rho} + o(\varepsilon) \ket{\psi}\bra{\psi},\quad\gamma_\varepsilon := (1-\varepsilon)\bar{\rho} + \varepsilon \ket{\psi}\bra{\psi} \quad\in \mathcal{D}(\mathcal{H}_A),
    \]
    for some orthogonal projection $P_A \in \mathcal{B}(\mathcal{H}_A)$ and pure state $\ket{\psi}\bra{\psi} \in \mathcal{D}(\mathcal{H}_A)$ satisfying $\ket{\psi}\bra{\psi}\perp \bar{\rho}$ but $\mathrm{supp}\; \mathcal{N}(\ket{\psi}\bra{\psi}) \leq \mathrm{supp} \;\mathcal{N}(\bar{\rho})$. Here, $\varepsilon \in (0,1)$ is small. \\

    % \noindent WLOG we take $f(0^+)<\infty$, as otherwise $D_f(\rho\|\gamma_\varepsilon)=\infty,\:D_f(\mathcal{N}(\rho)\|\mathcal{N}(\gamma_\varepsilon))<\infty$ and this automatically gives the result.
    
    \noindent Via the common eigendecomposition, let $p_\varepsilon = (\frac{1 - o(\varepsilon)}{\mathrm{rk}\, P_A},..., \frac{1 - o(\varepsilon)}{\mathrm{rk}\, P_A}, o(\varepsilon))$, $q_\varepsilon = (\frac{1 - \varepsilon}{\mathrm{rk}\, P_A},..., \frac{1 - \varepsilon}{\mathrm{rk}\, P_A}, \varepsilon)$, where $\mathrm{rk}P_A$ is the rank of the projection. Then:
    \begin{align*}
        D_{f}(\rho_\varepsilon \| \gamma_\varepsilon)&= D_{f}^\mathrm{cl} (p_\varepsilon \| q_\varepsilon)= \frac{1 - \varepsilon}{\mathrm{rk}\, P_A} \, 
        f \left( \frac{\frac{1-o(\varepsilon)}{\mathrm{rk}\, P_A}}{\frac{1-\varepsilon}{\mathrm{rk}\, P_A}} \right) \cdot \mathrm{rk}P_A 
        + \varepsilon \, f \left( \frac{o(\varepsilon)}{\varepsilon} \right) \\
        &= (1-\varepsilon) f \left( \frac{1-o(\varepsilon)}{1-\varepsilon} \right) + \varepsilon f(o(1))
    \end{align*}
    
    \noindent Observe that by the strict convexity of $f$ at 1:\\
    \begin{align*}
    \lim_{\varepsilon \to 0^+} \frac{1}{\varepsilon} D_f(\rho_{\varepsilon} \| \gamma_\varepsilon) 
    &= \lim_{\varepsilon\to0^+} \frac{1 - \varepsilon}{\varepsilon} f\left( \frac{1-o(\varepsilon)}{1 - \varepsilon} \right) + \lim_{\varepsilon\to0^+} f(o(1))  \\
    &= f'(1)+f(0^+) = f'(1) -\frac{f(1) - f(0^+)}{1 - 0} > 0, \quad 
    \end{align*}
    
    \noindent Thus, we have demonstrated the required first order behaviour. Now, to deduce the second order behaviour of $D_f(\mathcal{N}(\mathcal{\rho_\varepsilon})\|\mathcal{N}(\gamma_\varepsilon))$, we consider the maximal $f$-divergence. 
    
    \noindent By assumption, $D_f$ is a monotone quantum $f$-divergence, which means by the maximality of the maximal $f$-divergence \eqref{eqn:maxmax} that \begin{equation*}
        D_f(\mathcal{N}(\mathcal{\rho_\varepsilon})\|\mathcal{N}(\gamma_\varepsilon))\leq \widehat D_f(\mathcal{N}(\mathcal{\rho_\varepsilon})\|\mathcal{N}(\gamma_\varepsilon))
    \end{equation*}
    
    % $\left.\frac{d}{d\varepsilon}\right|_{\varepsilon=0} {D_f(\mathcal{N} (\rho_\varepsilon) \| \mathcal{N} (\gamma_\varepsilon))}$ exists 
    
    \noindent Since $\mathcal{N}(\ket{\psi}\bra{\psi})\in T_{\mathcal{N}(\rho)}\mathcal{D}(\mathcal{H}_B),\:\mathrm{supp}\:\mathcal{N}(\mathcal{\rho_\varepsilon})=\mathrm{supp}\:\mathcal{N}(\gamma_\varepsilon)$, and we can write the following expression \eqref{eqn:maxfdiv}:
    \begin{align*}
        \widehat D_f(\mathcal{N}(\mathcal{\rho_\varepsilon})\|\mathcal{N}(\gamma_\varepsilon)) &= \Tr\mathcal{N}(\gamma_\varepsilon)^{1/2}f(\Tr\mathcal{N}(\gamma_\varepsilon)^{-1/2}\Tr\mathcal{N}(\rho_\varepsilon)\Tr\mathcal{N}(\gamma_\varepsilon)^{-1/2})\Tr\mathcal{N}(\gamma_\varepsilon)^{1/2}
    \end{align*}
    \noindent Observing that $\mathcal{N}(\rho_\varepsilon) = \mathcal{N}(\gamma_\varepsilon)+\varepsilon\mathcal{N}(\ket{\psi}\bra{\psi}-\bar{\rho})+o(\varepsilon)$ and using the Taylor expansion of $f(x)$ about $x=1$:
    \begin{align*}
        \widehat D_f(\mathcal{N}(\mathcal{\rho_\varepsilon})\|\mathcal{N}(\gamma_\varepsilon)) &= \Tr\mathcal{N}(\gamma_\varepsilon)^{1/2}f(\mathcal{N}(\gamma_\varepsilon)^{-1/2}\mathcal{N}(\rho_\varepsilon)\mathcal{N}(\gamma_\varepsilon)^{-1/2})\mathcal{N}(\gamma_\varepsilon)^{1/2}\\
        &=f'(1)\Tr\mathcal{N}(\gamma_\varepsilon)^{1/2}(\mathcal{N}(\gamma_\varepsilon)^{-1/2}\mathcal{N}(\rho_\varepsilon)\mathcal{N}(\gamma_\varepsilon)^{-1/2}-I)\mathcal{N}(\gamma_\varepsilon)^{1/2}\\
        &\quad+\frac{f''(1)}{2}\Tr\mathcal{N}(\gamma_\varepsilon)^{1/2}(\mathcal{N}(\gamma_\varepsilon)^{-1/2}\mathcal{N}(\rho_\varepsilon)\mathcal{N}(\gamma_\varepsilon)^{-1/2}-I)^2\mathcal{N}(\gamma_\varepsilon)^{1/2}+o(\varepsilon^2)\\
        &=\frac{f''(1)}{2}\Tr\mathcal{N}(\gamma_\varepsilon)^{1/2}(\mathcal{N}(\gamma_\varepsilon)^{-1/2}\mathcal{N}(\rho_\varepsilon)\mathcal{N}(\gamma_\varepsilon)^{-1/2}-I)^2\mathcal{N}(\gamma_\varepsilon)^{1/2}+o(\varepsilon^2)\\
        &=\frac{\varepsilon^2}{2}f''(1)\Tr\mathcal{N}(\gamma_\varepsilon)^{1/2}(\mathcal{N}(\gamma_\varepsilon)^{-1/2}\mathcal{N}(\ket{\psi}\bra{\psi}-\bar{\rho})\mathcal{N}(\gamma_\varepsilon)^{-1/2})^2\mathcal{N}(\gamma_\varepsilon)^{1/2}+o(\varepsilon^2)\\
        &=o(\varepsilon)
    \end{align*}
    The final equality certainly holds, because $\mathcal{N}(\ket{\psi}\bra{\psi}-\bar{\rho})\in T_{\mathcal{N}(\gamma_\varepsilon)}\mathcal{D}(\mathcal{H}_B)$. \\
    \noindent Therefore, $D_f(\mathcal{N}(\mathcal{\rho_\varepsilon})\|\mathcal{N}(\gamma_\varepsilon))\leq \widehat D_f(\mathcal{N}(\mathcal{\rho_\varepsilon})\|\mathcal{N}(\gamma_\varepsilon))=o(\varepsilon).$
    % . Further, by the minimality of $D_f(\mathcal{N} (\rho_\varepsilon) \| \mathcal{N} (\gamma_\varepsilon))\geq0$ at $\varepsilon=0$, we necessarily have:\\
    % \begin{equation*}
    %     \left.\frac{d}{d\varepsilon}\right|_{\varepsilon=0} {D_f(\mathcal{N} (\rho_\varepsilon) \| \mathcal{N} (\gamma_\varepsilon))}=0
    % \end{equation*}
    
    \noindent As a result,
    \begin{align*}
        0 \leq \widecheck{\eta}_f(\mathcal{N})
        &:= \inf_{\substack{\rho \neq \gamma \in \mathcal{D}(\mathcal{H}_A), \\ \mathrm{supp}\: \rho = \mathrm{supp}\: \gamma}} 
        \frac{D_f \big(\mathcal{N}(\rho) \| \mathcal{N}(\gamma)\big)}{D_f(\rho \| \gamma)} \leq \lim_{\varepsilon \to 0^+} 
        \frac{D_f \big(\mathcal{N}(\rho_\varepsilon) \| \mathcal{N}(\gamma_\varepsilon)\big)}{D_f(\rho_\varepsilon \| \gamma_\varepsilon)} \\
        &= \lim_{\varepsilon \to 0^+} 
        \frac{ \frac{D_f \big(\mathcal{N}(\rho_\varepsilon) \| \mathcal{N}(\gamma_\varepsilon)\big)}{\varepsilon} }{ \frac{D_f(\rho_\varepsilon \| \gamma_\varepsilon)}{\varepsilon} } = 0
    \end{align*}
    
    \noindent Thus, $\widecheck{\eta}_f(\mathcal{N}) = 0$.
\end{proof}

This is a statement that, for any monotone quantum $f$-divergence, there is no reverse data-processing inequality for these quantum channels over all pairs of quantum states. If $d_B>d_A$, then this is not necessarily true, because an erasure channel $\mathcal{N}_\nu:\rho\mapsto (1-\nu)\rho+\nu\ket{e}\bra{e}$ has $\widecheck\eta_f(\mathcal{N})=1-\nu$ for any quantum $f$-divergence $D_f$ satisfying $D_f(A_1+A_2\|B_1+B_2)=D_f(A_1\|B_1)+D_f(A_2\|B_2)$ for positive semi-definite operators $A_1,B_1\perp A_2,B_2$; this includes all of the quantum $f$-divergences that were listed in Section~\ref{sec:introfdivs}. Further, if we similarly define the relative expansion for the Petz-Rényi and sandwiched Rényi divergences, we can apply precisely the same arguments to deduce that there is no reverse data processing inequality and that $\widecheck\eta(\mathcal{N})=1-\nu$ for both of these cases. Note that for the no reverse data-processing inequality results in these cases, the only modification to the proof is that it is not necessary to consider the maximal $f$-divergence since they can be seen directly to have that second order behaviour in $\epsilon$. 

% Our discussion and Theorem~\ref{nordpi} applies to a vast variety of previously considered quantum f-divergences \textcolor{purple}{(my generic argument involves upper bounding by the maximal f-divergence, which should be analytic)}:

% \begin{corollary}\label{cor1}
% Quantum minimal and maximal $f$-divergences, sandwiched Rényi divergences, and all standard $f$-divergences and Riemannian semi-norms have a trivial non-relative expansion coefficient over all states.\cite{HM17,hiai2021quantum,LR99,HR15}\paula{references to where they where considered}
% \end{corollary}

% \noindent We interpret these theorems as telling us that non-unitary channels whose input space dimension is no larger than the output space dimension do not uniformly preserve information about the distinguishability of all states, w.r.t. the quantum $f$-divergences and Riemannian semi-norms. Especially when one compares this behavior with expansion coefficients based on norms attributed to $\mathcal{B}(\mathcal{H}_A)$ and $\mathcal{B}(\mathcal{H}_B)$ (see Section~\ref{sec:norms}), it is interesting that we observe this for injective quantum channels (that map states uniquely into other states). 

While this lack of a reverse data processing inequality is remarkably general, because we are motivated by the connection between the relative entropy and its local behaviour, we will henceforth specialise to considering only standard $f$-divergences and Riemannian semi-norms. Riemannian relative expansion coefficients are particularly convenient to study, owing to their computability in the qubit setting (see Section~\ref{sec:explicitcoeffs}), so comparing the divergence and Riemannian relative expansion coefficients is fruitful. 

Theorem~\ref{nordpi} only addresses the preserved proportion of distinguishability when accounting for \textit{all} pairs of input quantum states. Naturally, one may wonder how restricting the domain of the quantum channel $\mathcal{N}$ can alter the expansion coefficient. Indeed, generalising and extending the computations of \cite{BGSW24}, we can find families of qubit channels $\mathcal{N}$ that have a positive expansion coefficient when their domain is restricted to the image of another channel $\mathcal{D}$ in the same family, i.e. $\widecheck\eta_f(\mathcal{N};\mathrm{Im}\:\mathcal{D})>0$. We go even further, and identify that the large class of primitive quantum channels (defined by a unique, full-rank fixed point) have a positive expansion coefficient over a restricted domain after some $n^\mathrm{th}$ iteration of the same channel applied consecutively. Let $\mc N^{n}=\mc N \circ \mc N^{n-1}$ denote the $n$-fold concatenation of a channel $\mc N$, then, for $n$ sufficiently large, independently of $f\in\mathcal{F}$ (see Theorem~\ref{thm:primcoeffpos}),
% \[
% \widecheck{\eta}_{f}^{\mathrm{std}}(\mathcal{N}; \mathrm{Im}\:\mathcal{N}^n)\equiv\widecheck{\eta}_{f}^{\mathrm{std}}(\mathcal{N}^{n+1}, \mathcal{N}^n) \equiv \inf_{\substack{\rho \in \mathcal{N}^n(\mathcal{D}_d), X \in T_\rho \mathcal{N}^n(\mathcal{D}_d)}} \frac{\| \mathcal{N}(X) \|^2_{\kappa, \mathcal{N}(\rho)}}{\| X \|^2_{\kappa, \rho}} > 0
% \]
\[
\widecheck{\eta}_{f}^{\mathrm{std}}(\mathcal{N}; \mathrm{Im}\:\mathcal{N}^n)\equiv\widecheck{\eta}_{f}^{\mathrm{std}}(\mathcal{N}^{n+1}, \mathcal{N}^n) \equiv \inf_{\substack{\rho\neq\gamma \in \mathcal{N}^n(\mathcal{D}_d),\\\mathrm{supp}\:\rho=\mathrm{supp}\:\gamma}} \frac{D^{\std}_f(\mathcal{N}(\rho) \|\mathcal{N}(\gamma))}{D^{\std}_f(\rho \|\gamma)}> 0
\]

i.e. a reverse DPI holds. We conclude for many examples of quantum channels $\mathcal{N}:\mathcal{B}(\mathcal{H})\to\mathcal{B}(\mathcal{H})$ that a positive expansion coefficient will eventually be obtained after a sufficient number of iterations. For example, this is true for all qubit Pauli channels after $n=1$ iteration (see Section~\ref{sec:explicitcoeffs}).

% \section{On the equality or equivalence of contraction and expansion coefficients for $f$-divergence and its induced Riemannian metric}\label{sec:equivalence}
% %\section{Equivalence Results}\label{sec:equivalence}
% In this section, we study the question when 
% $$\eta^{\std}_f(\mc N,\mc M) = \eta^{\Riem}_{\kappa_f}(\mc N,\mc M),\quad \widecheck \eta^{\std}_f(\mc N,\mc M) = \widecheck \eta^{\Riem}_{\kappa_f}(\mc N,\mc M)$$ for any channels $\mc N, \mc M$, where $\kappa_f$ is defined by \eqref{eqn:relation f and kappa}. One direction always holds for these two coefficients, based on the following observations:
% \begin{lemma}
    
% \end{lemma}

% \subsection{Equality of contraction and expansion coefficients for $f$-divergence and its induced Riemannian metric}

\section{Equality between Divergence and Riemannian Coefficients}\label{sec:equalityofcoeffs}
It is currently a challenge to obtain exact expressions for both divergence and Riemannian relative expansion coefficients. However, this does not stop us from identifying connections between them. Such connections between the coefficients are valuable because they can reveal similarities (or differences, as in Theorem~\ref{thm:RiemInequiv}) between the respective distinguishability measures. If we can find a channel-independent relationship between the relative expansion coefficient of a standard $f$-divergence and that of their induced Riemannian semi-norm, then this is indicative of an underlying dependence of the standard $f$-divergence on its local behaviour and one may utilise this information to develop heuristics for the analytical form of relative expansion coefficients \cite{BGSW24}. 

Divergence relative expansion coefficients turn out to be even more difficult to evaluate in general than Riemannian relative expansion coefficients, so cases where they coincide, $\widecheck\eta_f^\mathrm{std}\equiv\widecheck\eta_{\kappa_f}^\mathrm{Riem}$, are informative. This formed a core strategy in \cite{BGSW24}. In this section, we discuss the limitations of relating these two types of relative expansion coefficients by a strict equality (Section~\ref{subsec:inequality}) before extrapolating the known cases of this equality to previously unknown choices of $f\in\mathcal{F}$ (Section~\ref{subsec:equalitycases}). In the next section, Section~\ref{sec:equivalence}, we will see that equality is not our only option to reduce the problem.

\subsection{Inequality between Divergence and Riemannian Coefficients}\label{subsec:inequality}
By \eqref{eqn:stddivlocallyRiem}, the local behaviour of a standard $f$-divergence $D_f^\mathrm{std}(\rho\|\gamma)$ is characterised by the Riemannian semi-norm $\|\rho-\gamma\|_{\kappa_f,\:\rho}^2$. This means that we can re-express the Riemannian relative expansion coefficient as follows:
\begin{align}\label{eqn:RiemCoeffAlt}
    \widecheck{\eta}^{\mathrm{Riem}}_{\kappa_f}(\mathcal{N}, \mathcal{M})&=
	\inf_{\substack{\rho \in \mathcal{D}(\mathcal{H}_A),\\ X  \in T_\rho \mathcal{D}(\mathcal{H}_A)}}
	\frac{
		\|\mathcal{N}(X)\|_{\kappa_f,\, \mathcal{N}(\rho)}^2
	}{
		\|\mathcal{M}(X)\|_{\kappa_f,\, \mathcal{M}(\rho)}^2
	}=
	\inf_{\substack{\rho \in \mathcal{D}(\mathcal{H}_A), \\X \in T_\rho \mathcal{D}(\mathcal{H}_A)}}
	\frac{
		\lim_{\varepsilon \to 0^+} \frac{1}{\varepsilon^2}
		D_f^{\mathrm{std}}\left(\mathcal{N}(\rho) \| \mathcal{N}(\rho) + \varepsilon \mathcal{N}(X)\right)
	}{
		\lim_{\varepsilon \to 0^+} \frac{1}{\varepsilon^2}
		D_f^{\mathrm{std}}\left(\mathcal{M}(\rho) \| \mathcal{M}(\rho) + \varepsilon \mathcal{M}(X)\right)
	}\nonumber\\
    &= \inf_{\substack{\rho \in \mathcal{D}(\mathcal{H}_A), \\X \in  T_{\rho} \mathcal{D}(\mathcal{H}_A)}} \lim_{\varepsilon \to 0^+} \frac{D^{\mathrm{std}}_f (\mathcal{N}(\rho) \| \mathcal{N}(\rho+\varepsilon X))}{D^{\mathrm{std}}_f (\mathcal{M}(\rho) \| \mathcal{M}(\rho+\varepsilon X))}
\end{align}

Since the Riemannian coefficients are then effectively the corresponding divergence coefficients optimised over a smaller set of input states, it is not difficult to understand the following inequality between these relative expansion coefficients. The result was first recognised for contraction coefficients by \cite{LR99}; we give a short alternative proof in the general setting.%their proof relied on a uniform convergence, so the hope is that the following modification can provide clarity. \paula{I wouldn't write it like that. More like, we give a short, alternative proof instead of criticising (?) theirs}

\begin{prop}\label{prop:RiemDivIneq}
Suppose we are provided some $f(x) \in \mathcal{F}$ and its corresponding $\kappa_f(x) = \frac{f(x) + \widetilde{f}(x)}{f''(1)( x-1)^2} \in \mathcal{K}$. Given any quantum channels $\mathcal{N} : \mathcal{B}(\mathcal{H}_A) \rightarrow \mathcal{B}(\mathcal{H}_B)$, $\mathcal{M} : \mathcal{B}(\mathcal{H}_A) \rightarrow \mathcal{B}(\mathcal{H}'_B)$,
\begin{equation*}
    \widecheck{\eta}^{\mathrm{std}}_{f}(\mathcal{N}, \mathcal{M})  \leq \widecheck{\eta}^{\mathrm{Riem}}_{\kappa_f}(\mathcal{N}, \mathcal{M})
\end{equation*}
\end{prop}
\begin{proof}
$ $\newline
Observe that for $\rho \in \mathcal{D}(\mathcal{H}_A)$, $\mathrm{supp}\:\rho=\mathrm{supp}\:\gamma,X\in T_\rho \mathcal{D}(\mathcal{H}_A),  \varepsilon \in (0,1], \rho_\varepsilon :=\rho +\varepsilon X$ satisfies $\rho_\varepsilon\in\mathcal{D}(\mathcal{H}_A)$ and $\mathrm{supp} \:\rho_\varepsilon=\mathrm{supp}\: \rho$ for sufficiently small $\varepsilon$. As a result,
\begin{align*}
	\widecheck{\eta}_f^{\mathrm{std}}(\mathcal{N}, \mathcal{M}) &=
	\inf_{\substack{\rho \neq \gamma \in \mathcal{D}(\mathcal{H}_A), \\ \mathrm{supp} \,\rho = \mathrm{supp}\, \gamma}}
	\frac{
		D_f^{\mathrm{std}}(\mathcal{N}(\rho) \| \mathcal{N}(\gamma))
	}{
		D_f^{\mathrm{std}}(\mathcal{M}(\rho) \| \mathcal{M}(\gamma))
	}
	\\
	&\leq
	\inf_{\substack{\rho\neq\gamma \in \mathcal{D}(\mathcal{H}_A), \\X \in T_\rho \mathcal{D}(\mathcal{H}_A)}}
	\lim_{\varepsilon \to 0^+}
	\frac{
		D_f^{\mathrm{std}}\left(\mathcal{N}(\rho) \| \mathcal{N}(\rho) + \varepsilon \mathcal{N}(X)\right)
	}{
		D_f^{\mathrm{std}}\left(\mathcal{M}(\rho) \| \mathcal{M}(\rho) + \varepsilon \mathcal{M}(X)\right)
	}\\
	&
	= \widecheck{\eta}^{\mathrm{Riem}}_{\kappa_f}(\mathcal{N}, \mathcal{M})
\end{align*}
The inequality is obtained by the fact that $\widecheck{\eta}_f^{\mathrm{std}}(\mathcal{N}, \mathcal{M})$ is optimised over all pairs of states with equal support, which includes pairs $\rho,\rho_\varepsilon$ for sufficiently small $\varepsilon$. 
\end{proof}

Ideally, we would like the converse $\widecheck{\eta}^{\mathrm{std}}_{f}(\mathcal{N}, \mathcal{M})  \geq \widecheck{\eta}^{\mathrm{Riem}}_{\kappa_f}(\mathcal{N}, \mathcal{M})$ to hold, since it is sometimes easier to lower bound $\widecheck{\eta}^{\mathrm{Riem}}_{\kappa_f}(\mathcal{N}, \mathcal{M})$ (see Section~\ref{sec:explicitcoeffs}). In fact, by the above result, if the converse also holds, we have the desired strict equality between divergence and Riemannian relative expansion coefficients. Unless we restrict the class of channels $\mathcal{N},\mathcal{M}$ (see later, Section~\ref{sec:qc}), this equality is not true for all $f \in \mathcal{F}$ and all pairs of channels. A particular family of classical-quantum (CQ) qubit channels $\Phi_{\alpha,\tau}$ was found in \cite[Theorem 6.6]{HR15}, for $\alpha^2+\tau^2\leq 1$, satisfying:\\
\begin{equation}\label{eqn:equalitycounter}
		\hat{\eta}^{\mathrm{std}}_{(f_s)_{\mathrm{sym}}}(\Phi_{\alpha, \tau})
		=
		\widecheck{\eta}^{\mathrm{std}}_{(f_s)_{\mathrm{sym}}}
		(\mathrm{id}_{\mathcal{B}(\mathcal{H}_A)}, \Phi_{\alpha, \tau})^{-1}
		>
		\widecheck{\eta}^{\mathrm{Riem}}_{\kappa_s}(\Phi_{\alpha, \tau})
		=
		\widecheck{\eta}^{\mathrm{Riem}}_{\kappa_s}
		(\mathrm{id}_{\mathcal{B}(\mathcal{H}_A)}, \Phi_{\alpha, \tau})^{-1}
\end{equation}\\
for $s \in [0, 1]$ sufficiently close to $1$, where $(f_s)_{\mathrm{sym}}(x) := (x-1)^2 \kappa_s(x) \in \mathcal{F}_{\mathrm{sym}}$. 

This gives us a reason to believe that we cannot always reduce the problem of showing a positive divergence expansion coefficient to showing that the corresponding Riemannian expansion coefficient is positive - at least not without developing more advanced results. In this case, these channels $\Phi_{\alpha,\tau}$ happen to be primitive for $\alpha^2+\tau^2<1$ (they even have full-rank output), so positive divergence coefficients can still be obtained (see Theorem~\ref{thm:primcoeffpos}). This family of channels is notable for providing good counterexamples, as it is later utilised in Theorem~\ref{thm:RiemInequiv} to show a geometric difference between the Riemannian semi-norms for bounded vs. unbounded $\kappa\in\mathcal{K}$.

\subsection{The Special Cases of Generic Equality}\label{subsec:equalitycases}
\noindent The reformulation \eqref{eqn:RiemCoeffAlt} makes a relationship of the form $\widecheck{\eta}^{\mathrm{std}}_{f}(\mathcal{N}, \mathcal{M}) = \widecheck{\eta}^{\mathrm{Riem}}_{\kappa_f}(\mathcal{N}, \mathcal{M})$ for all quantum channels $\mathcal{N},\mathcal{M}$, significant. Besides making it easier to demonstrate the positivity of divergence expansion coefficients over a restricted-domain (as the problem reduces to dealing with the Riemannian coefficient directly), we also know that the relative expansion coefficient, $\widecheck{\eta}^{\mathrm{std}}_{f}(\mathcal{N}, \mathcal{M})$, is always attained by $\frac{D^{\mathrm{std}}_{f}(\mathcal{N}(\rho) \| \mathcal{N}(\gamma))}{D^{\mathrm{std}}_{f}(\mathcal{M}(\rho) \| \mathcal{M}(\gamma))}$ in the limit of a pair of states that approach each other \cite{LR99}. This provides a useful heuristic to propose closed-form expressions for expansion coefficients, making use of the structural properties of the quantum channel. For example, \cite{BGSW24} have applied this type of reasoning to qubit amplitude damping channels, proposing that its infimum is attained near $\ket{1}\bra{1}$ for the relative entropy. 

For the equality $\eta_f^\mathrm{std}(\mathcal{N},\mathcal{M})=\eta_\kappa^\mathrm{Riem}(\mathcal{N},\mathcal{M})$ to hold for a large family of (pairs of) channels, there turns out to be a deep connection with the existence of an integral relation between $D_f^\mathrm{std}(\rho\|\gamma)$ and $\|\rho-\gamma\|_{\kappa_f,\:\rho}^2$ for $\mathrm{supp}\, \rho \leq \mathrm{supp}\, \gamma$ (Section~\ref{sec:qc} elaborates on this). However, in the quantum setting, only two such relationships are known \cite{HR15,BGSW24}:
\begin{equation}\label{eqn:specialrelationships}
    D_{(x-1)^2}^\mathrm{std}(\rho\|\gamma)=\| \rho-\gamma \|^2_{\kappa_{\max,\gamma}}\equiv\Tr\:(\rho-\gamma)^2\gamma^{-1},\quad D(\rho\|\gamma) = \int_0^1 \int_0^s \| \rho-\gamma \|^2_{\kappa_{\mathrm{BKM}},\:\rho_t} dt\:ds
\end{equation}
where $\kappa_\mathrm{max}(x):=\frac{x+1}{2x},\kappa_\mathrm{BKM}(x):=\frac{\log x}{x-1},\rho_t := (1-t)\gamma + t\rho$.

Using these special relationships, we can establish two large classes of functions $f\in\mathcal{F}$ for which this equality between divergence and Riemannian coefficients holds:

\begin{theorem}\label{thm:equalitycases}
    For any two channels $\mathcal{N}: \mathcal{B}(\mathcal{H}_A) \to \mathcal{B}(\mathcal{H}_B)$, $\mathcal{M}: \mathcal{B}(\mathcal{H}_A) \to \mathcal{B}(\mathcal{H}_B')$
    \begin{enumerate}
        \item[(i)] For $f(x)=\alpha\cdot (x-1)^2+\beta\cdot\frac{(x-1)^2}{x},\alpha,\beta\geq0$,
        % :
        % \begin{equation}
        %     \widecheck\eta_{\max}^\mathrm{std}(\mathcal{N},\mathcal{M})=\widecheck\eta_{\max}^\mathrm{Riem}(\mathcal{N},\mathcal{M})
        % \end{equation}
        \begin{equation}
            \widecheck\eta_{f}^\mathrm{std}(\mathcal{N},\mathcal{M})=\widecheck\eta_{\kappa_\mathrm{max}}^\mathrm{Riem}(\mathcal{N},\mathcal{M})
        \end{equation}
        where $\kappa_f(x)=\frac{f(x)+\widetilde{f}(x)}{f''(1)(x-1)^2}=\frac{x+1}{2x}\equiv \kappa_\mathrm{max}(x)$.
        \item[(ii)] For $f(x)=\alpha\cdot x\log x+\beta\cdot-\log x,\alpha,\beta\geq0$,
        \begin{equation}
            \widecheck\eta_{f}^\mathrm{std}(\mathcal{N},\mathcal{M})=\widecheck\eta_{\kappa_\mathrm{max}}^\mathrm{Riem}(\mathcal{N},\mathcal{M})
        \end{equation}
        where $\kappa_f(x)=\frac{f(x)+\widetilde{f}(x)}{f''(1)(x-1)^2}=\frac{\log x}{x-1}\equiv \kappa_\mathrm{BKM}(x)$.
        % :
        % \begin{equation}
        %     \widecheck\eta_{x\log x}^\mathrm{std}(\mathcal{N},\mathcal{M})=\widecheck\eta_{\mathrm{BKM}}^\mathrm{Riem}(\mathcal{N},\mathcal{M})
        % \end{equation}
    \end{enumerate}
\end{theorem}
\begin{proof}
Since we already have that $\widecheck{\eta}^{\mathrm{std}}_f (\mathcal{N}, \mathcal{M}) \leq \widecheck{\eta}^{\mathrm{Riem}}_{\kappa_f} (\mathcal{N}, \mathcal{M}) $ for all $f$ by Proposition~\ref{prop:RiemDivIneq}, we will proceed to show that $\widecheck{\eta}^{\mathrm{std}}_f (\mathcal{N}, \mathcal{M}) \geq \widecheck{\eta}^{\mathrm{Riem}}_{\kappa_f} (\mathcal{N}, \mathcal{M}) $, in order to obtain equality.\\ \\
\textit{(i)} For $f(x)=\alpha\cdot (x-1)^2+\beta\cdot\frac{(x-1)^2}{x},\alpha,\beta\geq0$, we have for any $\rho',\gamma'\in\mathcal{D}_d$,
\[
D_{f}^\mathrm{std}(\rho' \| \gamma')=\alpha  D_{(x-1)^2}^\mathrm{std}(\rho' \| \gamma') + \beta D_{(x-1)^2}^\mathrm{std}(\gamma' \| \rho').
\]
This is because $\frac{(x-1)^2}{x}$ is the transpose of $(x-1)^2$.
Therefore, 
\begin{align*}
    \widecheck{\eta}^{\mathrm{std}}_f (\mathcal{N}, \mathcal{M}) 
    &= 
    \inf_{\substack{\rho \neq \gamma \in \mathcal{D}(\mathcal{H}_A), \\ \mathrm{supp}\, \rho =\mathrm{supp}\, \gamma}}
    \frac{ \alpha  D_{(x-1)^2}^\mathrm{std}(\mathcal{N}(\rho) \| \mathcal{N}(\gamma)) + \beta D_{(x-1)^2}^\mathrm{std}(\mathcal{N}(\gamma) \| \mathcal{N}(\rho))}
    { \alpha D_{(x-1)^2}^\mathrm{std}(\mathcal{M}(\rho) \| \mathcal{M}(\gamma))+\beta D_{(x-1)^2}^\mathrm{std}(\mathcal{M}(\gamma) \| \mathcal{M}(\rho)) } 
     \\
        &=
    \inf_{\substack{\rho \neq \gamma \in \mathcal{D}(\mathcal{H}_A), \\ \mathrm{supp}\, \rho =\mathrm{supp}\, \gamma}}
    \frac{
        \displaystyle \alpha\left\| \mathcal{N}(\rho-\gamma) \right\|^2_{\kappa_{\mathrm{max}},\, \mathcal{N}(\gamma)}+\beta\left\| \mathcal{N}(\rho-\gamma) \right\|^2_{\kappa_{\mathrm{max}},\, \mathcal{N}(\rho)} 
    }{
        \displaystyle \alpha\left\| \mathcal{M}(\rho-\gamma) \right\|^2_{\kappa_{\mathrm{max}},\mathcal{M}(\gamma)}+\beta \left\| \mathcal{M}(\rho-\gamma) \right\|^2_{\kappa_{\mathrm{max}},\, \mathcal{M}(\rho)} 
    }\\ \\
    &\geq 
    \inf_{\substack{\rho \neq \gamma \in \mathcal{D}(\mathcal{H}_A), \\ \mathrm{supp}\, \rho =\mathrm{supp}\, \gamma}} \frac{
        \displaystyle  \widecheck\eta_\mathrm{max}^\mathrm{Riem}(\mathcal{N},\mathcal{M})(\alpha\left\| \mathcal{M}(\rho-\gamma) \right\|^2_{\kappa_{\mathrm{max}},\, \mathcal{M}(\gamma)}+\beta\left\| \mathcal{M}(\rho-\gamma) \right\|^2_{\kappa_{\mathrm{max}},\, \mathcal{M}(\rho)}) 
    }{
        \displaystyle \alpha\left\| \mathcal{M}(\rho-\gamma) \right\|^2_{\kappa_{\mathrm{max}},\mathcal{M}(\gamma)}+\beta \left\| \mathcal{M}(\rho-\gamma) \right\|^2_{\kappa_{\mathrm{max}},\, \mathcal{M}(\rho)}
    }\\ \\
    &= \widecheck{\eta}^{\mathrm{Riem}}_{\kappa_{\mathrm{max}}}(\mathcal{N},\mathcal{M})
\end{align*}
The second equality comes from the relationship \eqref{eqn:specialrelationships}. The inequality holds by applying the definition of $\widecheck\eta_\mathrm{max}^\mathrm{Riem}(\mathcal{N},\mathcal{M})$ to each term of the numerator. 

\textit{(ii)} For $f(x)=\alpha\cdot x\log x+\beta\cdot-\log x,\alpha,\beta\geq0$, we have for any $\rho',\gamma'\in\mathcal{D}_d$,
\[
D_{f}^\mathrm{std}(\rho' \| \gamma')=\alpha  D(\rho' \| \gamma') + \beta D(\gamma' \| \rho').
\]
This is because $-\log x$ is the transpose of $x\log x$. The proof proceeds in a very similar way to \textit{(i)}.
\begin{align*}
    \widecheck{\eta}^{\mathrm{std}}_f (\mathcal{N}, \mathcal{M}) 
    &= 
    \inf_{\substack{\rho \neq \gamma \in \mathcal{D}(\mathcal{H}_A), \\ \mathrm{supp}\, \rho =\mathrm{supp}\, \gamma}}
    \frac{ \alpha D(\mathcal{N}(\rho) \| \mathcal{N}(\gamma)) + \beta D(\mathcal{N}(\gamma) \| \mathcal{N}(\rho))}
    { \alpha D(\mathcal{M}(\rho) \| \mathcal{M}(\gamma))+\beta D(\mathcal{M}(\gamma) \| \mathcal{M}(\rho)) } 
     \\
        &=
    \inf_{\substack{\rho \neq \gamma \in \mathcal{D}(\mathcal{H}_A), \\ \mathrm{supp}\, \rho =\mathrm{supp}\, \gamma}}
    \frac{
        \displaystyle \int_0^1 \int_0^s \alpha\left\| \mathcal{N}(\rho-\gamma) \right\|^2_{\kappa_{\mathrm{BKM}},\, \mathcal{N}(\rho_t)}+\beta\left\| \mathcal{N}(\rho-\gamma) \right\|^2_{\kappa_{\mathrm{BKM}},\, \mathcal{N}(\rho_{1-t})} \, dt \, ds
    }{
        \displaystyle \int_0^1 \int_0^s \alpha\left\| \mathcal{M}(\rho-\gamma) \right\|^2_{\kappa_{\mathrm{BKM}},\mathcal{M}(\rho_t)}+\beta \left\| \mathcal{M}(\rho-\gamma) \right\|^2_{\kappa_{\mathrm{BKM}},\, \mathcal{M}(\rho_{1-t})} \, dt \, ds
    }\\ \\
    &\geq 
    \inf_{\substack{\rho \neq \gamma \in \mathcal{D}(\mathcal{H}_A), \\ \mathrm{supp}\, \rho =\mathrm{supp}\, \gamma}} \frac{
        \displaystyle \int_0^1 \int_0^s \widecheck\eta_\mathrm{BKM}^\mathrm{Riem}(\mathcal{N},\mathcal{M})(\alpha\left\| \mathcal{M}(\rho-\gamma) \right\|^2_{\kappa_{\mathrm{BKM}},\, \mathcal{M}(\rho_t)}+\beta\left\| \mathcal{M}(\rho-\gamma) \right\|^2_{\kappa_{\mathrm{BKM}},\, \mathcal{M}(\rho_{1-t})}) \, dt \, ds
    }{
        \displaystyle \int_0^1 \int_0^s \alpha\left\| \mathcal{M}(\rho-\gamma) \right\|^2_{\kappa_{\mathrm{BKM}},\mathcal{M}(\rho_t)}+\beta \left\| \mathcal{M}(\rho-\gamma) \right\|^2_{\kappa_{\mathrm{BKM}},\, \mathcal{M}(\rho_{1-t})} \, dt \, ds
    }\\ \\
    &= \widecheck{\eta}^{\mathrm{Riem}}_{\kappa_{\mathrm{BKM}}}(\mathcal{N},\mathcal{M})
\end{align*}
% The second equality comes from the relationship \ref{eqn:specialrelationships}. The inequality holds by applying the definition of $\widecheck\eta_\mathrm{max}^\mathrm{Riem}(\mathcal{N},\mathcal{M})$ to each term of the numerator. 

\end{proof}
Effectively, two triplets of cases were already acknowledged in previous works, where the divergence and Riemannian contraction/relative expansion coefficients coincide for all pairs of
channels, $\mathcal{N} : \mathcal{B}(\mathcal{H}_A) \rightarrow \mathcal{B}(\mathcal{H}_B), \mathcal{M} : \mathcal{B}(\mathcal{H}_A) \rightarrow \mathcal{B}(\mathcal{H}_B')$. Namely, $f(x)=(x-1)^2,(x-1)\log x$ were observed in \cite{HR15}, and they proposed an open problem that $f(x)=x\log x$ is also a solution (so that $f_\mathrm{sym}(x)=(x-1)\log x$), which was recently solved by \cite{GR22,BGSW24}. Cases of equality appear either alone or in triplets (up to scaling $f\in\mathcal{F}$) because, by \eqref{eqn:stddivgeneral},
\begin{align*}
    \widecheck{\eta}^{\mathrm{std}}_{f}(\mathcal{N}, \mathcal{M}) &= \widecheck{\eta}^{\mathrm{Riem}}_{\kappa_f}(\mathcal{N}, \mathcal{M}) \implies \widecheck{\eta}^{\mathrm{std}}_{f}(\mathcal{N}, \mathcal{M}) = \widecheck{\eta}^{\mathrm{std}}_{\widetilde{f}}(\mathcal{N}, \mathcal{M})=\widecheck{\eta}^{\mathrm{std}}_{f_{\mathrm{sym}}}(\mathcal{N}, \mathcal{M}) = \widecheck{\eta}^{\mathrm{Riem}}_{\kappa_f}(\mathcal{N}, \mathcal{M}),
\end{align*}
i.e. this equality for a function $f$ automatically implies that it holds for the corresponding transpose $\widetilde{f}$ and symmetrised version $f_{\mathrm{sym}}$ (note: they all have the same $\kappa_f$). And now, in Theorem~\ref{thm:equalitycases}, we have extended these results to the large family of conical combinations of $f,\tilde f$, for $f(x)=(x-1)^2,x\log x$. 

As it turns out, these relationships uniquely characterise the standard $f$-divergence. We will later explain how in the classical setting, the relationship between an $f$-divergence and the (unique) Riemannian semi-norm fully determines $f\in\mathcal{F}$ (Theorem~\ref{thm:classicspecialrel}), and this is the reason for the uniqueness. However, we illustrate a convenient method here that may be used to check candidate relationships in the future. 

\begin{prop}\label{prop:MaxStdDivUnique}
Fix a complex Hilbert space $\mathcal H$ with $\dim\mathcal H\ge 2$. 
Let $\rho\in\mathcal D(\mathcal H)$ and let $X=X^\dagger$ be traceless with 
$\operatorname{supp}X\subseteq \operatorname{supp}\rho$. 
For $t\in[0,1]$, set $\rho_t:=\rho+tX$ and assume $\rho_t\ge 0$.

\begin{enumerate}
\item[(i)]
Suppose there exist $f\in\mathcal F$ and $w:[0,1]\to\mathbb R$ such that, 
for all choices of $(\rho,X,t)$ as above,
\begin{equation}\label{eq-OnlyDfmax}
    D_f(\rho\,\|\,\rho_t)= w(t)\,\|X\|_{\rho_t,\kappa_{f'}}^2 .
\end{equation}
Then necessarily $D_f$ is (a constant multiple of) 
the maximal standard $f$-divergence, and $w(t)=c\,t^2$.

\item[(ii)]
Suppose there exist $f\in\mathcal F$ and $w:[0,1]\to\mathbb R$ such that, 
for all choices of $(\rho,X,t)$ as above,
\begin{equation}\label{eqn:RelEntRelation2}
    \frac{d^2}{dt^2}\, D_f(\rho_t\,\|\,\rho) \;=\; 
    w(t)\,\|X\|_{\rho_t,\kappa_{f}}^2 .
\end{equation}
Then necessarily $D_f$ is (a constant multiple of) the 
Umegaki relative entropy. In particular, one may take $w(t)\equiv c$ and 
$\kappa_{f}$ the BKM kernel.
\end{enumerate}
\end{prop}

% \begin{prop}\label{prop:MaxStdDivUnique}
%     The only $f$-divergence satisfying, for all Hilbert spaces $\mathcal{H}$ with $\dim \mathcal{H} \geq 2$, $\rho \in \mathcal{D}(\mathcal{H})$, and traceless Hermitian operator $X \in \mathcal{B}(\mathcal{H})$
%     such that $\mathrm{supp}\, X \leq \mathrm{supp}\, \rho$ and $\rho_t := \rho + t X \geq 0$ for all $t \in [0,1]$:
%     \begin{enumerate}
%         \item[(i)] \begin{equation}\label{eq-OnlyDfmax}
%         D_f(\rho \| \rho_t) = w(t) \|X\|_{\rho_t,\kappa_{f}}^2
%     \end{equation}
%     for $f,f' \in \mathcal{F},\, w : [0,1] \to \mathbb{R}$, is the maximal standard $f$-divergence, $D_f^\mathrm{std},f(x) = (x-1)^2$, with $w(t) = t^2$ (up to constant scaling).
%     \end{enumerate}
%      \item[(ii)] \begin{equation}\label{eqn:RelEntRelation2}
% 		    \frac{d^2}{dt^2} D_f(\rho_t\|\rho)  = w(t) \| X \|_{\rho_t, \kappa_{f}}^2
% 		\end{equation}
% 		for $f, f' \in \mathcal{F}, w: [0,1] \to \mathbb{R}$, is the relative entropy, $D_f^\mathrm{std}\equiv D,f(x) = x\log x$ (up to constant scaling).
% \end{prop}

\begin{proof}
    Let $\rho = \lambda \ket{0}\bra{0} + (1-\lambda)\ket{1}\bra{1},X= \mu (\ket{0}\bra{0}-\ket{1}\bra{1})$ where $\lambda\in(0,1), \mu \in (-\lambda,0)\cup (0,1-\lambda)$ and $\ket{0}, \ket{1}$ are orthonormal. Define the density operator $\rho_t := \rho + t X,\, t \in [0,1]$.
    
    \textit{(i)} We evaluate:
    \begin{equation*}
        D_f(\rho \| \rho_t) = (\lambda + t\mu) f\left( \frac{\lambda}{\lambda + t\mu} \right) + (1-\lambda-t\mu) f \left( \frac{1-\lambda}{1-\lambda - t\mu} \right)\notag
    \end{equation*}
    and 
    \begin{equation}
        w(t) \|X\|^2_{\rho_t, \kappa_f}
        = w(t)\mu^2\left(\frac{1}{\lambda + t\mu} + \frac{1}{1 - \lambda - t\mu}\right)
    \end{equation}

    Suppose 
    \begin{equation}\label{eqn:onlyDfmax2}
    D_f^\mathrm{std}(\gamma \| \rho_t)=w(t) \| X \|^2_{\rho_t, \kappa_{f'}}
    \end{equation}
    \noindent Fix $t$, and keep $\lambda, \mu$ variable. Define $\chi = \lambda + t\mu$, and let us maintain $\lambda = x\chi$ for fixed $x \in (1, \frac{1}{1-t})$ by choosing $\mu$ appropriately. Thus, dividing both sides of \eqref{eqn:onlyDfmax2} by $\mu t = (1-x)\chi$: 
    \begin{equation}\label{eqn:subuniquemax}
        \frac{1}{1-x} f(x) + \frac{1-\chi}{(1-x)\chi} f\left( \frac{1-x\chi}{1-\chi} \right)
        = \frac{w(t)}{t^2} \left(1-x + \frac{(1-x)\chi}{1-\chi}\right)
    \end{equation}
    
    \noindent We let $\chi \to 0$ (note: $\chi \in ((1-t)\lambda, \lambda) \cup (\lambda,  t+(1-t)\lambda)$, so as we take $\lambda \to 0$, this limit becomes possible), considering the following:
    \[
    \frac{1-\chi}{\chi} f\left( \frac{1-x\chi}{1-\chi} \right) \sim \frac{1}{\chi}\cdot f'(1)\cdot \frac{(1-x)\chi}{1-\chi} \sim (1-x) f'(1) \text{ as } \chi \to 0.
    \]
    As a result, \eqref{eqn:subuniquemax} becomes:
    \begin{equation*}
        \frac{1}{1-x} f(x) + f'(1) 
        = \frac{w(t)}{t^2}(1-x)\\
        \implies \frac{f(x)}{(x-1)^2} + \frac{f'(1)}{x-1} = \frac{w(t)}{t^2} = C
    \end{equation*}
    for all $t,\, x \in (1, \frac{1}{1-t})$ (hence all $t$ and all $x\in(1,\infty)$, separately) and some constant $C\geq 0$. Observe that after dividing both sides by $x-1$, both sides of the first equality of the second equation above, have exactly one of the variables $x$ and $t$, which may be varied independently, in spite of the fact their ranges make them dependent. This is why both sides are constant.
    
    This can be re-expressed as
    \begin{equation*}
        f(x)=f'(1)(x-1)+C(x-1)^2
    \end{equation*}
    \noindent For all $x\in(1,\infty)$, but since $f\in\mathcal{F}$ is real analytic on $(0,\infty)$, we can extend this to all of $x\in\mathbb{R}$. wlog we take $f'(1) = 0,C=1$ (since $D_f^\mathrm{std} \equiv C\cdot D_{({f-f'(1)(x-1)})/C}^\mathrm{std}$) and obtain $f(x)=(x-1)^2$. So we show that only this $D_f$ is a candidate, corresponding to the \emph{maximal standard $f$-divergence}. \eqref{eqn:specialrelationships} tells us that the maximal standard $f$-divergence indeed satisfies \eqref{eq-OnlyDfmax}, so this is the only solution (up to scaling).

     \textit{(ii)} We evaluate:
    \begin{align*}
        D_f(\rho_t\|\rho) &= \lambda f\left( \frac{\lambda + t \mu}{\lambda} \right) + ( 1-\lambda) f\left( \frac{ 1-\lambda - t \mu}{1-\lambda} \right)
    \end{align*}
   Thus we have
    \begin{align*}
        \frac{d^2}{dt^2} D_f(\rho_t\|\rho) &= \frac{\mu^2}{\lambda} f''\left( \frac{\lambda + t \mu}{\lambda} \right) + \frac{\mu^2}{1-\lambda} f''\left( \frac{ 1-\lambda - t \mu}{1-\lambda} \right) \\
    \end{align*}
    and
    \begin{align*}
        w(t) \| X \|^2_{\rho_t, \kappa_{f'}} &= w(t) \Tr\:X^2\rho_t^{-1} = w(t) \left( \frac{\mu^2}{\lambda + t \mu}+\frac{\mu^2}{1-\lambda - t \mu}  \right).
    \end{align*}

Suppose 
\begin{equation}\label{eqn:RelEntRelation}
    \frac{d^2}{dt^2} D_f(\gamma \| \rho_t)=w(t) \| X \|^2_{\rho_t, \kappa_{f'}}
\end{equation}\\
We will now show that this equality uniquely corresponds to a particular function $f$ and a particular weight $w(t)$.

Dividing both sides by $\mu^2$ and taking $\mu \to 0$, we compare both sides: 
    \begin{align*}
        \left( \frac{1}{\lambda} + \frac{1}{1-\lambda} \right) f''(1) = \left( \frac{1}{\lambda} + \frac{1}{1 - \lambda} \right) w(t)
        \implies w(t) = f''(1) \quad \forall t \in [0,1]
    \end{align*}
    
   \noindent Returning to \eqref{eqn:RelEntRelation}, consider re-parametrising to only the variables $\lambda \in [0,1],\, \chi = \lambda+ t\mu \in (0,1)\setminus \{\lambda\}$, then we obtain:
    \begin{align*}
        &\frac{1}{\lambda} f''\left(\frac{\lambda}{\chi}\right)
        + \frac{1}{1-\lambda} f''\left(\frac{1-\lambda}{1-\chi}\right)
        = f''(1)\left(\frac{1}{\chi} + \frac{1}{1-\chi}\right)
        = \frac{f''(1)}{\chi(1-\chi)}
    \end{align*}
    Therefore,
    \begin{align*}
        \frac{\chi}{\lambda}\cdot  f''\left(\frac{\lambda}{\chi}\right)
        + \frac{\chi}{1-\lambda} f''\left(\frac{1-\lambda}{1-\chi}\right) = \frac{f''(1)}{1-\chi}
    \end{align*}
    
    Let $ \chi \to 0 $ while maintaining $\lambda = x\chi$, any fixed $x\in(0,\infty)\setminus\{1\}$ (so that also $\lambda\to0$):
    \begin{align*}
        \frac{\chi}{\lambda}\cdot  f''\left(\frac{\chi}{\lambda}\right)
        + \frac{\chi}{1-\lambda} f''\left(\frac{1-\chi}{1-\lambda}\right) = \frac{f''(1)}{1-\chi}\to x f''(x) + 0 \cdot f''(1) = x f''(x)=f''(1)
    \end{align*}
    Overall, we have
    \begin{align*}
        &\forall x\in(0,\infty):\quad x f''(x) = f''(1),\ f(1)=0\\
        &\implies f(x) = -f''(1) x\log x + c(x-1),\:\text{constant } c\in\mathbb{R},\ \forall x\in(0,\infty)
    \end{align*}
    \noindent Solving this equation, wlog setting $f'(1)=0,f''(1) = 1$ (since $D_f \equiv f''(1)D_{{(f-f'(1)(x-1))/f''(1)}}$), we obtain $f(x)=x\log x$. So we show that only this $D_f$ is a candidate, corresponding to the relative entropy. \eqref{eqn:specialrelationships} tells us that the relative entropy indeed satisfies 
    \eqref{eqn:RelEntRelation2}, so this is the only solution (up to scaling).
    % \noindent Alternatively,\paula{Why?} take $\chi = 1- \lambda-t\mu, 1-\lambda=x\chi$ for fixed $x \in (1-t, 1)$ then if we again divide both sides by $\mu t = (x-1)\chi$:
    % \[
    %     \frac{1-x}{(x-1)\chi} f\left( \frac{1-x\chi}{1-x} \right) + \frac{1}{x-1} f(x)
    %     = \frac{w(t)}{t^2}\left(\frac{(x-1)\chi}{1-\chi} + x-1\right)
    % \]
    % Letting $\chi \to 0$:
    
    % \begin{equation*}
    %     \quad \quad f'(1) + \frac{1}{x-1} f(x) 
    %     = \frac{w(t)}{t^2}(x-1)\\
    %     \implies \frac{f(x)}{(x-1)^2} + \frac{f'(1)}{x-1} = \frac{w(t)}{t^2} = C
    % \end{equation*}\\
    % for all $t,\, x \in (1-t,1)$ (hence all $t$ and all $x < 1$) and some constant $C\geq 0$.\\
    % \noindent wlog taking $C=1$, $f'(1) = 0$, we get the result.
\end{proof}
Thus, the same integral representation for the quantum relative entropy by \cite{BGSW24} (from \eqref{eqn:specialrelationships}) cannot be reused to relate other standard $f$-divergences with their respective Riemannian semi-norms. The above arguments are easily generalised to show that the integral relations for any of the conical combinations $\alpha f+\beta \widetilde{f},\alpha,\beta>0$ from Theorem~\ref{thm:equalitycases} are unique. Our strategy in Propositions~\ref{prop:MaxStdDivUnique} involves reducing the problem to checking the condition on commuting pairs of qubit states, which means that we only worked in the two-dimensional, classical setting. This turns out to be sufficient in these cases, in order to rule out all other $f\in\mathcal{F}$. After all, the standard $f$-divergences are uniquely determined by their corresponding classical $f$-divergences \cite[Lemma 2.9]{Hiai_2011}.
	% \item Notice that in the classical setting, a specific relationship between the classical $f$-divergence and the Riemannian semi-norm is not necessary for $\widecheck{\eta}_f^{cl}(A,B) = \widecheck{\eta}^{\mathrm{Riem}, cl}(A,B)$, which simply holds indefinitely.
In Proposition~\ref{prop:MaxStdDivUnique}, it is important to bear in mind that the Petz- and sandwiched- Rényi divergences \cite{nussbaum2009chernoff,nussbaum2011asymptotic,muller2013quantum}, $D_\alpha,\widetilde{D}_\alpha,\alpha\in(0,\infty)\setminus \{1\}$ respectively, are not standard $f$-divergences, so this needs to be checked separately. For example, the Petz-Rényi divergence  is $D_\alpha(\rho\|\gamma)\equiv\frac{1}{\alpha-1}\log D_{x^\alpha}^\mathrm{std}(\rho\|\gamma)$ for states $\rho,\gamma\in\mathcal{D}_d$, $\mathrm{supp}\:\rho\leq\mathrm{supp}\:\gamma$. One might suspect an integral relation between $D_\alpha$ or $\widetilde{D}_\alpha$ and some other Riemannian semi-norms, of the same form as the relative entropy \eqref{eqn:RelEntRelation2}, because
\[D_\alpha(\rho\|\gamma),\widetilde{D}_\alpha(\rho\|\gamma)\to D(\rho\|\gamma)\text{ as }\alpha\to1.
\] In fact, working again in the classical two-dimensional case, it can verified that these (or any) quantum Rényi divergences do not satisfy this form of relation for $\alpha\neq 1$. 

\subsection{Consequence of a Classical Output}\label{sec:qc}

\noindent It was shown in \cite{Cohen1993, Choi1994} that for classical channels $A: P_{d_A} \to P_{d_B}$, the contraction coefficients (over unrestricted-domain) become redundant for operator convex functions $f\in \mathcal{F}$, i.e. $\hat{\eta}_f^\mathrm{cl}(A) = \hat{\eta}^\mathrm{Riem,cl}(A) \;\forall f\in\mathcal{F}$. That is, any difference between standard $f$-divergences and Riemannian semi-norms via contraction coefficients (and relative expansion coefficients) is a consequence of considering quantum channels. By essentially the same ideas, we will see that when we restrict $\mathcal{N},\mathcal{M}$ to be quantum-classical (QC) channels, we obtain the same redundancy. We infer that having non-commuting output states is important for the relative expansion coefficients to have variation. In proving this result, we will see again (like in Section~\ref{subsec:equalitycases}) that it is key to have an integral representation of the standard $f$-divergence in terms of the Riemannian semi-norm, though restricted to commuting pairs of states. 

A quantum channel $\mathcal{N}: \mathcal{B}(\mathcal{H}_A) \to \mathcal{B}(\mathcal{H}_B)$ is called \emph{quantum-classical (QC)} if $\mathrm{Im} (\mathcal{N})$ is a commutative subalgebra of $\mathcal{B}(\mathcal{H}_B)$.
% A quantum channel $\mathcal{N}: \mathcal{B}(\mathcal{H}_A) \to \mathcal{B}(\mathcal{H}_B)$ is called \textit{classical-quantum (CQ)} if $\mathrm{Im} \, \hat{
% \mathcal{N}}$ (where $\hat{\mathcal{N}}$ is the adjoint) is in a commutative subalgebra of $\mathcal{B}(\mathcal{H}_A)$. 
As a result, $\mathcal{N}$ is QC if there is an orthonormal basis $\{ \ket{\psi_i}\}_{i=1}^{d_B}$ of $\mathcal{H}_B$ and a POVM $\{ F_i \}_{i=1}^{d_B} \subseteq \mathcal{B}(\mathcal{H}_A)$ such that \cite{HR15}
\[
\mathcal{N}(\rho) = \sum_i (\mathrm{Tr}\:F_i\rho)\ket{\psi_i}\bra{\psi_i}
\]

\noindent Observe that if $\mathcal{N}$ is QC, $\| \mathcal{N}(X) \|^2_{\kappa, \mathcal{N}(\rho)} = \mathrm{Tr}\:\mathcal{N}(X)^2 \mathcal{N}(\rho)^{-1}= \|\mathcal{N}(X)\|^2_{\kappa_\mathrm{max},\mathcal{N}(\rho)}$ for all $\kappa \in \mathcal{K}$. This causes the Riemannian relative expansion coefficients to completely lose their variation in $\kappa$ when $\mathcal{N}, \mathcal{M}$ are QC channels:

\[
\widecheck{\eta}_{\kappa}^{\mathrm{Riem}}(\mathcal{N}, \mathcal{M}) = 
\inf_{\substack{\rho \in \mathcal{D}_d \\ X \in T_\rho \mathcal{D}_d}} 
\frac{\mathrm{Tr}\, \mathcal{N}(X)^2 \mathcal{N}(\rho)^{-1}}{\mathrm{Tr}\, \mathcal{M}(X)^2 \mathcal{M}(\rho)^{-1}} 
= \widecheck{\eta}_{\kappa_\mathrm{max}}^\mathrm{Riem}(\mathcal{N}, \mathcal{M})
\]
for all $\kappa \in \mathcal{K}$.\\

\noindent In fact, in the classical setting, we will now show that the classical $f$-divergences can be written in an integral representation, in terms of the Riemannian semi-norms. This inevitably results in the equality of all of the relative expansion coefficients, when the channels $\mathcal{N},\mathcal{M}$ are QC.

\begin{theorem}\cite{Choi1994}\label{thm:classicspecialrel}
    Provided some $f\in\mathcal{F}$, for all pairs of states $\rho,\gamma\in\mathcal{D}_d$ that commute and have equal support, $\mathrm{supp}\:\rho=\mathrm{supp}\:\gamma$, for any $\kappa\in\mathcal{K}$:
    \begin{equation}\label{eqn:classicspecialrel}
        D_f^\mathrm{std}(\rho\|\gamma)=c\|\rho-\gamma\|_{\kappa,\gamma}^2+\int_{(1,\infty)}\frac{s^2+1}{s^2}\|\rho-\gamma\|_{\kappa,\gamma+\frac{\rho-\gamma}{s}}^2 d\mu(s)
    \end{equation}
    where the constant $c\geq 0$ and positive measure $\mu$ are the same as in the integral representation of $f$ \eqref{opcon}.\\
    Further, this uniquely defines the functional $D_f^\mathrm{std}$. 
\end{theorem}
\begin{proof}
    \noindent By operator convexity, $f$ has an integral representation \eqref{opcon} that can be written as:
    \[ f(w) = f'(1)(w - 1) + c(w - 1)^2 + \int_{(1,\infty)} \frac{(w - 1) (s (w - 1) - 1)}{w - 1 + s} \, d\mu(s) \]
    
    \noindent But observe that $\frac{u(su-1)}{u+s} - \frac{s^2+1}{s} \cdot \frac{u^2}{u+s} = -\frac{u}{s}$, $u:=w-1$, and hence define \( g_s(w) := \frac{s^2 + 1}{s} \cdot \frac{(w - 1)^2}{w - 1 + s}\). 

    Now, we use the fact that $D_f^\mathrm{std}$ reduces to $D_f^\mathrm{cl}$ over commuting states. Suppose the states $\rho,\gamma$ have the following expressions in their common eigenbasis $\ket{\psi_i},i=1,...,d$:
    \begin{equation*}
        \rho=\sum_i x_i\ket{\psi_i}\bra{\psi_i},\quad \gamma=\sum_i y_i\ket{\psi_i}\bra{\psi_i}
    \end{equation*}
    \noindent Noting that \( D_{\frac{w-1}{s}}^{\mathrm{std}}(\rho \| \gamma) = D_{\frac{w-1}{s}}^{\mathrm{cl}}(\mathbf{x} \| \mathbf{y}) \equiv 0 \;\forall s\geq 1\), we must have
    \begin{equation*}
        D_{f}^{\mathrm{std}}(\rho \| \gamma) = D_{f}^{\mathrm{cl}}(\mathbf{x} \| \mathbf{y})=cD_{(w-1)^2}^{\mathrm{cl}}(\mathbf{x} \| \mathbf{y})+\int_{(1,\infty)}D_{g_s}^{\mathrm{cl}}(\mathbf{x} \| \mathbf{y})d\mu(s)
    \end{equation*}

    So evaluating all of the components:
    \begin{align*}
        D_{g_t}^{\mathrm{std}}\left(\rho\| \gamma\right)
        &= D_{g_t}^{\mathrm{cl}}(\mathbf{x}\|\mathbf{y})\\
        &
        = \frac{s^2+1}{s} \sum_{i=1}^{d}
        y_i \cdot  \frac{\left( \frac{x_i}{y_i} - 1 \right)^2}{\frac{x_i}{y_i} - 1 + s}  = \frac{s^2+1}{s} \sum_{i=1}^{d} \frac{(\mathbf{x}-\mathbf{y})_i^2}{(\mathbf{x} - \mathbf{y} + s \mathbf{y})_i} \\
        &
        = \frac{s^2 + 1}{s^2}
        \sum_{i=1}^{d}
        \frac{(\mathbf{x}-\mathbf{y})_i^2}{ \left( \mathbf{y} + \frac{\mathbf{x}-\mathbf{y}}{s} \right)_i }  = \frac{s^2 + 1}{s^2}\| \rho-\gamma \|^2_{\gamma+\frac{\rho-\gamma}{s}}
    \end{align*}
    This gives us \eqref{eqn:classicspecialrel}. Further, since $D_f^\mathrm{std}$ is fully determined by $D_f^\mathrm{cl}$, it is uniquely defined by the relationship on commuting states.
\end{proof}

Propositions~\ref{prop:MaxStdDivUnique} is in fact a consequence of Theorem~\ref{thm:classicspecialrel}. This theorem supports the hypothesis that this type of integral relation between the standard $f$-divergences and Riemannian semi-norms is fundamental for a generic equality between their respective relative expansion coefficients. We conclude this subsection with the following result, based on \cite{Cohen1993,Choi1994}.

\begin{theorem}\label{thm6}
        (Relative Expansion Coefficients of QC Channels) 
        Let $\mathcal{N}: \mathcal{B}(\mathcal{H}_A) \to \mathcal{B}(\mathcal{H}_B)$, $\mathcal{M}: \mathcal{B}(\mathcal{H}_A) \to \mathcal{B}(\mathcal{H}_B')$ be QC quantum channels. Then for all $f \in \mathcal{F}$, $\kappa \in \mathcal{K}$,
        \[
        \widecheck{\eta}_f^{\mathrm{std}}(\mathcal{N}, \mathcal{M}) = \widecheck{\eta}_\kappa^{\mathrm{Riem}}(\mathcal{N}, \mathcal{M})
        \]
\end{theorem}

\begin{proof}
    $ $\newline
    Since we generally have $\widecheck\eta_f^\mathrm{std}(\mathcal{N},\mathcal{M})\leq \widecheck\eta_{\kappa_f}^\mathrm{Riem}(\mathcal{N},\mathcal{M})$, and for QC channels, $\widecheck\eta_{\kappa_f}^\mathrm{Riem}(\mathcal{N},\mathcal{M})\equiv\widecheck\eta_{\kappa}^\mathrm{Riem}(\mathcal{N},\mathcal{M})$, we only need to check $\widecheck\eta_f^\mathrm{std}(\mathcal{N},\mathcal{M})\geq \widecheck\eta_{\kappa}^\mathrm{Riem}(\mathcal{N},\mathcal{M})$

    \noindent By definition of QC, there exist orthonormal bases $\{ \ket{\psi_i} \}_{i=1}^{d_B}$ of $\mathcal{H}_B$, $\{ \ket{\varphi_i} \}_{i=1}^{d'_B}$ of $\mathcal{H}'_B$, and POVMs $\{ F_i\}_{i=1}^{d_B} \subseteq \mathcal{B}(\mathcal{H}_B)$, $\{G_i\}_{i=1}^{d_B'} \subseteq \mathcal{B}(\mathcal{H}'_B)$ such that for all $\rho \in\mathcal{D}(\mathcal{H}_A)$:
        
    \[
    \mathcal{N}(\rho) = \sum_{i=1}^{d_B} \mathrm{Tr} \, F_i \rho \, \ket{\psi_i}\bra{\psi_i}, 
    \quad 
    \mathcal{M}(\rho) = \sum_{i=1}^{d'_{B}} \mathrm{Tr} \, G_i \rho \, \ket{\varphi_i}\bra{\varphi_i}.
    \]
        
    Let $
    \rho\neq \gamma \in \mathcal{D}(\mathcal{H}_A), \, \mathrm{supp}\, \rho = \mathrm{supp}\, \gamma, X:=\rho-\gamma\in T_\gamma \mathcal{D}(\mathcal{H}_A)$.\\
    
    As a result, using Theorem~\ref{thm:classicspecialrel}:
    \begin{align*}
        D_f^{\mathrm{std}} \left( \mathcal{N}(\rho) \| \mathcal{N}(\gamma) \right) 
        &= c \, D_{(w-1)^2}^{\mathrm{std}} \left( \mathcal{N}(\rho) \| \mathcal{N}(\gamma) \right) 
        + \int_{(1,\infty)} D_{g_s}^{\mathrm{std}} \left( \mathcal{N}(\rho) \| \mathcal{N}(\gamma) \right) \, d\mu(s) \\
        &= c \, \| \mathcal{N}(X) \|_{\kappa,\:\mathcal{N}(\gamma)}^2 + \int_{(1,\infty)} \frac{s^2 + 1}{s} \| \mathcal{N}(X) \|_{\kappa,\:\mathcal{N} \left( \gamma + \frac{X}{s} \right)}^2 d\mu(s) \\
        &\geq \widecheck{\eta}_\kappa^{\mathrm{Riem}} (\mathcal{N}, \mathcal{M}) \left( c \, \| \mathcal{M}(X) \|_{\mathcal{M}(\gamma)}^2 + \int_{(1,\infty)} \frac{s^2 + 1}{s} \| \mathcal{M}(X) \|_{\mathcal{M} \left( \gamma + \frac{X}{s} \right)}^2 d\mu(s) \right) \\
        &= \widecheck{\eta}_\kappa^{\mathrm{Riem}} (\mathcal{N}, \mathcal{M}) \, D_f^{\mathrm{std}} \left( \mathcal{M}(\rho) \| \mathcal{M}(\gamma) \right)
    \end{align*}
    \noindent Where, in the inequality, we applied the definition of $\widecheck\eta_\kappa^\mathrm{Riem}(\mathcal{N},\mathcal{M})$. Therefore,
    \[
    \widecheck{\eta}_f^{\mathrm{std}} (\mathcal{N}, \mathcal{M}) = \inf_{\substack{\rho \neq \gamma \in \mathcal{D}_d, \\ \mathrm{supp} \, \rho = \mathrm{supp} \, \gamma}} 
    \frac{ D_f^{\mathrm{std}} \left( \mathcal{N}(\rho) \| \mathcal{N}(\gamma) \right) }{D_f^{\mathrm{std}} \left( \mathcal{M}(\rho) \| \mathcal{M}(\gamma)\right)}
    \geq \widecheck{\eta}_\kappa^{\mathrm{Riem}} (\mathcal{N}, \mathcal{M}).
    \]
    This was the direction that we needed, so we are finished with this proof.
\end{proof}

\section{Equivalence}\label{sec:equivalence}
As we cannot currently derive expressions for the relative expansion coefficients, the primary question that we would like to address regarding expansion coefficients (over a restricted domain) is about whether they are strictly positive or zero. We know from Theorem~\ref{nordpi} that the strict positivity of expansion coefficients is non-trivial, but by finding cases when we indeed have this (Section~\ref{sec:explicitcoeffs}), we obtain demonstrations that information about the input states is preserved by the channels $\mathcal{N}$. This can have implications to the recoverability and convergence rate of quantum channels (see Section~\ref{sec:applications}).

In this section, we establish a notion of \emph{equivalence} of relative expansion coefficients as an alternative to equality (which was the focus of Section~\ref{sec:equalityofcoeffs}), to make significant progress in understanding the positivity of these coefficients. In building several equivalence and inequivalence results, we learn more about how the properties of standard $f$-divergences and Riemannian semi-norms (such as boundedness) influence the relative expansion coefficients.

\subsection{The Notion of Equivalence and its Inheritability}\label{sec:inherit}
A takeaway from Section~\ref{subsec:equalitycases} is that it is expected to be difficult and rare to find such convenient relationships in the non-commutative setting between a standard $f$-divergence, $D_f^\mathrm{std}(\rho\|\gamma)$, and the corresponding Riemannian semi-norm, $\|\rho-\gamma\|_{\kappa_f,\:\rho}^2$, to have $\widecheck\eta_f^\mathrm{std}(\mathcal{N},\mathcal{M})=\widecheck\eta_{\kappa_f}^\mathrm{Riem}(\mathcal{N},\mathcal{M})$ for arbitrary pairs of channels $\mathcal{N},\mathcal{M}$. We care about this because it is easier to deal with and lower bound $\widecheck\eta_{\kappa_f}^\mathrm{Riem}(\mathcal{N},\mathcal{M})$. But it suffices for checking positivity to look for cases where, for general pairs of channels $\mathcal{N},\mathcal{M}$,
\[\widecheck\eta_f^\mathrm{std}(\mathcal{N},\mathcal{M})\geq \alpha\widecheck\eta_{\kappa_f}^\mathrm{Riem}(\mathcal{N},\mathcal{M})\]
for a channel-independent constant $\alpha\in(0,1)$; indeed, $\widecheck\eta_f^\mathrm{std}(\mathcal{N},\mathcal{M})\leq\widecheck\eta_{\kappa_f}^\mathrm{Riem}(\mathcal{N},\mathcal{M})$ (Proposition~\ref{prop:RiemDivIneq}), so we only have to consider $\alpha<1$. This is a looser requirement that allows for a much broader variety of $f\in\mathcal{F}$ to be constructed with this property. This is because of the \emph{inheritance} property that we will now meet. In preparation for this discussion to be taken much further in later sections, we first define a notion of equivalence of two relative expansion coefficients:
\begin{definition}[Equivalence]
    Suppose we restrict to some class $\mathcal{Q}$ of quantum channels. Provided relative expansion coefficients $\widecheck\eta,\widecheck\eta'$, each corresponding to respective distinguishability measures, they are equivalent w.r.t. constants $0<\alpha\leq\beta$, denoted $\widecheck\eta\cong_{\alpha,\beta}\widecheck\eta'$ (or simply $\widecheck\eta\cong\widecheck\eta'$ if $\alpha,\beta$ are kept implicit), if:
    \begin{equation*}
        \alpha\leq\frac{\widecheck\eta(\mathcal{N},\mathcal{M})}{\widecheck\eta'(\mathcal{N},\mathcal{M})}\leq\beta
    \end{equation*}
    for all quantum channels $\mathcal{N}:\mathcal{B}(\mathcal{H}_A)\to\mathcal{B}(\mathcal{H}_B),\mathcal{M}:\mathcal{B}(\mathcal{H}_A)\to\mathcal{B}(\mathcal{H}_B')\:\in\mathcal{Q}$. The case $\alpha=\beta=1$ corresponds to the generic equality $\widecheck\eta\equiv\widecheck\eta'$ over all pairs of channels in $\mathcal{Q}$.
\end{definition}

As illustrated in Section~\ref{sec:qc}, where we saw the effect of considering only QC channels, there is merit in identifying subclasses $\mathcal{Q}$ of quantum channels that allow for equivalence; in Section~\ref{sec:fullrankoutputchannels} we continue this discussion with full-rank output channels, but for these first two subsections we find cases of equivalence for $\mathcal{Q}$ the set of \emph{all} quantum channels. The following theorem offers a simple construction of $f\in\mathcal{F}, \kappa_f \in \mathcal{K}$ such that $\widecheck{\eta}_f^{\mathrm{std}} \cong_{\alpha,1} \widecheck{\eta}_{\kappa_f}^{\mathrm{Riem}}$ from the cases of generic equality that we have (recall from Theorem~\ref{thm:equalitycases}). 

\begin{theorem}[Inheritance of Equivalence]\label{thm:inherit}
	Provided two functions $f,g\in\mathcal{F}$, corresponding to $\kappa_f,\kappa_g \in \mathcal{K}$ respectively, wlog $f''(1)=g''(1)=2$, that satisfy the following relationship for some $0< a < b$ (necessarily $a\leq1\leq b$):
	\[
	a\,f(x) \leq g(x) \leq b\,f(x) \quad \text{for all } x\in(0,\infty),
	\]
    If we suppose further that 
    \begin{equation}\label{eqn:fatherequiv}
        \widecheck{\eta}_f^{\mathrm{std}} \cong_{\gamma,\delta} \widecheck{\eta}_{\kappa_f}^{\mathrm{Riem}} \text{ for some } 0<\gamma\leq\delta\leq1
    \end{equation}
    Then we can conclude
    \begin{equation}\label{eqn:daughterequiv}
    \widecheck{\eta}_f^{\mathrm{std}} \cong_{\frac{a}{b},\frac{b}{a}}\widecheck{\eta}_{g}^{\mathrm{std}} \cong_{\alpha,\beta} \widecheck{\eta}^{\mathrm{Riem}}_{\kappa_g}\cong_{\frac{a}{b},\frac{b}{a}}\widecheck{\eta}_{\kappa_f}^{\mathrm{Riem}}
    \end{equation}
	for $\alpha = \frac{a^2\gamma}{b^2}, \beta=\min\{\frac{b^2\delta}{a^2},1\}$.
\end{theorem}

\begin{proof} Since 
    \[a\,f(x) \leq g(x) \leq b\,f(x)\implies a\,\kappa_f(x) \leq \kappa_g(x) \leq b\,\kappa_f(x),
    \]
    we also have by the definitions \eqref{eqn:stdfdiv}, \eqref{eqn:riemdiv} respectively:
    \begin{equation}\label{eqn:induceddivrelations}
        a D_f^{\mathrm{std}}(\rho\|\gamma) \leq D_g^{\mathrm{std}}(\rho\|\gamma) \leq b D_f^{\mathrm{std}}(\rho\|\gamma) \text{ and }a\|X\|^2_{\kappa_f,\rho} \leq \|X\|^2_{\kappa_g,\rho} \leq b\|X\|^2_{\kappa_f,\rho}
    \end{equation}
    for any traceless Hermitian $X$, density operators $\rho,\gamma$, $\mathrm{supp}\;\rho= \mathrm{supp}\;\gamma$, $X,\rho,\gamma \in \mathcal{B}(\mathcal{H})$ for some Hilbert space $\mathcal{H}$. \\
    
    \noindent Therefore, by applying \eqref{eqn:induceddivrelations} to the definitions of the relative expansion coefficients, and using \eqref{eqn:fatherequiv}, we deduce for all pairs of quantum channels $\mathcal{N}:\mathcal{B}(\mathcal{H}_A)\to\mathcal{B}(\mathcal{H}_B),\mathcal{M}:\mathcal{B}(\mathcal{H}_A)\to\mathcal{B}(\mathcal{H}_B')$:
    \begin{equation*}
        \widecheck{\eta}^{\mathrm{std}}_g(\mathcal{N},\mathcal{M}) \geq \frac{a}{b}\widecheck{\eta}^{\mathrm{std}}_f(\mathcal{N},\mathcal{M}) \geq \frac{a\gamma}{b}\widecheck{\eta}^{\mathrm{Riem}}_{\kappa_f}(\mathcal{N},\mathcal{M}) \geq \frac{a^2\gamma}{b^2} \widecheck{\eta}^{\mathrm{Riem}}_{\kappa_g}(\mathcal{N},\mathcal{M}),
    \end{equation*}
    so then $\alpha = a^2\gamma/b^2$. 
    Similarly,
    \begin{equation*}
        \widecheck{\eta}^{\mathrm{std}}_g(\mathcal{N},\mathcal{M}) \leq \frac{b}{a}\widecheck{\eta}^{\mathrm{std}}_f(\mathcal{N},\mathcal{M}) \leq \frac{b\delta}{a}\widecheck{\eta}^{\mathrm{Riem}}_{\kappa_f}(\mathcal{N},\mathcal{M}) \leq \frac{b^2\delta}{a^2} \widecheck{\eta}^{\mathrm{Riem}}_{\kappa_g}(\mathcal{N},\mathcal{M}),
    \end{equation*}
    By Proposition~\ref{prop:RiemDivIneq}, we can choose $\beta=\min\{\frac{b^2\delta}{a^2},1\}$. Combining these inequalities, we have \eqref{eqn:daughterequiv}.
\end{proof}

To understand how this theorem can be used to construct equivalence, suppose we are provided $f\in\mathcal{F}_\mathrm{sym}$ such that $\widecheck\eta_f^\mathrm{std}=\widecheck\eta_{\kappa_f}^\mathrm{Riem}$, and any $\kappa'\in\mathcal{K}$ such that $\kappa'(x)\leq c\kappa_f(x)\;\forall x\in(0,\infty)$, for $c\geq1$ some constant. From this, we can define $\kappa\in\mathcal{K}$ by $\kappa(x):=\lambda\,\kappa_f(x)+(1-\lambda)\,\kappa'(x)\equiv\kappa_g(x),\lambda\in(0,1)$ for $g(x):=(x-1)^2\kappa(x)\in\mathcal{F}_\mathrm{sym}$. By Theorem~\ref{thm:inherit}:
\begin{equation*}
    \frac{\lambda^2}{(\lambda+(1-\lambda)c)^2}\widecheck\eta_{\kappa_g}^\mathrm{Riem}\leq\widecheck\eta_g^\mathrm{std}\leq\widecheck\eta_{\kappa_g}^\mathrm{Riem}
\end{equation*}

For $\lambda\approx1$, we may then notice that the divergence and Riemannian coefficients remain similar:
\begin{equation*}
    \widecheck\eta_g^\mathrm{std}\approx\widecheck\eta_{\kappa_g}^\mathrm{Riem}.
\end{equation*}
A key benefit of considering equivalence over generic equality, is that this inheritance property of equivalence allows us to say more about the connections between different standard $f$-divergences and their induced Riemannian semi-norms, as we are no longer limited to the cases in Theorem~\ref{thm:equalitycases}. In contrast, for the case of generic equality, it is not clear whether even all of the conical combinations of the equality cases yield generic equality.

% \textbf{Open Problem:} \emph{Do we have $\widecheck\eta_g^\mathrm{std}\equiv\widecheck\eta_{\kappa_g}^\mathrm{Riem}$ for any $g(x)=\alpha\cdot\frac{(x-1)^2(x+1)}{2x} +\beta\cdot (x-1)\log x,\;\alpha,\beta>0,\alpha\neq \beta$?}\\

Recall that for all $\kappa\in\mathcal{K}$, $\kappa_{\min}(x)\leq\kappa(x)\leq\kappa_{\max}(x)\;\forall x\in(0,\infty)$. As a result, there are two especially useful applications for Theorem~\ref{thm:inherit}: 
\begin{enumerate}
    \item We can always choose $\kappa_g(x):=\lambda\kappa_{\mathrm{max}}(x)+(1-\lambda)\kappa'(x),g\in\mathcal{F}_\mathrm{sym}$ for any $\kappa'\in\mathcal{K},\lambda=0$.\\ In this case, we take $\kappa_f=\kappa_{\max},c=1$. 
    \item We can always choose $\kappa_g(x):=\lambda\,\kappa_f(x)+(1-\lambda)\,\kappa_{\min}(x),g\in\mathcal{F}_\mathrm{sym}$ for any $\kappa_f\in\mathcal{K}$, $f(x)=(x-1)^2\kappa_f(x)\in\mathcal{F}_\mathrm{sym}$ such that $\widecheck\eta_f^\mathrm{std}=\widecheck\eta_{\kappa_f}^\mathrm{Riem}$. \\
    In this case, we take $\kappa'=\kappa_{\min},c=1$.
\end{enumerate}

To elaborate on the second point, since $\kappa_s$ is uniformly decreasing in $s\in[0,1]$ (recall $\kappa_{\min}:=\kappa_0,\kappa_{\max}:=\kappa_1$), this suggests we can take $\kappa' = \kappa_s$ for $s\in(0,1]$ sufficiently large. There are also some remarks we can make about the first point. First of all, we deduce that for 
\begin{align*}
    g(x):=\:&\alpha\cdot((x-1)^2)_\mathrm{sym}+\beta\cdot(x\log x)_\mathrm{sym}=\alpha\cdot\frac{(x-1)^2(x+1)}{2x}+\beta\cdot(x-1)\log x\in\mathcal{F}_\mathrm{sym},\alpha,\beta>0,\\
\kappa_g(x)\equiv\:&\lambda\,\kappa_{\max}(x)+(1-\lambda)\,\kappa_\mathrm{BKM}(x),\;\lambda=\frac{\alpha}{\alpha+\beta}\in(0,1),
\end{align*}
We obtain the following equivalence: 
\begin{equation*}
    \widecheck\eta_g^\mathrm{std}\cong_{\lambda^2,1}\widecheck\eta_{\kappa_g}^\mathrm{Riem}.
\end{equation*}
It is unknown whether $\widecheck\eta_g^\mathrm{std}\equiv\widecheck\eta_{\kappa_g}^\mathrm{Riem}$, so this equivalence offers progress towards the open problem, because we can relate the divergence and Riemannian relative expansion coefficients.

We can also recall the integral representation \eqref{opmon} for $\kappa \in \mathcal{K}$, where $m$ is the unique probability measure on $[0,1]$ corresponding to $\kappa$:
\begin{equation}
	\kappa(x) = \int_{[0,1]} \kappa_s(x)\,dm(s), \qquad 
	\kappa_s(x)= \frac{1+s}{2} \left( \frac{1}{x+s} + \frac{1}{s x + 1} \right)
\end{equation}

\noindent
$\kappa_{\max}(x) := \kappa_0(x)$ is the only $\kappa_s(x)$ that diverges as $x\to 0^+$. Rewriting
\begin{equation}
	\kappa(x) = m(0)\kappa_{\max}(x)+\int_{(0,1]} \kappa_s(x)\,dm(s)\equiv m(0)\kappa_{\max}(x)+(1-m(0))\kappa'(x),
\end{equation}
$m(0)>0$ certainly ensures that $\kappa$ is unbounded and that \[\widecheck\eta_{(x-1)^2\kappa(x)}^\mathrm{std}\cong_{m(0)^2,\:1}\widecheck\eta_{\kappa}^\mathrm{Riem}.\] 
This establishes the importance of $\kappa_{\max}$ in constructing cases of equivalence, but it does not fully explain the apparent connection between $\kappa_f(0^+)=\infty$ and $\widecheck{\eta}^{\mathrm{std}}_f \cong \widecheck{\eta}^{\mathrm{std}}_{\kappa_f}$. 

A problem remains that these properties can coexist even while $m(0)=0$; it is only the concentration of $m(s)$ near $s=0$ that causes the divergence of $\kappa(x)$ as $x\to0^+$. That is, the relative entropy doesn't inherit its equivalence ($\widecheck\eta_{x\log x}^\mathrm{std}=\widecheck\eta_{\kappa_\mathrm{BKM}}^\mathrm{Riem}$) from the maximal standard $f$-divergence, $D_{(x-1)^2}^\mathrm{std}$; for $\kappa=\kappa_\mathrm{BKM}$, this is clear from $m(0)=\lim_{x\to\infty} 2\kappa_\mathrm{BKM}(x)=\lim_{x\to\infty} \frac{2\log x}{x-1}=0$. In fact, the integral representation is
\[
\kappa_{\mathrm{BKM}}(x) = \frac{\log x}{x-1} = \int_0^1 \kappa_s(x)\cdot \frac{2}{(1+s)^2} ds
\]

\noindent We also cannot yet guarantee that $\kappa_f(0^+)=\infty \implies \widecheck{\eta}^{\mathrm{std}}_{f_\mathrm{sym}} \cong \widecheck{\eta}^{\mathrm{Riem}}_{\kappa_f}$. The next section, Section~\ref{sec:bddvsunbdd}, will explore the connections between equivalences and boundedness further, by looking into what happens in the bounded setting, $\kappa_f(0^+)<\infty$.

\subsection{Relative Expansion Coefficients in the Bounded Setting (vs. Unbounded Settings)}\label{sec:bddvsunbdd}
Section~\ref{subsec:equalitycases} demonstrated more cases where $\widecheck\eta_{g}^\mathrm{std}\equiv\widecheck\eta_{\kappa_g}^\mathrm{Riem}$ than were previously noticed, and Section~\ref{sec:inherit} showed that we could construct even more cases with the equivalence: $\widecheck\eta_{g}^\mathrm{std}\cong\widecheck\eta_{\kappa_g}^\mathrm{Riem}$. However, all of these constructions involve unbounded $\kappa_g$ (i.e. $\kappa_g(0^+)=\infty$), which they inherited from the generic equality cases in Theorem~\ref{thm:equalitycases}. Note that unbounded $\kappa_g$ is the same as saying that $D_g^\mathrm{std}$ is unbounded \cite{Hiai_2011,HM17}. 

There are a couple intuitive reasons for believing that this is no coincidence, and that in the setting where $\kappa_g(0^+)<\infty$, we should expect that $\widecheck\eta_{g}^\mathrm{std}\ncong\widecheck\eta_{\kappa_g}^\mathrm{Riem}$. This was first suggested by the result \eqref{eqn:equalitycounter} from \cite{HR15}, which reveals a disparity between $\widecheck\eta_{(x-1)^2\kappa_s(x)}^\mathrm{std}$ and $\widecheck\eta_{\kappa_s(x)}^\mathrm{Riem}$ for sufficiently large $s\in(0,1]$. We infer that in the bounded setting, $\kappa_g(0^+)<\infty$, sometimes the divergence relative expansion coefficients have greater dependence on the global behaviour of $D_g^\mathrm{std}$. One may further hypothesise that the asymmetry in the boundedness properties of the standard $f$-divergence versus it's induced Riemannian semi-norm makes it difficult for a relationship similar to the forms \eqref{eqn:specialrelationships} to exist in the quantum setting. That is, when $\kappa_g$ is bounded, $D_g^\mathrm{std}(\rho\|\gamma)$ is bounded over all states $\rho,\gamma\in\mathcal{D}(\mathcal{H})$, while $\chi_{\kappa_g}^2(\rho\|\gamma)=\|\rho-\gamma\|_{\kappa_g,\gamma}^2$ is not.  Due to the importance of such integral relations, this would make the generic equality $\widecheck\eta_{g}^\mathrm{std}\equiv\widecheck\eta_{\kappa_g}^\mathrm{Riem}$ unlikely --- and if so, constructing cases of equivalence in the bounded $\kappa_g$ setting may not be possible. On the other hand, one should bear in mind that in the classical setting, we have these integral relations  for all $f\in\mathcal{F}$, despite such an asymmetry. 

% relied on a principle that, for $f,g\in\mathcal{F}$, the relative expansion coefficients $\widecheck\eta_f^\mathrm{std},\widecheck\eta_g^\mathrm{std}$ and $\widecheck\eta_{\kappa_f}^\mathrm{Riem},\widecheck\eta_{\kappa_g}^\mathrm{Riem}$ are respectively equivalent when $f(x),g(x)$ have similar behaviour as $x\to0^+$ and $x\to\infty$. 

Instead, when $\kappa_f(x),\kappa_g(x)$ are bounded, we find in the following two results that the equivalences $\widecheck\eta_f^\mathrm{std}\cong\widecheck\eta_g^\mathrm{std}$, $\widecheck\eta_{\kappa_f}^\mathrm{Riem}\cong\widecheck\eta_{\kappa_g}^\mathrm{Riem}$ hold in general; again, we consider $\mathcal{Q}$ to be the set of all channels. We then conclude this subsection with a proof of the inequivalence, $\widecheck\eta_{\kappa_f}^\mathrm{Riem}\ncong\widecheck\eta_{\kappa_g}^\mathrm{Riem}$, for bounded $\kappa_f(x)$ and unbounded $\kappa_g(x)$, which illustrates the significance of comparing these distinguishability measures via boundedness.
% The following equivalence results provide a sense in which the respective distinguishability measures (standard $f$-divergences and Riemannian semi-norms) are all related in the bounded $\kappa$ setting:
\begin{theorem}[Equivalence of Bounded Standard Divergences]\label{thm:bddstddivEquiv}
    Given any $f,g\in\mathcal{F}$, corresponding to $\kappa_f,\kappa_{g}\in\mathcal{K}$ such that $\kappa_f(0^+),\kappa_{g}(0^+)<\infty$, then there exist $0<\alpha<\beta$ s.t. for all $\rho\neq\gamma\in\mathcal{D}_d$,
    \begin{equation*}
        \alpha\leq \frac{D_f^\mathrm{std}(\rho\|\gamma)}{D_{g}^\mathrm{std}(\rho\|\gamma)}\leq\beta.
    \end{equation*}
    In particular, this implies that \begin{equation}\label{eqn:bddstddivEquiv}
        \widecheck\eta_f^\mathrm{std}\cong_{\frac{\alpha}{\beta},\frac{\beta}{\alpha}}\widecheck\eta_{g}^\mathrm{std}.
    \end{equation}
\end{theorem}
\begin{proof}
    wlog let $f'(1)=g'(1)=0$ (since $D_f^\mathrm{std}(\rho\|\gamma)\equiv D_{f-f'(1)(x-1)}^\mathrm{std}(\rho\|\gamma)$). By the integral representation of operator convex functions \eqref{opcon}, this ensures that $f'(\infty), \widetilde{f}'(\infty) > 0$, where we recall that $\widetilde{f}$ is the transpose of $f$.
    
    \noindent Since $\kappa_f$ is bounded (i.e. $\kappa_f(0^+) < \infty$), and defining $f_{\mathrm{sym}}(x) := (x - 1)^2 \kappa_f(x)$:
    \begin{align*}
        f_{\mathrm{sym}}'(\infty)
        &:= \lim_{x \to \infty} \frac{f_{\mathrm{sym}}(x)}{x}= \frac{f'(\infty) + \widetilde{f}'(\infty)}{f''(1)} \\
        &= \lim_{x \to \infty} \frac{(x-1)^2 \kappa_f(x)}{x}= \lim_{x \to \infty} x \kappa_f(x)=\lim_{x \to \infty} \kappa_f(x^{-1}) \\
        &= \kappa_f(0^+) < \infty
    \end{align*}
    
    \noindent
    $\Rightarrow\; f'(\infty)$ and $\widetilde{f}'(\infty) \equiv f(0^+)$ are finite.
    % , and thus by \cite[Proposition 3.8]{HM17}, $D^{\mathrm{std}}_f$ is continuous (and finite)
    Similarly, we can conclude the same for $g$.
    \\
    
    \noindent Note that: $\lim_{x \to 1} \frac{f(x)}{g(x)} = \lim_{x \to 1} \frac{f''(x)}{g''(x)} =\frac{f''(1)}{g''(1)}\in(0,\infty)$ by L'Hôpital's rule. 
    \\\\ Therefore by continuity, and the fact $f,g\in\mathcal{F}$ are strictly positive for $x\neq 1$, $\frac{f(x)}{g(x)}$ is bounded, \\
    i.e. $\exists 0<\alpha<\beta$ s.t. $\alpha\leq\frac{f(x)}{g(x)}\leq\beta$ for all $x\in(0,\infty)$. 
    
    By recalling the definition of standard $f$-divergence \eqref{eqn:stdfdiv}, this implies that 
    \[
    \alpha\leq\frac{D_f^\mathrm{std}(\rho\|\gamma)}{D_g^\mathrm{std}(\rho\|\gamma)}\leq\beta\;\text{ for all states $\rho\neq\gamma\in\mathcal{D}_d$}
    \]
    \noindent And finally, applying this inequality to the definition of the divergence relative expansion coefficients, gives \eqref{eqn:bddstddivEquiv}.
\end{proof}

\begin{theorem}[Equivalence of Bounded Metrics]\label{thm:bddRiemEquiv}
	$ $\newline 
    Any $\kappa_f, \kappa_g \in \mathcal{K}$ s.t. $\kappa_f(0^+),\kappa_g(0^+)<\infty$, corresponding to $f,g\in\mathcal{F}$, satisfy the following inequality,
	\[
    \alpha\leq\frac{\kappa_f(x)}{\kappa_g(x)}\leq\beta,\quad \forall x\in(0,\infty),
	\]
	% As a consequence $\|X\|^2_{\kappa,\:\rho} \cong_{\alpha,\beta}  \|X\|^2_{\kappa',\:\rho}$ for all traceless Hermitian $X$, density operators $\rho$, $\rho,X \in \mathcal{B}(\mathcal{H})$, some finite-dimensional Hilbert space $\mathcal{H}$.\\
	where $\alpha := \frac{2}{\kappa_f(0^+)}\leq1, \ \beta := \frac{\kappa_g(0^+)}{2}\geq 1$. As a result, this implies that 
    \begin{equation}\label{eqn:bddRiemEquiv}
	\widecheck{\eta}^{\mathrm{Riem}}_{\kappa_f} \cong_{\frac{\alpha}{\beta}, \frac{\beta}{\alpha}} \widecheck{\eta}^{\mathrm{Riem}}_{\kappa_g}.
    \end{equation}
\end{theorem}
\begin{proof}
    $ $\newline
    Let $s, s' \in (0, 1]$, so that the corresponding $\kappa_s(x), \kappa_{s'}(x)$ are bounded (i.e. $\kappa_s(0^+), \kappa_{s'}(0^+) < \infty$).
    \begin{align*}
        \frac{\kappa_s(x)}{\kappa_{s'}(x)}
        &= \frac{
            \frac{(1 + s)^2}{2} \cdot \frac{1 + x}{(x+s)(1+sx)}
        }{
            \frac{(1+s')^2}{2}\cdot\frac{1+x}{(x+s')(1+s'x)}
        } 
        = \frac{s'(1+s)^2}{s(1+s')^2}
        \cdot \frac{(x+s')(x+s'^{-1})}{(x+s)(x+s^{-1})} \\
        &= \frac{s'(1+s)^2}{s(1+s')^2} \cdot
        \left(
        1 +
        \left[
        (s' + s'^{-1}) - (s + s^{-1})
        \right]\cdot
        \frac{x}{x^2 + 1 + (s + s^{-1})x}
        \right)
    \end{align*}
    Since \[\frac{\dd}{\dd x}\frac{\kappa_s(x)}{\kappa_{s'}(x)} = \frac{s'(1 + s)^2}{s(1 + s')^2}\cdot \left[(s'+s'^{-1})-(s+s^{-1})\right]\cdot\frac{1-x^2}{x^2+1+(s+s^{-1})x},\]
    
    \noindent We're interested in taking $s' = 1$ ($\kappa_\mathrm{min}(x):=\kappa_1(x)$), so necessarily $s \leq s'$.
    
    $\frac{\kappa_s(x)}{\kappa_{\min}(x)}$ is decreasing over $x \in (0,1)$ and it is increasing over $x \in (1,\infty)$. This means that it has a minimum value $\frac{\kappa_s(1)}{\kappa_{\min}(1)}=1$, and a maximum value either as $x\to 0^+$ or as $x\to\infty$. Using the fact $\kappa_s(x^{-1})=x\kappa_s(x)$ and $\kappa_{\min}(0^+)=2$, its maximum value is:  \\
    \[
    \lim_{x \to \infty} \frac{\kappa_s(x)}{\kappa_{\min}(x)} = \lim_{x \to \infty} \frac{x\kappa_s(x)}{x\kappa_{\min}(x)} = \lim_{x \to \infty}\frac{\kappa_s(x^{-1})}{\kappa_{\min}(x^{-1})} = \frac{\kappa_s(0^+)}{\kappa_{\min}(0^+)} = \frac{\kappa_s(0^+)}{2}
    \]
    Therefore,
    \begin{equation*}
        \kappa_\mathrm{min}(x) \leq \kappa_s(x) \leq \frac{\kappa_s(0^+)}{2} \cdot \kappa_\mathrm{min}(x)\text{ for all }s\in[0,1]
    \end{equation*}
    
    \noindent By considering the integral representation \eqref{opmon} that expresses any $\kappa_f(x)\in \mc K$ in terms of the extreme points $\kappa_s(x)$, we can extend this inequality to general $\kappa_f(x),\kappa_g(x)\in\mathcal{K}$ and use this to bound $\frac{\kappa_f(x)}{\kappa_g(x)}$:
    \begin{align}
        & \kappa_\mathrm{min}(x) \leq \kappa_f(x) \leq \frac{\kappa_f(0^+)}{2} \, \kappa_\mathrm{min}(x), \;\kappa_\mathrm{min}(x) \leq \kappa_g(x) \leq \frac{\kappa_g(0^+)}{2} \, \kappa_\mathrm{min}(x)\notag\\
        &\implies  \frac{2}{\kappa_g(0^+)} \leq \frac{\kappa_f(x)}{\kappa_g(x)} \leq \frac{\kappa_f(0^+)}{2}\label{eqn:kappakappa}
    \end{align}
    
    By applying \eqref{eqn:kappakappa} to the definition of the Riemannian semi-norm \eqref{eqn:riemdiv}, we obtain  
    \[
    \alpha<\frac{\|\rho-\gamma\|^2_{\kappa_f,\:\rho}}{\|\rho-\gamma\|^2_{\kappa_g,\:\rho}}<\beta
    \]
    for all $\rho,\gamma \in \mathcal{D}_d$ with $\mathrm{supp}\:\rho=\mathrm{supp}\:\gamma$, and $\alpha:=\frac{2}{\kappa_g(0^+)},\beta:=\frac{\kappa_f(0^+)}{2}$. \\
    
    Finally, applying this inequality to the definition of Riemannian relative expansion coefficients, gives \eqref{eqn:bddRiemEquiv}.
\end{proof}

% \subsection{Bounded-Case Riemannian Relative Expansion Coefficients are their own Equivalence Class}

% The point of this section is to review the proof of an important result in \cite{HR15} that demonstrates that the Riemannian and divergence coefficients don't have to coincide. To better understand the motivation for this result, some extra elaboration is included in the proof of the first assertion. 

Theorem~\ref{thm:bddstddivEquiv} and Theorem~\ref{thm:bddRiemEquiv} demonstrate that, even if the equivalence $\widecheck\eta_f^\mathrm{std}\cong\widecheck\eta_{\kappa_f}^\mathrm{Riem}$ is not possible, we still have a significant reduction in the problem of determining positive expansion coefficients. 

Let us now revisit the family of CQ qubit channels $\Phi_{\alpha,\tau},\:\alpha^2+\tau^2\leq 1, \alpha,\tau\in\mathbb{R}$ that provided the first counterexample of $\eta_f^\mathrm{std}\equiv\eta_{\kappa_f}^\mathrm{Riem}$ \cite{LR99,HR15}. This is expressed in the Bloch vector representation as
\begin{equation}\label{eqn:counter}
    \Phi_{\alpha,\tau}:\frac{1}{2}(I+\mathbf{w}\cdot\sigma)\mapsto\frac{1}{2}(I+\alpha w_1 \sigma_1+\tau \sigma_3),\;\alpha^2+\tau^2\leq 1, \alpha,\tau\in\mathbb{R}.
\end{equation}
For a CQ channel $\Phi:\mathcal{B}(\mathcal{H}_A)\to\mathcal{B}(\mathcal{H}_B)$, by definition, the image of its adjoint map, $\mathrm{Im}\:\hat\Phi$, is a commutative subalgebra of $\mathcal{B}(\mathcal{H}_A)$. \cite[Proposition 5.5]{HR15} states that if $\kappa_f(x)\leq\kappa_g(x)$ (corresponding to $f,g\in\mathcal{F}$) for all $x\in(0,\infty)$, then:
\[
\eta_{\kappa_f}^\mathrm{Riem}(\Phi)\leq\eta_{\kappa_g}^\mathrm{Riem}(\Phi)\text{ for every CQ channel }\Phi.
\]
We therefore anticipate that an unbounded $\kappa_g$ should have a larger contraction coefficient than $\kappa_f$ for the CQ channel $\Phi_{\alpha,\tau}$. From the explicit cases calculated for $\eta_\kappa(\Phi_{\alpha,\tau})$ by \cite{HR15}, it seemed possible that perhaps any such pair $\eta_{\kappa_f}(\Phi_{\alpha,\tau}),\eta_{\kappa_g}(\Phi_{\alpha,\tau})$ could be many orders of magnitude apart. The following result finally verifies the validity of this claim, providing the first acknowledgement of there being two very different classes of contraction coefficients for a quantum channel. It also teaches us to not take the equivalence of relative expansion coefficients for granted.
% We take inspiration from the  method used to evaluate $\eta_{\kappa_s}^\mathrm{Riem}(\Phi_{\alpha,\tau})$ in 

\begin{theorem}[A Case of Inequivalence between Riemannian Coefficients]\label{thm:RiemInequiv}
    $ $\newline
    Consider the following family of classical-quantum, primitive, qubit channels: 
    \[
    \Phi_{\alpha,\sqrt{1-\alpha^2}}:\frac{1}{2}(I+\mathbf{w}\cdot\sigma)\mapsto\frac{1}{2}(I+\alpha w_1 \sigma_1+\sqrt{1-\alpha^2} \sigma_3),\quad 0<|\alpha|<1,\alpha\in\mathbb{R}
    \]
    If $\kappa\in\mathcal{K}$ is bounded, then:
    \begin{equation}\label{eqn:asympRiembdd}
        \frac{1}{\alpha^2}\eta_{\kappa}^\mathrm{Riem}(\Phi_{\alpha,\sqrt{1-\alpha^2}})=\Theta(1)\text{ as $\alpha\to0$}\end{equation}
    Otherwise, if $\kappa\in\mathcal{K}$ is unbounded, then:
    \begin{equation}
        \frac{1}{\alpha^2}\eta_{\kappa}^\mathrm{Riem}(\Phi_{\alpha,\sqrt{1-\alpha^2}})\to\infty\text{ as }\alpha\to 0
    \end{equation}
    In particular, if $\kappa_f\in\mathcal{K}$ is bounded and $\kappa_g\in\mathcal{K}$ is unbounded, for some $f,g\in\mathcal{F}$,
    \begin{equation}\label{eqn:inequivalence}
        \widecheck\eta_{\kappa_f}^\mathrm{Riem}\ncong\widecheck\eta_{\kappa_g}^\mathrm{Riem}
    \end{equation}
\end{theorem}
\begin{proof}
    $ $\\
    The bounded case \eqref{eqn:asympRiembdd} is immediately deduced via $\eta_{\kappa_\mathrm{min}}^\mathrm{Riem}(\Phi_{\alpha,\sqrt{1-\alpha^2}})=\alpha^2$ \cite[Theorem 6.2]{HR15} and the fact that all Riemannian contraction coefficients are equivalent, for bounded $\kappa$, by Theorem \ref{thm:bddRiemEquiv}. We henceforth consider only the unbounded setting of $\kappa$.

    \noindent For $\rho=\frac{I+\mathbf{w}\cdot\sigma}{2},X=\mathbf{y}\cdot \sigma,\xi_s(x):=(1+s)^2-(1-s)^2x=(1+s)^2(1-x)+4sx$. We will use the following results (the latter is taken from \cite[Appendix B.2.7]{HR15}):
    \begin{equation*}
    \kappa(x)=\int_{[0,1]}\kappa_s(x)\:dm(s),\;\|X\|_{\kappa_s,\rho}^2=\frac{4(1+s)^2}{\xi_s(|\mathbf{w}|^2)}\left[|\mathbf{y}|^2+\frac{4s(\mathbf{w}\cdot \mathbf{y})^2}{(1+s)^2(1-|\mathbf{w}|^2)}\right]
    \end{equation*}
    Note that $\xi_s(x)$ is a decreasing function with $\xi_s(1)=4s$. We can apply these results to the image states of $\Phi_{\alpha,\sqrt{1-\alpha^2}}$:
    \begin{align*}
        \liminf_{\alpha\to0}\frac{1}{\alpha^2}\|\Phi_{\alpha,\sqrt{1-\alpha^2}}(X)\|_{\kappa,\Phi_{\alpha,\sqrt{1-\alpha^2}}(\rho)}^2&=\liminf_{\alpha\to0}\int\frac{4(1+s)^2}{\xi_s(1-\alpha^2(1-w_1^2))}\cdot\left(y_1^2+\frac{4s(w_1y_1)^2}{(1+s)^2(1-w_1^2)}\right)\:dm(s)\\&\geq\int\liminf_{\alpha\to0}\frac{4(1+s)^2}{\xi_s(1-\alpha^2(1-w_1^2))}\cdot\left(y_1^2+\frac{4s(w_1y_1)^2}{(1+s)^2(1-w_1^2)}\right)\:dm(s)\\
        &=\int\frac{(1+s)^2}{s}\cdot y_1^2\left(1+\frac{4sw_1^2}{(1+s)^2(1-w_1^2)}\right)\:dm(s)
    \end{align*}
    The inequality above comes from applying Fatou's lemma. To compute a lower bound on the contraction coefficient $\eta_\kappa^\mathrm{Riem}(\Phi_{\alpha,\sqrt{1-\alpha^2}})$, wlog we take $y_1=1$, which makes $\|\Phi_{\alpha,\sqrt{1-\alpha^2}}(X)\|_{\kappa,\Phi(\rho)}^2$ independent of $\mathbf{y}$. We therefore proceed to minimise and upper bound $\|X\|_{\kappa_s,\rho}^2$ w.r.t. $\mathbf{y}$ for fixed $\mathbf{w}$ \cite[Lemma B.2]{HR15}:
    \begin{equation*}
        \min_{\mathbf{y}:y_1=1}\frac{4(1+s)^2}{\xi_s(|\mathbf{w}|^2)}\left[|\mathbf{y}|^2+\frac{4s(\mathbf{w}\cdot \mathbf{y})^2}{(1+s)^2(1-|\mathbf{w}|^2)}\right]=\frac{\mu(\mu+\nu|\mathbf{w}|^2)}{\mu+\nu(|\mathbf{w}|^2-w_1^2)}
    \end{equation*}
    where $\mu:=\int\frac{4(1+s)^2}{\xi_s(|\mathbf{w}|^2)}\:dm(s),\nu:=\int\frac{16}{\xi_s(|\mathbf{w}^2|)}\cdot\frac{s}{1-|\mathbf{w}|^2}\:dm(s)$. Note that $\mu+\nu|\mathbf{w}|^2=\frac{4}{1-|\mathbf{w}|^2}$, and thus:
    \begin{align}\label{eqn:inequivsimplify}
        \begin{split}
            \frac{\mu\cdot \frac{4}{1-|\mathbf{w}|^2}}{\frac{4}{1-|\mathbf{w}|^2}-\nu\cdot w_1^2}&=\frac{4\mu}{4-\nu w_1^2(1-|\mathbf{w}|^2)}=\frac{\int\frac{16(1+s)^2}{\xi_s(|\mathbf{w}|^2)}\:dm(s)}{4-w_1^2\int\frac{16s}{\xi_s(|\mathbf{w}|^2)}\:dm(s)}\\&\leq \sup_s \frac{16(1+s)^2}{4\xi_s(|\mathbf{w}|^2)-w_1^2\cdot16s}=\sup_s\frac{4}{1-w_1^2-\frac{(1-s)^2}{(1+s)^2}(w_2^2+w_3^2)}\\&=\frac{4}{1-|\mathbf{w}|^2}
        \end{split}
    \end{align}
    Now, we can establish a lower bound on $\eta_\kappa^\mathrm{Riem}(\Phi_{\alpha,\sqrt{1-\alpha^2}})$ as $\alpha\to0$:
    \begin{align*} \liminf_{\alpha\to0}\sup_{\mathbf{w},\mathbf{y}:|\mathbf{w}|<1,y_1=1}\frac{1}{\alpha^2}&\frac{\|\Phi_{\alpha,\sqrt{1-\alpha^2}}(X)\|_{\kappa,\Phi_{\alpha,\sqrt{1-\alpha^2}}(\rho)}^2}{\|X\|_{\kappa,\rho}^2}\\&\geq\sup_{\mathbf{w},\mathbf{y}:|\mathbf{w}|<1,y_1=1}\liminf_{\alpha\to0}\frac{1}{\alpha^2}\frac{\|\Phi_{\alpha,\sqrt{1-\alpha^2}}(X)\|_{\kappa,\Phi_{\alpha,\sqrt{1-\alpha^2}}(\rho)}^2}{\|X\|_{\kappa,\rho}^2}\\&=\sup_{\mathbf{w}:|\mathbf{w}|<1}\frac{\int\frac{(1+s)^2}{s}\cdot\left(1+\frac{4sw_1^2}{(1+s)^2(1-w_1^2)}\right)\:dm(s)}{\frac{\mu(\mu+\nu|\mathbf{w}|^2)}{\mu+\nu(|\mathbf{w}|^2-w_1^2)}}\\&\geq \frac{1}{4} \sup_{\mathbf{w}:|\mathbf{w}|<1}(1-|\mathbf{w}|^2)\int\frac{(1+s)^2}{s}\cdot\left(1+\frac{4sw_1^2}{(1+s)^2(1-w_1^2)}\right)\:dm(s) \\&=\frac{1}{4} \sup_{|w_1|<1}\int\frac{(1+s)^2}{s}(1-w_1^2+\frac{4sw_1^2}{(1+s)^2})\:dm(s) \\&=\frac{1}{4} \sup_{|w_1|<1}\int\frac{(1+s)^2}{s}(1-\frac{(1-s)^2}{(1+s)^2}w_1^2)\:dm(s) \\&=\frac{1}{4} \int\frac{(1+s)^2}{s}\:dm(s)=\frac{\kappa(0^+)}{2}
    \end{align*}
    The first inequality comes from the fact that, for any real function $f(\alpha,\mathbf{w},\mathbf{y})$, 
    \[
    \liminf_{\alpha\to 0}\sup_{\mathbf{w},\mathbf{y}}f(\alpha,\mathbf{w},\mathbf{y})\geq\liminf_{\alpha\to 0}f(\alpha,\mathbf{w}',\mathbf{y}')\quad\forall\mathbf{w}',\mathbf{y}'\quad
    \implies\quad\liminf_{\alpha\to 0}\sup_{\mathbf{w},\mathbf{y}}f(\mathbf{w},\mathbf{y})\geq \sup_{\mathbf{w},\mathbf{y}}\liminf_{\alpha\to 0}f(\mathbf{w},\mathbf{y})
    \]
    The second inequality used \eqref{eqn:inequivsimplify}, and the second equality came from the fact that for any fixed $w_1$, $w_2=w_3=0$ is optimal. The final equality came from the integral representation  \eqref{opmon} of $\kappa$.
    This gives the result in the unbounded case:
    \begin{equation*}
        \lim_{\alpha\to0}\frac{1}{\alpha^2}\eta_\kappa^\mathrm{Riem}(\Phi_{\alpha,\sqrt{1-\alpha^2}})=\liminf_{\alpha\to0}\frac{1}{\alpha^2}\eta_\kappa^\mathrm{Riem}(\Phi_{\alpha,\sqrt{1-\alpha^2}})\geq \frac{\kappa(0^+)}{2}=\infty
    \end{equation*}
\end{proof}

Note that since $\eta_g^\mathrm{std}(\Phi_{\alpha,\tau})\geq\eta_{\kappa_g}^\mathrm{Riem}(\Phi_{\alpha,\tau})$, we also have, for bounded $\kappa_f$ and unbounded $\kappa_g$,
\begin{equation}\label{eqn:divrieminequiv}
    \frac{1}{\alpha^2}\eta_g^\mathrm{std}(\Phi_{\alpha,\sqrt{1-\alpha^2}})\to\infty\text{ as }\alpha\to 0,\text{ and }\widecheck\eta_g^\mathrm{std}\ncong\widecheck\eta_{\kappa_f}^\mathrm{Riem}.
\end{equation}

This counterexample helps to establish that the Riemannian relative expansion coefficients in the bounded $\kappa$ setting form an entire equivalence class, as they are not equivalent to Riemannian coefficients corresponding to unbounded $\kappa$. For example, we find that $\widecheck\eta_{x\log x}^\mathrm{std}\equiv\widecheck\eta_\mathrm{BKM}^\mathrm{Riem}\ncong\widecheck\eta_{\kappa_f}^\mathrm{Riem}$ for all bounded $\kappa_f\in\mathcal{K}$. Thus, the notion of equivalence proves to capture important differences between distinguishability measures. The potentially many order of magnitude difference may have serious implications to the accuracy of using different Riemannian contraction coefficients to upper bound the mixing time of a quantum Markov chain based on a primitive channel (see Section~\ref{subsec:primitive}).

\subsection{A Redundancy for Strictly Positive Channels}\label{sec:fullrankoutputchannels}
In the setting of quantum-classical channels (or of classical channels), we saw that the divergence and Riemannian relative expansion coefficients are all the same, $\widecheck\eta_f^\mathrm{std}\equiv\widecheck\eta_\kappa^\mathrm{Riem}$ (Section~\ref{sec:qc}); so quantum channels with non-commutative output are necessary for any differences (e.g. \eqref{eqn:equalitycounter}). In the quantum case, we saw in Section~\ref{sec:bddvsunbdd} that we can sometimes have inequivalence (e.g. \eqref{eqn:divrieminequiv}). This subsection offers some further insight into this loss of equivalence, by restricting to strictly positive channels, i.e. quantum channels that map all states into full-rank states. Specifically, within the class of strictly positive channels, we can ensure the equivalence, $\widecheck\eta_f^\mathrm{std}\cong\widecheck\eta_\kappa^\mathrm{Riem}$. 

We will find it useful to define the following notation:
\begin{enumerate}
    \item $\lambda_\mathrm{min}(\rho)$ (resp. $\lambda_\mathrm{max}(\rho)$) denotes the minimum (resp. maximum) eigenvalue of a density operator $\rho\in\mathcal{D}_d$.
    \item The \textit{condition number} of a density operator $\rho\in\mathcal{D}(\mathcal{H})$ is $c(\rho):=\frac{\lambda_\mathrm{max}(\rho)}{\lambda_\mathrm{min}(\rho)}$.
    \item For a quantum channel $\mathcal{N}:\mathcal{B}(\mathcal{H}_A)\to\mathcal{B}(\mathcal{H}_B),$
    \[\lambda_\mathrm{min}(\mathcal{N}):=\underset{\rho\in\mathcal{D}(\mathcal{H}_A)}{\mathrm{min}}\;\lambda_\mathrm{min}(\mathcal{N}(\rho)),
    \]
    \[\lambda_\mathrm{max}(\mathcal{N}):=\underset{\rho\in\mathcal{D}(\mathcal{H}_A)}{\mathrm{max}}\;\lambda_\mathrm{max}(\mathcal{N}(\rho)).
    \]
    
    % \item The \textit{output condition number} of a CPTP map $\mathcal{N}:\mathcal{B}(\mathcal{H}_A)\to\mathcal{B}(\mathcal{H}_B)$ is \[c(\mathcal{N}):=\underset{\rho\in\mathcal{D}(\mathcal{H}_A)}{\mathrm{max}}\;c(\mathcal{N}(\rho))\]
\end{enumerate}

\noindent Observe that, by continuity, a quantum channel $\mathcal{N}$ is strictly positive iff $\lambda_\mathrm{min}(\mathcal{N})>0$. For a choice of parameter $\lambda\in(0,1)$, we will compare quantum channels from the class:
\begin{equation*}
    \mc Q_\lambda:=\{\text{CPTP maps $\mathcal{N}$ with }\lambda_{\min}(\mathcal{N})\geq \lambda\}.
\end{equation*}

% We have previously seen that by restricting to QC channels, the divergence and Riemannian expansion coefficients become redundant, i.e.,
% \[
% \widecheck{\eta}_f^{\mathrm{std}} \equiv \widecheck{\eta}_\kappa^{\mathrm{Riem}}
% \]
% for all $f\in\mathcal{F}, \kappa\in\mathcal{K}$.
% Here, we will consider another restriction to a family of channels that ensures (the slightly weaker) \[
% \widecheck{\eta}_f^{\mathrm{std}} \cong \widecheck{\eta}_\kappa^{\mathrm{Riem}}
% \]
% for all $f\in\mathcal{F}, \kappa\in\mathcal{K}$. Namely, this happens to hold when we are comparing quantum channels whose output density operators have full rank. We will find it useful to define the following notation:

\noindent The following results take inspiration from \cite[Appendix B.2.7]{HR15}, where it can be noticed from the qubit formulae for the Riemannian semi-norms, $\|X\|_{\kappa,\:\rho}^2$, that $\|X\|_{\kappa_\mathrm{max},\:\rho}^2$ often shows up as a factor. The proof methods used highlight the importance of having the integral representations \eqref{opcon}, \eqref{opmon}.
%\paula{we should add a comment somewhere highlighting about how this then includes the unbounded cases?}\textcolor{purple}{[but why?]}

\begin{theorem}[Equivalence of Riemannian Coefficients on Strictly Positive Channels]\label{thm:fullrankRiemEquiv}
	$ $\newline 
    Any $\kappa \in \mathcal{K}$ satisfies for all positive definite density operators $\rho$, traceless Hermitian $X$, $\rho,X \in \mathcal{B}(\mathcal{H})$, some finite-dimensional Hilbert space $\mathcal{H}$:
	\begin{equation*}
	    \alpha\leq\frac{\|X\|^2_{\kappa,\:\rho}}{\|X\|^2_{\kappa_\mathrm{max},\:\rho}}\leq\beta
	\end{equation*}
    
	\noindent For $\alpha=\kappa(c(\rho))\geq \kappa(\lambda^{-1}),\:\beta=1$. \\ \\
    \noindent In particular, for the class of quantum channels $\mathcal{Q}=\mathcal{Q}_\lambda$, some $\lambda>0$,
    \begin{equation}
        \widecheck{\eta}^{\mathrm{Riem}}_{\kappa} \cong_{\alpha,\beta} \widecheck{\eta}^{\mathrm{Riem}}_{\kappa_\mathrm{max}}
    \end{equation} 
    For $\alpha=\kappa(\lambda^{-1}),\beta=\kappa(\lambda^{-1})^{-1}$
    % for all quantum channels $\mathcal{N}:\mathcal{B}(\mathcal{H}_A)\to\mathcal{B}(\mathcal{H}_B),\:\mathcal{M}:\mathcal{B}(\mathcal{H}_A)\to\mathcal{B}(\mathcal{H}_B')$ whose output density operators have full rank.
\end{theorem}
\begin{proof}
    Consider the spectral decomposition $\rho=\sum_{a\:\in\:\mathrm{supp}(\rho)}aP_a$ and the integral representation for $\kappa\in\mathcal{K}$  \eqref{opmon}:
    \[
    \kappa(x) = \int_{[0,1]} \frac{1 + s}{2} \left( \frac{1}{x + s} + \frac{1}{s x + 1} \right)\, dm(s)
    \]

    We can express the Riemannian semi-norms as:
    \begin{equation}\label{Riemdecomp}
        \|X\|_{\kappa,\:\rho}^2=\sum_{a,b\:\in\:\mathrm{spec}\:\rho} b^{-1}\kappa(ab^{-1})\:\Tr XP_aXP_b,\quad \|X\|_{\kappa_{\max},\:\rho}^2=\Tr X^2\rho^{-1} = \sum_{a\:\in\:\mathrm{spec}\:\rho}a^{-1}\Tr X^2P_a
    \end{equation}
    
    \noindent Therefore,
    \begin{align*}
        1&\geq\frac{\|X\|_{\kappa,\:\rho}^2}{\|X\|_{\kappa_\mathrm{max},\:\rho}^2}=\frac{\sum_{a,b\:\in\:\mathrm{supp}(\rho)}\:\Tr\:XP_aXP_b\:\cdot\:\int_{[0,1]}\frac{1+s}{2}(\frac{1}{a+sb}+\frac{1}{sa+b})\:dm(s)}{\sum_{a,b\:\in\:\mathrm{supp}(\rho)}\:\Tr\:XP_aXP_b\:\cdot\:\frac{1}{a}}\\
        &\geq\underset{a,b}{\mathrm{min}} \int_{[0,1]}\frac{1+s}{2}\left(\frac{a}{a+sb}+\frac{a}{sa+b} \right)\:dm(s)\\
        &=\underset{a,b}{\mathrm{min}}\:\kappa\left(\frac{b}{a}\right)=\kappa(c(\rho))>0
    \end{align*}
    \noindent Where the first inequality follows from the maximality of $\kappa_\mathrm{max}$ and the final equality follows from the fact $\kappa\in\mathcal{K}$ is necessarily (positive, operator-) monotone decreasing.\\
    
    Now considering these inequalities, we have for any quantum channels $\mathcal{N}:\mathcal{B}(\mathcal{H}_A)\to\mathcal{B}(\mathcal{H}_B),\mathcal{M}:\mathcal{B}(\mathcal{H}_A)\to\mathcal{B}(\mathcal{H}_B')\in\mathcal{Q}$ and this time $\rho,X\in\mathcal{B}(\mathcal{H}_A)$:

    \begin{equation}
        \kappa(c(\mathcal{N}(\mathcal{\rho})))\frac{\|\mathcal{N}(X)\|_{\kappa_{\max},\:\mathcal{N}(\rho)}^2}{\|\mathcal{M}(X)\|_{\kappa_{\max},\:\mathcal{M}(\rho)}^2}\leq \frac{\|\mathcal{N}(X)\|_{\kappa,\:\mathcal{N}(\rho)}^2}{\|\mathcal{M}(X)\|_{\kappa,\:\mathcal{M}(\rho)}^2} \leq \kappa(c(\mathcal{M}(\rho))^{-1}\frac{\|\mathcal{N}(X)\|_{\kappa_{\max},\:\mathcal{N}(\rho)}^2}{\|\mathcal{M}(X)\|_{\kappa_{\max},\:\mathcal{M}(\rho)}^2}
    \end{equation}

    Since $\mathcal{N},\mathcal{M}\in\mathcal{Q}_\lambda$, 
    \[c(\mathcal{N}(\mathcal{\rho})),c(\mathcal{M}(\mathcal{\rho}))\leq\lambda^{-1}\]
    The fact that $\kappa$ is monotone decreasing, and taking the infimum over all $\rho\in\mathcal{D}(\mathcal{H}_A), X\in T_\rho\mathcal{D}(\mathcal{H}_A)$, gives the result. 
\end{proof}

\begin{theorem}[Equivalence of Divergence Coefficients on Strictly Positive Channels]\label{thm:fullrankDivEquiv}
    $ $\newline 
    Any $f\in \mathcal{F}$ satisfies for all positive definite density operators $\rho,\gamma \in \mathcal{D}^+(\mathcal{H})$, some finite-dimensional Hilbert space $\mathcal{H}$:
    \begin{equation*}
        \alpha\leq\frac{D_f^\mathrm{std}(\rho\|\gamma)}{D_\mathrm{max}^\mathrm{std}(\rho\|\gamma)}\leq\beta
    \end{equation*}
    \noindent For $\alpha=\nu_f(\lambda_\mathrm{max}(\rho)/\lambda_\mathrm{min}(\gamma)),\beta=\nu_f(\lambda_\mathrm{min}(\rho)/\lambda_\mathrm{max}(\gamma))>0$, where $\nu_f(x):=\frac{f(x)-f'(1)(x-1)}{(x-1)^2}$.\\ \\
    \noindent In particular, for the class of quantum channels $\mathcal{Q}=\mathcal{Q}_\lambda$, some $\lambda>0$, \[\widecheck{\eta}^{\mathrm{std}}_{f} \cong \widecheck{\eta}^{\mathrm{std}}_{(x-1)^2}\]
    For $\alpha=\nu_f(\lambda^{-1})/\nu_f(\lambda),\:\beta=\nu_f(\lambda)/\nu_f(\lambda^{-1})$.
\end{theorem}
\begin{proof} Consider the spectral decompositions $\rho=\sum_{a\:\in\:\mathrm{supp}(\rho)}aP_a,\:\gamma=\sum_{b\:\in\:\mathrm{supp}(\gamma)}bQ_b$ and the integral representation for $f\in\mathcal{F}$  (where $c\geq0$) \eqref{opcon}:
    \[
    f(x) = f'(1)(x - 1) + c(x - 1)^2 + \int_{[0,\infty)} \frac{(x - 1)^2}{x + s} \, d\mu(s), \quad x \in (0,\infty).
    \]
    Notice in particular that:
    \[
    \nu_f(x) = c + \int_{[0,\infty)} \frac{1}{x + s} \, d\mu(s), \quad x \in (0,\infty)
    \]
    is a decreasing, strictly positive, function and $\nu_f(\infty)=c$. 

    We can express standard $f$-divergences as:
    \[D_f^\mathrm{std}(\rho\|\gamma)=\sum_{a\:\in\:\mathrm{spec}(\rho),b\in\mathrm{spec}(\gamma)} bf\left(\frac{a}{b}\right)\Tr P_aQ_b
    \]
    
    Therefore:
    \begin{align*}
        \frac{D_f^\mathrm{std}(\rho\|\gamma)}{D_\mathrm{max}^\mathrm{std}(\rho\|\gamma)}&=\frac{\sum_{a\in\mathrm{spec}(\rho),b\in\mathrm{spec}(\gamma)}\Tr\:P_aQ_b\cdot[cb(a/b-1)^2+\int_{[0,\infty)}\:b\:\cdot\:\frac{(a/b-1)^2}{a/b+s}d\mu(s)]}{\sum_{a\in\mathrm{spec}(\rho),b\in\mathrm{spec}(\gamma)}\Tr\:P_aQ_b\cdot b(a/b-1)^2}\\
        &\in\left[ \underset{a,b}{\mathrm{min}}\;c+\int_{[0,\infty)}\frac{1}{a/b+s}\;d\mu(s),\:\underset{a,b}{\mathrm{max}}\;c+\int_{[0,\infty)}\frac{1}{a/b+s}\;d\mu(s)\right]\\
        &=[\nu_f(\lambda_\mathrm{max}(\rho)/\lambda_\mathrm{min}(\gamma)),\nu_f(\lambda_\mathrm{min}(\rho)/\lambda_\mathrm{max}(\gamma))]
    \end{align*}

    Now, considering these inequalities, we have for any quantum channels $\mathcal{N}:\mathcal{B}(\mathcal{H}_A)\to\mathcal{B}(\mathcal{H}_B),\:\mathcal{M}:\mathcal{B}(\mathcal{H}_A)\to\mathcal{B}(\mathcal{H}_B')\in\mathcal{Q}_\lambda$ and this time $\rho,\gamma\in\mathcal{D}(\mathcal{H})$:
    \[
    \frac{\nu_f\left(\frac{\lambda_\mathrm{max}(\mathcal{N}(\rho))}{\lambda_\mathrm{min}(\mathcal{N}(\gamma))}\right)}{\nu_f\left(\frac{\lambda_\mathrm{min}(\mathcal{M}(\rho))}{\lambda_\mathrm{max}(\mathcal{M}(\gamma))}\right)}\cdot\frac{D_\mathrm{max}^\mathrm{std}(\mathcal{N}(\rho)\|\mathcal{N}(\gamma))}{D_\mathrm{max}^\mathrm{std}(\mathcal{M}(\rho)\|\mathcal{M}(\gamma))}\leq \frac{D_f^\mathrm{std}(\mathcal{N}(\rho)\|\mathcal{N}(\gamma))}{D_f^\mathrm{std}(\mathcal{M}(\rho)\|\mathcal{M}(\gamma))}\leq\frac{\nu_f\left(\frac{\lambda_\mathrm{min}(\mathcal{N}(\rho))}{\lambda_\mathrm{max}(\mathcal{N}(\gamma))}\right)}{\nu_f\left(\frac{\lambda_\mathrm{max}(\mathcal{M}(\rho))}{\lambda_\mathrm{min}(\mathcal{M}(\gamma))}\right)}\cdot\frac{D_\mathrm{max}^\mathrm{std}(\mathcal{N}(\rho)\|\mathcal{N}(\gamma))}{D_\mathrm{max}^\mathrm{std}(\mathcal{M}(\rho)\|\mathcal{M}(\gamma))}
    \]
    But since $\nu_f(x)$ is decreasing:
    \[
    \frac{\nu_f(\lambda^{-1})}{\nu_f(\lambda)}\leq\frac{\nu_f\left(\frac{\lambda_\mathrm{max}(\mathcal{N}(\rho))}{\lambda_\mathrm{min}(\mathcal{N}(\gamma))}\right)}{\nu_f\left(\frac{\lambda_\mathrm{min}(\mathcal{M}(\rho))}{\lambda_\mathrm{max}(\mathcal{M}(\gamma))}\right)}\quad\text{ and }\quad \frac{\nu_f\left(\frac{\lambda_\mathrm{min}(\mathcal{N}(\rho))}{\lambda_\mathrm{max}(\mathcal{N}(\gamma))}\right)}{\nu_f\left(\frac{\lambda_\mathrm{max}(\mathcal{M}(\rho))}{\lambda_\mathrm{min}(\mathcal{M}(\gamma))}\right)}\leq\frac{\nu_f(\lambda)}{\nu_f(\lambda^{-1})}
    \]
    Taking the infimum over all $\rho\neq\gamma,\:\mathrm{supp}\:\rho=\:\mathrm{supp}\:\gamma$ thus gives the result. 
\end{proof}

\begin{corollary}\label{fullrankEquiv}
    $ $\newline
    For the class of quantum channels $\mathcal{Q}=\mathcal{Q}_\lambda$, some $\lambda>0$, and any $f\in\mathcal{F},\kappa\in\mathcal{K}$:
    \[
    \widecheck{\eta}^{\mathrm{std}}_{f} \cong \widecheck{\eta}^{\mathrm{Riem}}_{\kappa}
    \]
\end{corollary}
\begin{proof}
    Since $\eta_{(x-1)^2}^\mathrm{std}\equiv\eta_{\kappa_\mathrm{max}}^\mathrm{Riem}$ (see Theorem~\ref{thm:equalitycases}), applying Theorems \ref{thm:fullrankRiemEquiv} and \ref{thm:fullrankDivEquiv} gives the result.
\end{proof}

The interpretation of Corollary~\ref{fullrankEquiv} is that any inequivalence between divergence or Riemannian relative expansion coefficients is related to the behaviour of the distinguishability measure $D_f(\rho\|\gamma)$ as $\rho$ or $\gamma$ approach non-positive definite states. Notice that in the proof of Theorem~\ref{thm:RiemInequiv} may have taken advantage of this; the lower bound on the unbounded-case Riemannian contraction coefficient, that allowed us to establish inequivalence, was effectively obtained by taking the Bloch vectors $\mathbf{y}\to(1,0,0),\mathbf{w}\to(0,0,0)$. This limit of $\mathbf{y}$ corresponds to a difference between a pure state $\widetilde{\mathbf{w}}=(1,0,0)$ and the limit of $\mathbf{w}$, and $\widetilde{\mathbf{w}}$ is also mapped close to another pure state $\widetilde{\mathbf{w}}'=(0,0,1)$, under $\Phi_{\alpha,\sqrt{1-\alpha^2}}$, as we took the limit $\alpha\to 0$.

\section{Applications}\label{sec:applications}
In this section, we provide two main applications of relative expansion coefficients. We provide explicit estimates of relative expansion coefficients in the last subsection. 

\subsection{Approximate Recoverability}\label{subsec:recover}
A positive expansion coefficient suggests that a certain proportion of quantum information is preserved by a quantum channel. In this subsection, we will explore how this interpretation fits into the problem of recovering input states from the output of a quantum channel. 

Firstly, let us recall the following result from \cite{junge2018universal}, in the finite dimensional setting:
\begin{theorem}[Approximate Sufficiency of Quantum Relative Entropy]\label{thm:RelEntSuff}
    $ $\newline Suppose that we are provided a quantum channel $\mathcal{N}:\mathcal{B}(\mathcal{H}_A)\to\mathcal{B}(\mathcal{H}_B)$, and any quantum states $\rho,\gamma\in\mathcal{D}(\mathcal{H}_A)$ such that $\mathrm{supp}\:\rho\leq\mathrm{supp}\:\gamma$, then:
    \begin{equation}\label{eqn:jungerecover}
        D(\rho\|\gamma)-D(\mathcal{N}(\rho)\|\mathcal{N}(\gamma))\geq -2\log F(\rho,\mathcal{R}_{\gamma,\mathcal{N}}^\mathrm{uni}\circ\mathcal{N}(\rho))\geq \|\rho-\mathcal{R}_{\gamma,\mathcal{N}}^\mathrm{uni}\circ\mathcal{N}(\rho)\|_1^2
    \end{equation}
    where $F(\rho,\gamma):=\|\sqrt{\rho}\sqrt{\gamma}\|_1$ is the fidelity and the universal recovery map is:
    \begin{equation*}
        \mathcal{R}_{\gamma,\mathcal{N}}^\mathrm{uni}(\cdot):=\int_\mathbb{R} dt\:\beta_0(t)\:\mathcal{R}_{\gamma,\mathcal{N}}^\frac{t}{2}(\cdot)
    \end{equation*}
    for the rotated Petz recovery map $\mathcal{R}_{\gamma,\mathcal{N}}^t:\mathrm{supp}\:\mathcal{N}(\gamma)\to\mathcal{B}(\mathcal{H}_A)$ and $\beta_0(t)$ a probability density function on $\mathbb{R}$:
    \begin{equation*}
        \mathcal{R}_{\gamma,\mathcal{N}}^t:A\mapsto \gamma^{1/2-it}\widehat{\mathcal{N}}(\mathcal{N}(\gamma)^{-1/2+it}A\mathcal{N}(\gamma)^{-1/2-it})\gamma^{1/2+it},\quad \beta_0(t):=\frac{\pi}{2}(\cosh(\pi t)+1)^{-1}
    \end{equation*}
\end{theorem}

This result formalises an intuition that if there is only a small reduction in the relative entropy of two input states as they are transmitted through a quantum channel, then it should be possible to accurately recover the input states. The recovery map is universal, because it only depends on $\mathcal{N},\gamma$, and not $\rho$; this is similar to the perfect recovery result, \cite[Theorem 7.1]{Hiai_2011}, which in the case the relative entropy does not decrease, one may perfectly recover the input states using the Petz recovery map $\mathcal{R}_{\gamma,\mathcal{N}}^0$. By incorporating the definition of the relative expansion coefficient, we obtain the following corollary:

\begin{corollary}[Approximate Sufficiency via BKM Expansion Coefficient]\label{cor:BKMrecovery}
    $ $\newline Suppose that we are provided quantum channels $\mathcal{D}:\mathcal{B}(\mathcal{H}_A)\to\mathcal{B}(\mathcal{H}_A'),\mathcal{N}:\mathcal{B}(\mathcal{H}_A')\to\mathcal{B}(\mathcal{H}_B)$ and any quantum states $\rho,\gamma\in\mathcal{D}(\mathcal{H}_A)$ such that $\mathrm{supp}\:\rho\leq\mathrm{supp}\:\gamma$, then:
    \begin{equation}\label{eqn:BKMrecover}
        (1-\widecheck\eta_\mathrm{BKM}^\mathrm{Riem}(\mathcal{N};\mathrm{Im}\:\mathcal{D}))D(\mathcal{D}(\rho)\|\mathcal{D}(\gamma))\geq \|\mathcal{D}(\rho)-\mathcal{R}_{\mathcal{D}(\gamma),\mathcal{N}}^\mathrm{uni}\circ\mathcal{N}\circ\mathcal{D}(\rho)\|_1^2
    \end{equation}
    In particular, when $\widecheck\eta_\mathrm{BKM}^\mathrm{Riem}(\mathcal{N};\mathrm{Im}\:\mathcal{D})>0$, this upper bound (LHS) is an improvement over recovery via the replacer channel $\mathcal{R}_\gamma(\cdot):=\Tr(\cdot)\gamma$, which has the bound:
    \begin{equation}\label{eqn:hiddenpinsker}
        D(\mathcal{D}(\rho)\|\mathcal{D}(\gamma))\geq \|\mathcal{D}(\rho)-\mathcal{R}_{\mathcal{D}(\gamma)}\circ\mathcal{N}\circ\mathcal{D}(\rho)\|_1^2
    \end{equation}
\end{corollary}
\begin{proof}
    
    Recalling that $\widecheck\eta_{x\log x}^\mathrm{std}\equiv\widecheck\eta_\mathrm{BKM}^\mathrm{Riem}$ (see \eqref{thm:equalitycases}), we have the following reverse data-processing inequality by the definition of the divergence expansion coefficient:
    \begin{equation*}
        D(\mathcal{N}\circ\mathcal{D}(\rho)\|\mathcal{N}\circ\mathcal{D}(\gamma))\geq \widecheck\eta_\mathrm{BKM}^\mathrm{Riem}(\mathcal{N};\mathrm{Im}\:\mathcal{D})\cdot D(\mathcal{D}(\rho)\|\mathcal{D}(\gamma))
    \end{equation*}
    Then, by combining this inequality with \eqref{eqn:jungerecover} applied to $\mathcal{D}(\rho),\mathcal{D}(\gamma)$ (certainly, $\mathrm{supp}\:\mathcal{D}(\rho)\leq\mathrm{supp}\:\mathcal{D}(\gamma)$), we obtain \eqref{eqn:BKMrecover}.

    Also note that \eqref{eqn:hiddenpinsker} is simply the Pinsker inequality, which holds in general for states $\mathrm{supp}\:\rho'\leq\mathrm{supp}\:\gamma'$, i.e.
    \begin{equation*}
        D(\rho'\|\gamma')\geq \|\rho'-\gamma'\|_1^2
    \end{equation*}
\end{proof}

In fact, this provides an alternative perspective to the positive expansion coefficients, $\widecheck\eta_\mathrm{BKM}^\mathrm{Riem}(\mathcal{N};\mathrm{Im}\:\mathcal{D})>0$, demonstrated in \cite{BGSW24}, and they now have a practical application to the approximate recoverability of input states to the quantum channels $\mathcal{N}$ considered. Recall that we restrict the domain of the input states to $\mathrm{Im}\:\mathcal{D}$, because here we are interested in $d_B\leq d_A'$. 

Analogous to Theorem~\ref{thm:RelEntSuff}, \cite{gao2023sufficient} provided various examples where a small reduction in a Riemannian semi-norm (or $\chi^2$-divergence), again between an input state and an input reference state, implies the accurate recovery of the input states. We discuss one key example from the paper (\cite[Corollary 4.12]{gao2023sufficient}):

\begin{theorem}[Approximate Sufficiency of $\alpha$-Metric]
    $ $\newline Suppose that we are provided a quantum channel $\mathcal{N}:\mathcal{B}(\mathcal{H}_A)\to\mathcal{B}(\mathcal{H}_B)$. For $\kappa_\alpha(x)=x^\alpha,\alpha\in(-\frac{1}{2},0)$ and $\gamma\in\mathcal{D}(\mathcal{H}),\mathrm{supp}\:A\leq\mathrm{supp}\:\gamma$, 
\begin{equation}\label{eqn:gaorecover}
    \|A\|_{\kappa_\alpha,\:\gamma}^2-\|\mathcal{N}(A)\|_{\kappa_\alpha,\:\mathcal{N}(\gamma)}^2\geq \left(\frac{\pi}{\cosh \pi t}\frac{\|A-\mathcal{R}_{\gamma,\:\mathcal{N}}^t(A)\|_1^2}{K(\gamma,A,\alpha)}\right)^\frac{2}{\alpha}
\end{equation}
for some non-negative function $K(\gamma,A,\alpha)$.
\end{theorem}

And similarly, we can incorporate the Riemannian expansion coefficient into the result:
\begin{corollary}[Approximate Sufficiency via $\alpha$-Riemannian Expansion Coefficient]\label{cor:alpharecover}
    $ $\newline Suppose that we are provided quantum channels $\mathcal{D}:\mathcal{B}(\mathcal{H}_A)\to\mathcal{B}(\mathcal{H}_A'),\mathcal{N}:\mathcal{B}(\mathcal{H}_A')\to\mathcal{B}(\mathcal{H}_B)$ and any quantum states $\rho,\gamma\in\mathcal{D}(\mathcal{H}_A)$ such that $\mathrm{supp}\:\rho\leq\mathrm{supp}\:\gamma$, then:
\begin{equation}\label{eqn:alpharecover}
    (1-\widecheck\eta_{\kappa_\alpha}^\mathrm{Riem}(\mathcal{N};\mathrm{Im}\:\mathcal{D}))\|\mathcal{D}(\rho)-\mathcal{D}(\gamma)\|_{\kappa_\alpha,\:\mathcal{D}(\gamma)}^2\geq \left(\frac{\pi}{\cosh \pi t}\frac{\|\mathcal{D}(\rho)-\mathcal{R}_{\mathcal{D}(\gamma),\:\mathcal{N}}^t(\mathcal{N}\circ\mathcal{D}(\rho))\|_1^2}{K(\mathcal{D}(\gamma),\mathcal{D}(\rho)-\mathcal{D}(\gamma),\alpha)}\right)^\frac{2}{\alpha}
\end{equation}
for some non-negative function $K(\gamma,A,\alpha)$.
\end{corollary}
\begin{proof}
    We have the following reverse data-processing inequality by the definition of the Riemannian expansion coefficient:
    \begin{equation*}
        \|\mathcal{N}\circ\mathcal{D}(\rho)-\mathcal{N}\circ\mathcal{D}(\gamma)\|_{\kappa_\alpha,\:\mathcal{N}\circ\mathcal{D}(\gamma)}^2\geq \widecheck\eta_{\kappa_\alpha}^\mathrm{Riem}(\mathcal{N};\mathrm{Im}\:\mathcal{D})\cdot \|\mathcal{D}(\rho)-\mathcal{D}(\gamma)\|_{\kappa_\alpha,\:\mathcal{D}(\gamma)}^2
    \end{equation*}

    Also, note that $\mathcal{R}_{\mathcal{D}(\gamma),\:\mathcal{N}}^t(\mathcal{N}\circ\mathcal{D}(\gamma))=\mathcal{D}(\gamma)$, so 
    \[
        \|\mathcal{D}(\rho)-\mathcal{D}(\gamma)-\mathcal{R}_{\mathcal{D}(\gamma),\mathcal{N}}^t(\mathcal{D}(\rho)-\mathcal{D}(\gamma))\|_1^2=\|\mathcal{D}(\rho)-\mathcal{R}_{\mathcal{D}(\gamma),\:\mathcal{N}}^t(\mathcal{N}\circ\mathcal{D}(\rho))\|_1^2
    \]
    
    Then, by combining the reverse data-processing inequality with \eqref{eqn:gaorecover} applied to $\mathcal{D}(\rho),\mathcal{D}(\gamma)$, we obtain \eqref{eqn:alpharecover}.

\end{proof}
Importantly, \cite{gao2023sufficient} developed a technique (via \cite[Lemma 4.4]{gao2023sufficient}) to find similar approximate sufficiency bounds for different Riemannian semi-norms, i.e. different $\kappa\in\mathcal{K}$. For any such inequality, we can translate them into bounds based on Riemannian expansion coefficients in much the same way as Corollary~\ref{cor:alpharecover}. 

In this paper, we show examples of positive Riemannian expansion coefficients, simultaneously over all $\kappa\in\mathcal{K}$ for particular choices of channels. These cases may yield non-trivial recovery bounds, like how they implied better recovery bounds using the universal recovery map over the replacer channel in Corollary~\ref{cor:BKMrecovery}. We will see such positive expansion coefficients in the context of primitive quantum channels (Section~\ref{subsec:primitive}), and later we find them again when we consider qubit channels (Section~\ref{sec:explicitcoeffs}), such as dephasing. 

% We would ideally like Riemannian expansion coefficients close to 1, as they would result in very tight bounds to $\|\mathcal{D}(\rho)-\mathcal{R}_{\mathcal{D}(\gamma),\:\mathcal{N}}^t(\mathcal{N}\circ\mathcal{D}(\rho))\|_1^2$.

\subsection{Primitive Quantum Channels and Quantum Markov Chains}\label{subsec:primitive}
A Markov chain describes the evolution of a classical or quantum state over multiple iterations of a channel. Contraction coefficients were first introduced to study the evolution of the relative entropy in a classical Markov chain \cite{Cohen1993}, and were later applied to Markov chains based on ergodic quantum channels, which satisfy a quantum detailed balance condition with respect to some steady states \cite{Temme_2010}. Specifically, they are used to upper bound the mixing time --- the number of iterations it takes for the image of the channel to be contained within an $\varepsilon$-net of the fixed points. 

Quantum primitive channels are precisely quantum channels $\mathcal{N}:\mathcal{B}(\mathcal{H})\to\mathcal{B}(\mathcal{H})$ such that for sufficiently large $m\in\mathbb{N}$, $\mathrm{Im}\:\mathcal{N}^m$ contains only full rank states, i.e. $\lambda_\mathrm{min}(\mathcal{N}^m)>0$. They can be equivalently characterised as quantum channels with a unique full-rank fixed point \cite{SGWC10}. When trying to bound the trace distance of the evolution of a state to the steady state, Riemannian contraction and expansion coefficients allow easy access to the spectral properties of the quantum channel, as we will discuss. Studying the convergence rate via expansion coefficients is a new contribution of this work. 

There are some comments that we can make on upper bounding the convergence rate of discrete time-homogeneous quantum Markov chains, based on Section~\ref{thm:RiemInequiv}, so we recall the following convergence theorem:

\begin{theorem}[Quantum Markov Convergence Theorem]\cite[Proposition 24]{GT25}\label{thm:convergencethm}
    Let $\rho^*\in\mathcal{D}^+(\mathcal{H})$ be the fixed point of the primitive quantum channel $\mathcal{N}:\mathcal{B}(\mathcal{H})\to\mathcal{B}(\mathcal{H})$. For any $\kappa\in\mathcal{K},n\in\mathbb{N}$,
    \begin{equation}\label{eqn:convergencethm}
        \|\mathcal{N}^n(\rho)-\rho^*\|_1\leq\eta_\kappa^\mathrm{Riem}(\mathcal{N},\rho^*)^{n/2}\|\rho-\gamma\|_{\kappa,\rho^*}
    \end{equation}
    where
    \begin{equation*}
        \eta_\kappa^\mathrm{Riem}(\mathcal{N},\rho^*):=\inf_{\substack{\rho\in \mc D(\mathcal{H}_A)}\setminus\{\rho^*\}}\frac{\|\mathcal{N}(\rho)-\rho^*\|_{\kappa,\:\rho^*}^2}{\|\rho-\rho^*\|_{\kappa,\:\rho^*}^2}
    \end{equation*}
    In particular, for $\|\mathcal{N}^n(\rho)-\rho^*\|_1\leq\delta$ for all $\rho\in\mathcal{D}(\mathcal{H})$, it suffices to have $n\geq \log\left(\frac{2}{\delta^2\lambda_{\min}}\right)/\log\left(1/\eta_\kappa^\mathrm{Riem}(\mathcal{N},\rho^*)\right)$. In other words:
    \begin{equation*}
        t_{\mathrm{mix}}\leq\log\left(\frac{2}{\delta^2\lambda_{\min}}\right)/\log\left(1/\eta_\kappa^\mathrm{Riem}(\mathcal{N},\rho^*)\right)
    \end{equation*}
\end{theorem}

There is a reason to be cautious about the choice of $\kappa$ used to bound the convergence rate. For example, in \eqref{eqn:inequivalence}, family of primitive channels \eqref{eqn:counter} is  classical-quantum, i.e. the adjoint's image is a commutative subalgebra, which implies that smaller choices of $\kappa$ will always be optimal \cite[Corollary 5.6]{HR15}. We saw from this example that different Riemannian contraction coefficients can be several orders of magnitude apart for some primitive channels, although the analysis turns out to be based on the behaviour near a point that isn't close to the fixed point; the extent of the differences in $\eta_\kappa^\mathrm{Riem}(\mathcal{N},\rho^*)$ specifically could be worth some further investigation. 

We met in Section~\ref{sec:fullrankoutputchannels} that there is an equivalence $\widecheck\eta_f^\mathrm{std}\cong\widecheck\eta_\kappa^\mathrm{Riem}$ for all $f\in\mathcal{F},\kappa\in\mathcal{K}$ for the classes $\mathcal{Q}_\lambda$ of strictly positive quantum channels. It turns out that there is a further redundancy in relative expansion coefficients in this case, for which we define relative expansion coefficients for the Schatten 2-norm:

\begin{equation*}
    \widecheck\eta_2(\mathcal{N},\mathcal{M}):=
    \inf_{\substack{\rho\neq \gamma \in \mc D(\mathcal{H}_A),\\\supp(\rho) = \supp (\gamma)}} \frac{\|\mc N(\rho) -\mc N(\gamma)\|_2}{\|\mc M(\rho) -\mc M(\gamma)\|_2}
\end{equation*}

We will first establish another equivalence result that, actually, $\widecheck\eta_f^\mathrm{std}\cong\widecheck\eta_\kappa^\mathrm{Riem}\cong \eta_2$. This arises from the following generalisation of \cite[Lemma 2.1]{GR22} from the BKM metric ($\kappa_\mathrm{BKM}(x):=\frac{\log x}{x-1}$)

\begin{lemma}\label{lem6}
	$ $\newline 
    Given two density operators $\rho, \gamma \in \mathcal{D}(\mathcal{H})$, some finite-dimensional Hilbert space $\mathcal{H}$, and some $c \in (0, \infty), \kappa \in \mathcal{K}$:
	\[
	\rho \leq c \gamma \implies \| X \|_{\kappa, \gamma}^2 \leq c \| X \|_{\kappa, \rho}^2 \quad \text{for all traceless Hermitian operators } X \in \mathcal{B}(\mathcal{H})
	\]
\end{lemma}

\begin{proof}
    By the integral representation of $\kappa$, corresponding to the probability measure $m$ on $[0,1]$ \eqref{opmon}:
    \begin{align*}
            \| X \|_{\kappa, \rho}^2 
            &= \int_{[0,1]} \left\langle X, \frac{1}{L_\rho + s R_\rho}(X) \right\rangle_{HS}\; dm(s) \\
            &\geq \frac{1}{c}\int_{[0,1]} \left\langle X, \frac{1}{L_\gamma + s R_\gamma} (X)\right\rangle_{HS}\; dm(s) \\
            &= \frac{1}{c}\| X \|_{\kappa, \gamma}^2
        \end{align*}
	where we used the fact that $x^{-1}$ is operator monotone decreasing on $(0, \infty)$ to obtain $\frac{1}{L_\rho + s R_\rho} \geq \frac{1}{L_{c\gamma} + s R_{c\gamma}}$.\\
\end{proof}

As a special case, we obtain a similar result to \cite[Proposition 5.4]{BGSW24} (which, again, worked only with $\kappa_\mathrm{BKM}$): 
\begin{lemma}\label{lem:2norm}
    $ $\newline 
    Suppose we are given quantum channels $\mathcal{N}:\mathcal{B}(\mathcal{H}_A)\to\mathcal{B}(\mathcal{H}_B),\:\mathcal{M}:\mathcal{B}(\mathcal{H}_A)\to\mathcal{B}(\mathcal{H}_B')$ and some $\kappa\in\mathcal{K}$, then for all density operators $\rho\in\mathcal{D}(\mathcal{H})$, traceless Hermitian operators $X\in\mathcal{B}(\mathcal{H}):$
    \begin{align*}
        \lambda_\mathrm{max}^{-1}(\mathcal{N})\|\mathcal{N}(X)\|_2^2 &\leq \|\mathcal{N}(X)\|_{\kappa,\:\mathcal{N}(\rho)}^2\leq\lambda_\mathrm{min}^{-1}(\mathcal{N})\|\mathcal{N}(X)\|_2^2\;,\\\\
        \lambda_\mathrm{max}^{-1}(\mathcal{M})\|\mathcal{M}(X)\|_2^2 &\leq \|\mathcal{M}(X)\|_{\kappa,\:\mathcal{M}(\rho)}^2\leq\lambda_\mathrm{min}^{-1}(\mathcal{M})\|\mathcal{M}(X)\|_2^2
    \end{align*}

    \noindent In particular, 
    \begin{equation*}
        \widecheck\eta_\kappa^\mathrm{Riem}(\mathcal{N},\mathcal{M})\cong_{\alpha,\beta}\widecheck\eta_2(\mathcal{N},\mathcal{M})
    \end{equation*}
    Where $\alpha=\frac{\lambda_\mathrm{min}(\mathcal{M})}{\lambda_\mathrm{max}(\mathcal{N})},\beta=\frac{\lambda_\mathrm{max}(\mathcal{M})}{\lambda_\mathrm{min}(\mathcal{N})}$
    
\end{lemma}
\begin{proof}
    Apply Lemma~\ref{lem6} considering \[     \lambda_\mathrm{min}(\mathcal{N})I_{d_B}\leq\mathcal{N}(\rho)\leq\lambda_\mathrm{max}(\mathcal{N})I_{d_B}     \] \[     \lambda_\mathrm{min}(\mathcal{M})I_{d_B'}\leq\mathcal{M}(\rho)\leq\lambda_\mathrm{max}(\mathcal{M})I_{d_B'}     \] and, e.g. $\|\mathcal{N}(X)\|_{\kappa,\:I_{d_B}}^2\equiv\|\mathcal{N}(X)\|_2^2$ .
\end{proof}

 \cite{LR99} conjectured for unital channels that $\hat\eta_\kappa^\mathrm{Riem}(\mathcal{N})\equiv\hat\eta_2(\mathcal{N})$. For unital channels with a unique fixed point (the maximally mixed state), the above result gives an equivalence, that is sometimes close (recall $\hat\eta_\kappa^\mathrm{Riem}(\mathcal{N}):=\widecheck\eta_\kappa^\mathrm{Riem}(id_{\mathcal{B}(\mathcal{H})},\mathcal{N})^{-1},\hat\eta_2(\mathcal{N}):=\widecheck\eta_2(id_{\mathcal{B}(\mathcal{H})},\mathcal{N})^{-1}$).

\begin{theorem}\label{thm:primcoeffpos}
    Let $\mathcal{N}:\mathcal{B}(\mathcal{H})\to\mathcal{B}(\mathcal{H})$ be primitive quantum channel with fixed point $\rho^*$ and $\kappa\in\mathcal{K}$. There exists $M\in\mathbb{N}$ s.t. for all $m\geq M$:
    \begin{equation}
        \widecheck\eta_\kappa^\mathrm{Riem}(\mathcal{N},\rho^*;\mathrm{Im}\:\mathcal{N}^{m-1}):=\underset{\rho\in\mathcal{D}(\mathcal{H})}{\mathrm{inf}}\;\frac{\|\mathcal{N}^{m}(\rho)-\rho^*\|_{\kappa,\rho^*}^2}{\|\mathcal{N}^{m-1}(\rho)-\rho^*\|_{\kappa,\rho^*}^2}\geq\widecheck\eta_\kappa^\mathrm{Riem}(\mathcal{N};\mathrm{Im}\:\mathcal{N}^{m-1})>0
    \end{equation}
\end{theorem}
\begin{proof}
    Since $\mathcal{N}$ is a primitive quantum channel, for $m$ sufficiently large, $\mathcal{N}^{m-1}$ satisfies $\lambda_\mathrm{min}(\mathcal{N}^{m-1})>0$. Further, any quantum channel $\mathcal{N}:\mathcal{B}(\mathcal{H})\to\mathcal{B}(\mathcal{H})$ is injective on $\mathrm{Im}\:\mathcal{N}^{m-1}$ for $m$ sufficiently large, by the first isomorphism theorem and the finiteness of $\dim\mathcal{H}$. 
    
    By Lemma~\ref{lem:2norm}, we have the result. 
\end{proof}

While Theorem~\ref{nordpi} tells us that $\widecheck\eta_f^\mathrm{std}(\mathcal{N};\mathcal{D}(\mathcal{H}))=\widecheck\eta_\kappa^\mathrm{Riem}(\mathcal{N};\mathcal{D}(\mathcal{H}))=0$ for all $f\in\mathcal{F},\kappa\in\mathcal{K}$, by Theorem~\ref{thm:primcoeffpos} we can constrain the domain for a positive expansion coefficient. We can now use this expansion coefficient as a non-trivial lower bound on the convergence rate of the Markov chain based on the quantum primitive channels $\mathcal{N}$:

\begin{corollary}[Reverse Quantum Markov Convergence Theorem]\label{cor:qmkv}
    Let $\mathcal{N}:\mathcal{B}(\mathcal{H})\to\mathcal{B}(\mathcal{H})$ be a primitive quantum channel with fixed point $\rho^*$ and $\kappa\in\mathcal{K}$. There exists $M\in\mathbb{N}$ s.t. for all $m\geq M$:
    \begin{align*}
        \|\mathcal{N}^m(\rho)-\rho^*\|_1&\geq\lambda^{1/2}_\mathrm{min}(\mathcal{N}^m)\eta_\kappa^\mathrm{Riem}(\mathcal{N},\rho^*;\mathrm{Im}\:\mathcal{N}^{M-1})^\frac{m-M+1}{2}\|\mathcal{N}^{M-1}(\rho)-\rho^*\|_{\kappa,\:\rho^*}
    \end{align*}
\end{corollary}
\begin{proof}
    \begin{align*}
        \|\mathcal{N}^m(\rho)-\rho^*\|_1&\geq\|\mathcal{N}^m(\rho)-\rho^*\|_2\geq\lambda_\mathrm{min}^{1/2}(\mathcal{N}^m)\|\mathcal{N}^m(\rho)-\rho^*\|_{\kappa,\:\rho^*}\\
        &\geq\lambda_\mathrm{min}^{1/2}(\mathcal{N}^m)[\widecheck \eta_\kappa^\mathrm{Riem}(\mathcal{N},\rho^*;\mathrm{Im}\:\mathcal{N}^{m-1})...\eta_\kappa^\mathrm{Riem}(\mathcal{N},\rho^*;\mathrm{Im}\:\mathcal{N}^{M-1})]^{1/2}
        \|\mathcal{N}^{M-1}(\rho)-\rho^*\|_{\kappa,\:\rho^*}
        \\&\geq\lambda_\mathrm{min}^{1/2}(\mathcal{N}^m)\eta_\kappa^\mathrm{Riem}(\mathcal{N},\rho^*;\mathrm{Im}\:\mathcal{N}^{M-1})^\frac{m-M+1}{2}\|\mathcal{N}^{M-1}(\rho)-\rho^*\|_{\kappa,\:\rho^*}
    \end{align*}
\end{proof}

\noindent This result accompanies Theorem~\ref{thm:convergencethm} and \cite[Theorem 9]{Temme_2010}, which demonstrate an exponential upper bound on $\|\mathcal{N}^m(\rho)-\rho^*\|_1$ in terms of a contraction coefficient. By standard analysis, one can get a state-dependent lower bound on the (relative) mixing time. We cannot expect much more from a `reverse quantum Markov convergence theorem', since the cut-off behaviour \cite{kastoryano2012cutoff} means that the quantum Markov chain has non-asymptotic behaviour and the trace distance only decays exponentially after a number of iterations. Further, $\mathrm{lim}_{M\to\infty}\eta_\kappa^\mathrm{Riem}(\mathcal{N},\rho^*;\mathrm{Im}\:\mathcal{N}^{M-1})$ is by definition the least upper bound for the asymptotic convergence rate.

\color{black}

\subsection{Explicit Demonstrations of Positive Riemannian Relative Expansion Coefficients}\label{sec:explicitcoeffs}
% Following the discussion in the previous subsection, it is important to derive explicit estimates for the relative expansion coefficient for a pair of quantum channels, since it induces a mixing time lower bound. 

% \paula{it should be emphasized again that the most straightforward thing does not hold bc of inequivalence and inequality. also, what do our results about equivalence mean for this section}

Recall that, by Theorem~\ref{nordpi}, for the many quantum channels $\mathcal{N}$ whose output dimension is no more than the input dimension, the divergence and Riemannian relative expansion coefficients are all zero. The goal of this section is to produce examples where all of the Riemannian relative expansion coefficients are positive. 

We will apply some of the theory from previous sections to demonstrate parametrised families of channels on $\mathcal{D}(\mathcal{H}), \{\Phi_x\}_{x\in\mathcal{X}}, \mathcal{X} $ some valid parameter family, such that for all $x, x''\in\mathcal{X}, \Phi_{x'}\circ\Phi_x=\Phi_{x''}$ for some $x'\in\mathcal{X}$, the channels are perhaps sufficiently similar that we obtain positive relative expansion coefficients $\widecheck{\eta}^{\mathrm{Riem}}(\Phi_x, \Phi_{x''})>0$. If $\Phi_{x''}=\Phi_{x'}\circ\Phi_x$ for some $x' \in \mathcal{X}$, this indeed tells us that $\Phi_{x'}$ uniformly preserves a certain proportion of the distinguishability of states in $\Phi_x(\mathcal{D}(\mathcal{H}))$. Except when we are dealing with primitive channels, it is not clear how we can show the positivity of $\widecheck{\eta}_f^{\mathrm{std}}(\Phi_x, \Phi_{x''}), f\in \mathcal{F}_{\mathrm{sym}}$, in general. 

This section extends the calculations by \cite{BGSW24} for the BKM metric, via two observations. The first observation is that the properties they used for the corresponding Riemannian semi-norm are not unique, but apply to all Riemannian semi-norms. The second observation is that, for qubit channels, any Riemannian relative expansion coefficient $\widecheck\eta_\kappa,\kappa\in\mathcal{K}$ can be related to $\widecheck\eta_{\kappa_{\max}}$, giving a reduction for demonstrating positivity.

Even when dealing with Riemannian expansion coefficients, computation is still difficult, so we mostly work with qubit channels. In this setting, it is possible to write the Riemannian semi-norms and essentially decompose the problem into a few independent optimisation problems; it helps that all traceless Hermitian operators $X$ acting on qubits satisfy $X^2\:\propto\: I$. 

\subsubsection{Generalised Dephasing Channel}

We already know that when a quantum channel $\mathcal{N}$ is primitive, i.e.\ it has a unique full rank fixed point, there exists another (primitive) quantum channel $\mathcal{M}$ such that
\[
\widecheck{\eta}^{\mathrm{Riem}}_\kappa(\mathcal{N}\circ \mathcal{M}, \mathcal{M}), \;
\widecheck{\eta}^{\mathrm{std}}_f(\mathcal{N}\circ \mathcal{M}, \mathcal{M}) > 0
\quad \text{for all} \; \kappa \in \mathcal{K}, \; f \in \mathcal{F}_{\mathrm{sym}}.
\]
The generalised dephasing channels $\Phi_\Gamma : \mathcal{B}(\mathcal{H}) \to \mathcal{B}(\mathcal{H}), \dim \mathcal{H} = d,$ are parametrised by
\[
\Gamma \in \mathcal{X} := \left\{ \Gamma' \in \mathcal{B}(\mathcal{H}) : \Gamma' \geq 0, \;
\Gamma'_{ij} \in [0,1], \; \Gamma'_{ij} = 1, \; 0 \leq i,j \leq d-1 \right\} \setminus \{ I_d \},
\]
and have the form:
\[
\Phi_\Gamma(\rho) = \Gamma \odot \rho = \sum_{i,j=0}^{d-1} \Gamma_{ij} \rho_{ij} \ket{ii}\bra{jj}
\]

\noindent As generalisations of the qubit dephasing channels, these channels model decoherence and preserve all density operators that are diagonal w.r.t.\ the standard basis. Since at least some of the channels' fixed points are not full rank, they are not primitive. To deal with the relative expansion coefficients in this case, we cater to the specific form of these channels, and compare generalised dephasing channels whose parameters are \emph{close}, exactly as \cite{BGSW24}, did.

\noindent To be able to do this, we consider the following generalisation of \cite[Lemma 4.3]{BGSW24} (note: a slight improvement has been made):

\begin{lemma}\label{lem7}
	Given quantum channels $\mathcal{M}, \mathcal{N}, \Phi \in \mathcal{B}(\mathcal{B}(\mathcal{H}_A), \mathcal{B}(\mathcal{H}_B))$ and $\varepsilon \in (0, \frac{1}{2})$ such that:
	\[
	\mathcal{N} = (1 - \varepsilon)\mathcal{M} + \varepsilon \Phi
	\]
	If there exists a quantum channel $\mathcal{D} \in \mathcal{B}(\mathcal{B}(\mathcal{H}_B), \mathcal{B}(\mathcal{H}_B))$ such that $\Phi = \mathcal{D} \circ \mathcal{N}$ and $\mathcal{D}(\omega) \leq c\omega$ for some fixed density operator $\omega$ and $c > 0$. Then for any operator $X\in \mathcal{B}(\mathcal{H}_A)$ and $\kappa\in\mathcal{K}$:
	\[
	\| \mathcal{N}(X) \|_{\kappa, \omega}^2 \geq \frac{(1-\varepsilon)(1-2\varepsilon)}{1+c\varepsilon(1-2\varepsilon)} \| \mathcal{M}(X) \|_{\kappa, \omega}^2
	\]
\end{lemma}

\begin{proof}
    $ $\newline 
    For $X, Y \in \mathcal{B}(\mathcal{H}_B),\ \mathrm{supp}\;X,\: \mathrm{supp}\;Y \leq \mathrm{supp}\;\omega$: $X, Y \mapsto \langle X, R_\omega^{-1} \kappa(L_\omega R_\omega^{-1}) (Y) \rangle_{HS}$ defines an inner product, then necessarily $X \mapsto \|X\|^2_{\kappa, \rho}$ defines a genuine norm, and in particular the triangle inequality holds.
    
    Hence, by the triangle inequality,
    \[\| \mathcal{N}(X) \|_{\kappa, \omega} 
        \geq (1 - \varepsilon) \| \mathcal{M}(X) \|_{\kappa, \omega} - \varepsilon\| \Phi(X)\|_{k, \omega}\]
    We then square both sides:
    \begin{align*}
        \| \mathcal{N}(X) \|_{\kappa, \omega}^2
        & \geq (1 - \varepsilon)^2 \| \mathcal{M}(X) \|_{\kappa, \omega}^2 - 2 \varepsilon (1 - \varepsilon) \| \mathcal{M}(X)\|_{\kappa,\omega} \|\Phi(X)\|_{\kappa, \omega} +\varepsilon^2 \|\Phi(X)\|^2_{\kappa, \omega} \\ \\
        &\overset{2ab\leq a^2+b^2}{\geq} (1 - \varepsilon)^2 \| \mathcal{M}(X) \|_{\kappa, \omega}^2 - \varepsilon (1 - \varepsilon) (\| \mathcal{M}(X) \|^2_{\kappa, \omega}+\|\Phi(X)\|^2_{\kappa, \omega}) + \varepsilon^2  \|\Phi(X)\|_{\kappa, \omega}^2\\
        %\quad \quad\quad\quad\quad\quad\quad\quad\quad\quad\quad\quad\quad\quad\quad\quad\quad\quad\quad\quad\quad\quad\quad\quad \quad\quad+  \\
        &= (1 - \varepsilon)(1-2\varepsilon) \| \mathcal{M}(X) \|_{\kappa, \omega}^2 - \varepsilon(1-2\varepsilon) \| \Phi(X) \|_{\kappa, \omega}^2 \\ \\
        &\overset{\text{Lemma~\ref{lem6}}}{\geq} (1 - \varepsilon)(1 - 2\varepsilon) \| \mathcal{M}(X) \|_{\kappa, \omega}^2 - c\varepsilon(1 - 2\varepsilon) \| \mathcal{D} \circ \mathcal{N}(X) \|_{\kappa, \mathcal{D}(\omega)}^2 \\
        &\overset{\text{Monotonicity}}{\geq} (1 - \varepsilon)(1 - 2\varepsilon) \| \mathcal{M}(X) \|_{\kappa, \omega}^2 - c\varepsilon(1 - 2\varepsilon) \| \mathcal{N}(X) \|_{\kappa, \omega}^2
    \end{align*}
    
    \noindent This gives the result.
\end{proof}

As a consequence, we can achieve the same result as \cite{BGSW24} about the possibility of a positive Riemannian expansion coefficient, for general choices $\kappa\in\mathcal{K}$:

\begin{prop}\cite{BGSW24}
    $ $\newline
    Let $\Gamma = (\Gamma_{ij}), \Gamma' = (\Gamma'_{ij}) \in \mathcal{X}$. Suppose there exists $\varepsilon \in (0,1)$ such that
    \begin{enumerate}
        \item $(1 - \varepsilon)\Gamma \leq \Gamma' \leq (1 + \varepsilon)\Gamma$
        \item $\hat\Gamma= (\hat{\Gamma}_{ij})$ defined via the following is in $\mathcal{X}$ (in particular, we want $\Gamma'$ PSD):
        \[
        \quad 
        \hat{\Gamma}_{ij} := 
        \begin{cases}
        0, & \text{if } \Gamma'_{ij} = 0, \\
        \frac{\Gamma'_{ij} - (1 - \varepsilon)\Gamma_{ij}}{\varepsilon \Gamma_{ij}'}, & \text{if } \Gamma'_{ij} > 0.
        \end{cases}
        \]
    \end{enumerate}
    
    \noindent Then for any $\kappa\in\mathcal{K}$:
    \[
    \widecheck{\eta}_\kappa^\mathrm{Riem}(\Phi_{\Gamma'}, \Phi_\Gamma) \geq \frac{(1 - 2\varepsilon)(1 - \varepsilon)}{(1 + 2\varepsilon)(1 + \varepsilon)}. 
    \]
\end{prop}

\begin{proof}
    $ $\newline
    The idea is to apply Lemma~\ref{lem7} with $\mathcal{N}=\Phi_{\Gamma'}, \mathcal{M}=\Phi_\Gamma$, because this directly gives a bound $\widecheck{\eta}({\Phi_{\Gamma'}, \Phi_\Gamma}) \geq \frac{(1 - \varepsilon)(1 - 2\varepsilon)}{(1+\varepsilon)(1+c\varepsilon(1-2\varepsilon))}$ if arbitrary choices of $\omega$ in the image $\Phi_{\Gamma'}(\mathcal{D}(\mathcal{H}))$ work.
    
    \noindent Let us construct $\widetilde{\Gamma}_{ij} := \frac{\Gamma'_{ij} - (1 - \varepsilon)\Gamma_{ij}}{\varepsilon}$, $\widetilde{\Gamma} = (\widetilde{\Gamma}_{ij})_{0 \leq i,j \leq d-1}\in\mathcal{X}$ ($\widetilde{\Gamma}=\frac{\Gamma'-(1-\varepsilon)\Gamma}{\varepsilon}\geq 0$ by assumption), then we have
    \[
    \Phi_{\Gamma'} = (1 - \varepsilon) \Phi_\Gamma + \varepsilon \Phi_{\widetilde{\Gamma}},
    \]
    
    \noindent It remains to show that:
    \begin{enumerate}
        \item There exists a quantum channel $\mathcal{D}$ such that $\Phi_{\widetilde{\Gamma}}=\mathcal{D} \circ \Phi_{\Gamma'}$.
        \item There exists a universal constant $c > 0$ such that for any density operator $\gamma$, $\mathcal{D}(\Phi_{\Gamma'}(\gamma)) \leq c\Phi_{\Gamma'}(\gamma)$.
    \end{enumerate}
    
    \noindent For the first point, we note that $\hat\Gamma$ was constructed so that $\hat{\Gamma}\odot\Gamma'=\widetilde{\Gamma}$, which implies:
    \[
    \Phi_{\hat{\Gamma}} \circ \Phi_{\Gamma'} = \Phi_{\widetilde{\Gamma}}.
    \]
    This means that we can take $\mathcal{D} := \Phi_{\widetilde{\Gamma}}$.
    
    \noindent For the other condition, we observe that $c=\frac{2}{1-2\varepsilon}$ is a valid choice:
    \[
    \mathcal{D} \circ \Phi_{\Gamma'} = \Phi_{\widetilde{\Gamma}} = \frac{ \Phi_{\Gamma'} - (1 - \varepsilon)\Phi_\Gamma }{ \varepsilon } \leq_{\text{cp}} \frac{ (1 + \varepsilon)\Phi_\Gamma - (1 - \varepsilon)\Phi_\Gamma }{ \varepsilon } = \frac{2\varepsilon \Phi_\Gamma}{\varepsilon} = 2\Phi_\Gamma \leq_{\text{cp}} \frac{2}{1 - 2\varepsilon} \Phi_{\Gamma'}.
    \]
    
    \noindent Finally, we observe for all $\omega=\Phi_{\Gamma'}(\gamma)$ for some $\gamma\in\mathcal{D}(\mathcal{H})$:
    \begin{align*}
    \|\Phi_{\Gamma'}(X)\|_{\kappa,\Phi_{\Gamma'}(\gamma)}^2
    &\overset{\text{Lemma~\ref{lem7}}}{\geq} \frac{(1 - 2\varepsilon)(1 - \varepsilon)}{1 + 2\varepsilon} \|\Phi_{\Gamma}(X)\|_{\kappa,\Phi_{\Gamma'}(\gamma)}^2 \\
    &\overset{\text{Lemma~\ref{lem6}}}{\geq} \frac{(1 - 2\varepsilon)(1 - \varepsilon)}{(1 + 2\varepsilon)(1 + \varepsilon)} \|\Phi_{\Gamma}(X)\|_{\kappa,\Phi_{\Gamma}(\gamma)}^2.
    \end{align*}
    \noindent Where for the final inequality, we used the fact $\Gamma'\leq(1+\varepsilon)\Gamma$ (and thus $\Phi_\Gamma'\leq_{cp}(1+\varepsilon)\Phi_\Gamma$).
\end{proof}

\subsubsection{Qubit Dephasing Channel and Amplitude Damping Channel}

For the qubit dephasing and amplitude damping channels, some care needs to be taken around their pure fixed points, which make it difficult to compute a lower bound on the relative expansion coefficients. However, the strategy for showing that the relative expansion coefficients are positive is similar for both channels. For $s \in [0,1]$, density operator $\rho = \frac{1}{2} (\mathbb{I}_2 + \mathbf{w}\cdot\sigma) \in \mathcal{D}(\mathbb{C}^2)$, traceless Hermitian operator $X = \mathbf{y} \cdot\sigma \in \mathcal{B}(\mathbb{C}^2)$  \cite[Appendix B.2.7]{HR15},
\begin{align*}
	\| X \|^2_{\kappa_s, \rho} 
	&= \frac{2|\mathbf{y}|^2}{1 - |\mathbf{w}|^2} \cdot \frac{(1 + s^2)(1 - |\mathbf{w}|^2) + 4s |\mathbf{w}|^2 \cos^2 \theta}{(1 + s^2)(1 - |\mathbf{w}|^2) + 4s |\mathbf{w}|^2} \\[1.5ex]
	&= \frac{2|\mathbf{y}|^2}{1 - |\mathbf{w}|^2} \left( h_s (|\mathbf{w}|^2) + (1 - h_s (|\mathbf{w}|^2)) \cos^2 \theta \right) \\[1.5ex]
	&= \| X \|^2_{\kappa_{\max}, \rho} \left( h_s (|\mathbf{w}|^2) + (1 - h_s (|\mathbf{w}|^2)) \cos^2 \theta \right)
\end{align*}

\noindent where $h_s(x) := \frac{(1+s)^2(1-x)}{(1+s)^2(1-x)+4sx}$. Since every $\kappa \in \mathcal{K}$ has the form $\kappa(x) = \int_{[0,1]}\kappa_s(x) \, dm(s)$ for all $x \in (0, \infty)$, we generally have:
\[
\| X \|^2_{\kappa,\rho} = \| X \|^2_{\kappa_{\max},\rho} (h(|\mathbf{w}|^2) +(1-h(|\mathbf{w}|^2)) \cos^2 \theta ),\quad h(x):=\int_{[0,1]} h_s(x) \, dm(s)
\]
\noindent and $h(|\mathbf{w}|^2) = 0 \iff |\mathbf{w}|=1, \text{ otherwise } h(|\mathbf{w}|^2) \in (0,1] \text{ (as } |\mathbf{w}| < 1)$.\\

\noindent If we are considering two qubit quantum channels
\[
\mathcal{N}: \frac{1}{2} (\mathbb{I}_2 + \mathbf{w} \cdot \sigma)\mapsto \frac{1}{2}(\mathbb{I}_2 + (T \mathbf{w}+\mathbf{t}) \cdot \sigma)\; ,\quad 
\mathcal{M}: \frac{1}{2} (\mathbb{I}_2 + \mathbf{w} \cdot \sigma)\mapsto \frac{1}{2}(\mathbb{I}_2 + (T' \mathbf{w}+\mathbf{t}') \cdot \sigma)
\]
the problem of lower bounding the relative Riemannian expansion coefficient reduces slightly:
\[
\widecheck{\eta}^{\mathrm{Riem}}_{\kappa}(\mathcal{N}, \mathcal{M})\ge \widecheck{\eta}_{\kappa_{\max}}^{\mathrm{Riem}}(\mathcal{N}, \mathcal{M})\cdot \inf_{\mathbf{y},\mathbf{w}:|\mathbf{w}|\leq1} \frac{h_{\mathcal{N}}(\mathbf{w},\mathbf{y})}{h_{\mathcal{M}}(\mathbf{w}, \mathbf{y})}
\]

\noindent where we denote $h_{\mathcal{N}}(\mathbf{w}, \mathbf{y}) = |(T \mathbf{w} + \mathbf{t})\cdot T \mathbf{y} / |T \mathbf{y}||^2 (1 - h(|T \mathbf{w} + \mathbf{t}|^2)) + |T \mathbf{w} + \mathbf{t}|^2  h(|T \mathbf{w} + \mathbf{t}|^2)$, and similarly for $\mathcal{M}$.

\begin{prop}\label{prop5}
	Let $\Phi_p(\rho) = \frac{1}{2} (\mathbb{I}_2 + T_p \mathbf{w}\cdot\sigma)$, $T_p := \mathrm{diag}(1-p, 1-p, 1)$, denote the dephasing channel. For $0 < p_2< p_1 < 1$ and any $\kappa\in\mathcal{K}$, we have
	\[
	\widecheck{\eta}_\kappa^{\mathrm{Riem}}(\Phi_{p_1}, \Phi_{p_2}) > 0
	\]
\end{prop}

\begin{proof}
		$ $\newline 
        Denote $\mathbf{w}_p = T_p \mathbf{w}$, $\mathbf{y}_p = T_p \mathbf{y}, \theta_p := \cos ^{-1}\frac{|\mathbf{w}_p\cdot \mathbf{y}_p|}{|\mathbf{w}_p\|\mathbf{y}_p|} \text{ and } S_\varepsilon:=B(\mathbf{e}_3,\varepsilon)^c \cap B(-\mathbf{e}_3, \varepsilon)^c$. For any $\varepsilon > 0$, $c(p, \varepsilon) := \inf_{\mathbf{w} \in S_\varepsilon} h(|\mathbf{w}_p|^2) > 0$ since $S_\varepsilon$ is compact (so the infimum is attained). Thus for any $\mathbf{w} \in S_\varepsilon$:
		\[
		\frac{h_{\Phi_{p_1}}(\mathbf{w}, \mathbf{y})}{h_{ \Phi_{p_2}}(\mathbf{w},\mathbf{y})}=\frac{\cos^2 \theta_{p_1} + \sin^2 {\theta_{p_1}} h(|\mathbf{w}_{p_1}|^2)}{{\cos^2 \theta_{p_2} + \sin^2 {\theta_{p_2}} h(|\mathbf{w}_{p_2}|^2}) }\geq c(p_1, \varepsilon) > 0
		\]
		
		\noindent Now consider $\mathbf{w} \mapsto \pm \mathbf{e}_3\cdot \text{WLOG}$ we consider $|\mathbf{y}|=1$, and define (by minimisation of a continuous function over a compact set):
		\[
		\mathbf{y}(\mathbf{w}) := \operatorname*{arg\,min}_{|\mathbf{y}|=1} \frac{h_{\Phi_{p_1}}(\mathbf{w}, \mathbf{y})}{h_{\Phi_{p_2}}(\mathbf{w}, \mathbf{y})}
		\]
		\noindent and $\mathbf{y}_{p_i}(\mathbf{w}) :=T_{p_i}\mathbf{y}(\mathbf{w})$.\\ \\
		To simplify the expression for $\mathbf{y}(\mathbf{w})$:
		\begin{align*}
			|\mathbf{w}_{p_2}\cdot \mathbf{y}_{p_2}|^2 &= |(1-p_2)^2 (y_1 w_1 + y_2 w_2) + y_3 w_3|^2 \\ \\
			&= |\mathbf{w}_{p_1}\cdot \mathbf{y}_{p_1}+((1-p_2)^2-(1-p_1)^2)(y_1w_1+y_2w_2)|^2 \\ \\			
			&\leq  2|\mathbf{w}_{p_1}\cdot \mathbf{y}_{p_1}|^2 + 2((1-p_2)^2 - (1-p_1)^2)^2|y_1 w_1 + y_2 w_2|^2 \\ \\
			&\overset{C-S}{\leq} 2|\mathbf{w}_{p_1}\cdot \mathbf{y}_{p_1}|^2 + 2(2 - p_1 - p_2)^2(p_1 - p_2)^2(y_1^2 + y_2^2)(w_1^2 + w_2^2) \\ \\
			&\leq 2|\mathbf{w}_{p_1}\cdot \mathbf{y}_{p_1}|^2 + \frac{2(2-p_1-p_2)^2(p_1-p_2)^2}{(1-p_2)^2} \cdot |\mathbf{y}_{p_2}|^2 \cdot \frac{|\mathbf{w}|^2 - |\mathbf{w}_{p_2}|^2}{p_2(2-p_2)} \\ \\
			&\leq 2|\mathbf{w}_{p_1}\cdot \mathbf{y}_{p_1}|^2 + \frac{2(2-p_1-p_2)^2(p_1-p_2)^2}{(1-p_2)^2 p_2 (2-p_2)} \cdot |\mathbf{y}_{p_2}|^2 \cdot (1 - |\mathbf{w}_{p_2}|^2) \\ \\
			&\leq 2|\mathbf{w}_{p_1}\cdot \mathbf{y}_{p_1}|^2 + \frac{2(2-p_1-p_2)^2(p_1-p_2)^2}{(1-p_2)^2 p_2 (2-p_2)} \cdot |\mathbf{y}_{p_2}|^2 \cdot h(|\mathbf{w}_{p_2}|^2)
		\end{align*}

		\noindent Where we used the following facts:
        \begin{enumerate}
            \item $|\mathbf{y}_{p_2}|^2 \geq (\mathbf{y}_{p_2})^2_1 + (\mathbf{y}_{p_2})^2_2 = (1 - p_2)^2 (y_1^2 + y_2^2)$
            \item $1 - |\mathbf{w}_{p_2}|^2 = 1 - (1-p_2)^2(w_1^2 + w_2^2) - w_3^2
        =1-|\mathbf{w}|^2+p_2(2-p_2)(w_1^2+w_2^2)$
            \item $h(x) \geq h_1(x) = \frac{4(1-x)}{4(1-x) + 4x}=1-x \quad \forall x \in [0,1]$
        \end{enumerate}
		\noindent Define $c'(p_1,p_2):=\frac{(2-p_1-p_2)^2(p_1-p_2)^2}{2(1-p_2)^2p_2(2-p_2)}$, then:
		% \begin{align*}
		% 	\liminf_{\substack{w\to \pm \mathbf{e}_3,\\|w|\leq1}} \inf_{|y|=1} \frac{h_{\Phi_{p_1}}(w,y)}{h_{\Phi_{p_2}}(w,y)}
		% 	&\geq \min \Bigg\{
		% 	\liminf_{|w|=1, |y|=1} \inf_{|y|=1} \frac{1-h_{A_2}(|w|^2)}{1-h_{B_2}(|w|^2)}, \liminf_{|w|=1, |y|=1} \inf_{|y|=1} \frac{h_{A_2}(|w|^2)}{h_{B_2}(|w|^2)}
		% 	\Bigg\} \\
		% 	&\geq \min \left\{
		% 	\frac{1}{2},\ \frac{1}{c'(\rho_1,\rho_2)}
		% 	\right\} > 0.
		% \end{align*}
		\begin{align*}
			&\liminf_{\substack{\mathbf{w}\to \pm \mathbf{e}_3,\\|\mathbf{w}|\leq1}} \inf_{|\mathbf{y}|=1} \frac{h_{\Phi_{p_1}}(\mathbf{w},\mathbf{y})}{h_{\Phi_{p_2}}(\mathbf{w},\mathbf{y})}\\
			&\quad\geq \liminf_{\substack{\mathbf{w}\to \pm \mathbf{e}_3,\\|\mathbf{w}|\leq 1}} \inf_{|\mathbf{y}|=1} 
			\frac{|\mathbf{w}_{p_1} \cdot \mathbf{y}_{p_1}/|\mathbf{y}_{p_1}||^2 (1 - h(|\mathbf{w}_{p_1}|^2)) + |\mathbf{w}_{p_1}|^2 h(|\mathbf{w}_{p_1}|^2)}
			{2|\mathbf{w}_{p_1} \cdot \mathbf{y}_{p_1}/|\mathbf{y}_{p_2}||^2 (1 - h(|\mathbf{w}_{p_2}|^2)) + (|\mathbf{w}_{p_2}|^2 + c'(p_1, p_2)) h(|\mathbf{w}_{p_2}|^2)}
			\\[1.5ex]
			&\geq \min \left\{
			\liminf_{\substack{\mathbf{w}\to \pm \mathbf{e}_3,\\|\mathbf{w}|\leq1}} \inf_{|\mathbf{y}|=1} \frac{|\mathbf{y}_{p_2}|^2}{2|\mathbf{y}_{p_1}|^2} \cdot\frac{1 - h(|\mathbf{w}_{p_1}|^2)}{1-h(|\mathbf{w}_{p_2}|^2)},
			\frac{1}{1 + c'(p_1, p_2)} 
			\liminf_{\substack{\mathbf{w}\to \pm \mathbf{e}_3,\\|\mathbf{w}|\leq1}} \frac{h(|\mathbf{w}_{p_1}|^2)}{h(|\mathbf{w}_{p_2}|^2)}
			\right\}\\
			&= \min \left\{ \frac{1}{2}, \ \frac{1}{1 + c'(p_1, p_2)} \right\} > 0
		\end{align*}

		Therefore, for $\varepsilon > 0 \text{ sufficiently small},\ 
		\inf_{\mathbf{w} \in B(\mathbf{e}_3, \varepsilon) \cup B(-\mathbf{e}_3, \varepsilon)}
		\frac{h_{\Phi_{p_1}}(\mathbf{w}, \mathbf{y})}{h_{\Phi_{p_2}}(\mathbf{w}, \mathbf{y})} > 0.$
		
		\noindent Now, we only have to show $\widecheck\eta_{\kappa_{\max}}^{\mathrm{Riem}} \left( \Phi_{p_1}, \Phi_{p_2} \right) > 0:$
		
		\begin{align*}
			\widecheck\eta_{\kappa_{\max}}^{\mathrm{Riem}} \left( \Phi_{p_1}, \Phi_{p_2} \right)
			&\geq \inf_{\mathbf{y} : |\mathbf{y}_{p_1}| = 1} \frac{|\mathbf{y}_{p_1}|^2}{|\mathbf{y}_{p_2}|^2}
			\cdot \inf_{\mathbf{w}:|\mathbf{w}|\leq 1} \frac{1 - |\mathbf{w}_{p_2}|^2}{1 - |\mathbf{w}_{p_1}|^2} \\
			&= \left( \frac{1-p_1}{1-p_2} \right)^2
			\cdot \inf_{\mathbf{w}:|\mathbf{w}|\leq 1} \frac{1 - |\mathbf{w}|^2 + p_2 (2-p_2)(w_1^2 + w_2^2)}
			{1 - |\mathbf{w}|^2 + p_1 (2-p_1)(w_1^2 + w_2^2)} \\
			&\geq \left( \frac{1-p_1}{1-p_2} \right)^2
			\min \left\{ 1, \frac{p_2 (2-p_2)}{p_1 (2-p_1)} \right\} > 0
		\end{align*}

\end{proof}

\begin{prop} \label{prop6}Let $\mathcal{A}_\gamma(\rho) := \frac{1}{2}(\mathbb{I}_2 + (T_\gamma \mathbf{w} + \mathbf{t}_\gamma)\cdot \sigma)$, $T_\gamma = \operatorname{diag}(\sqrt{1-\gamma}, \sqrt{1-\gamma}, 1-\gamma)$, $t_\gamma = \gamma \mathbf{e}_3$, denote the amplitude damping channel. For $0<\gamma_2 < \gamma_1 <1$, and any $\kappa\in\mathcal{K}$ we have
\[
\widecheck{\eta}_\kappa^{\mathrm{Riem}}(\mathcal{A}_{\gamma_1},\mathcal{A}_{\gamma_2}) > 0.
\]
\end{prop}
\begin{proof}
$ $\newline
This time, denote 
\[
\mathbf{w}_\gamma := T_\gamma \mathbf{w} + \mathbf{t}_\gamma = (\sqrt{1-\gamma} w_1, \sqrt{1-\gamma} w_2, (1-\gamma) w_3 + \gamma)
\] 
\[
\mathbf{y}_\gamma := T_\gamma \mathbf{y}_1, (\sqrt{1-\gamma} y_1, \sqrt{1-\gamma} y_2, (1-\gamma) y_3),
\]
\[\theta_\gamma := \cos^{-1}\frac{|\mathbf{w}_\gamma\cdot \mathbf{y}_\gamma|}{|\mathbf{w}_\gamma||\mathbf{y}_\gamma|},
\] and $S_\varepsilon := B(\mathbf{e}_3,\varepsilon)^c$. \\

For any $\varepsilon > 0, c(\gamma,\varepsilon):=\inf_{\mathbf{w}\in S_\varepsilon} h(|\mathbf{w}_\gamma|^2) > 0$ since $S_\varepsilon$ is compact (so the infimum is attained). Thus for any $\mathbf{w} \in S_\varepsilon$:
\[
\frac{h_{\mathcal{A}_{\gamma_1}}(\mathbf{w},\mathbf{y})}{h_{\mathcal{A}_{\gamma_2}}(\mathbf{w},\mathbf{y})} 
= \frac{\cos^2 \theta_{\gamma_1} + \sin^2 \theta_{\gamma_1} h(|\mathbf{w}_{\gamma_1}|^2)}
{\cos^2 \theta_{\gamma_2} + \sin^2 \theta_{\gamma_2} h(|\mathbf{w}_{\gamma_2}|^2)} 
\geq c(\gamma_1,\varepsilon) > 0
\]
Now consider $\mathbf{w} \to \pm \mathbf{e}_3$, WLOG we consider $|\mathbf{y}|=1$, and define (by the minimisation of a continuous function over a compact set)
\[
\mathbf{y}(\mathbf{w}) := \operatorname*{argmin}_{|\mathbf{y}|=1} \frac{h_{\mathcal{A}_{\gamma_1}}(\mathbf{w},\mathbf{y})}{h_{\mathcal{A}_{\gamma_2}}(\mathbf{w},\mathbf{y})}
\]
and $\mathbf{y}_{\gamma_i}(\mathbf{w}) := T_{\gamma_i} \mathbf{y}(\mathbf{w})$\\

\noindent To simplify the expression for $\mathbf{y}(\mathbf{w})$:
\begin{align*}
	|\mathbf{w}_{\gamma_2}\cdot \mathbf{y}_{\gamma_2}|^2 &= |(1-\gamma_2)(\mathbf{w}\cdot \mathbf{y}+\gamma_2 y_3(1-w_3))|^2 \\
	&= \left(\frac{1-\gamma_2}{1-\gamma_1}\right)^2|(1-\gamma_1) (\mathbf{w}\cdot \mathbf{y}) + \gamma_1y_3(1-w_3)+(\gamma_2-\gamma_1)y_3(1-w_3))|^2 \\ \\
	&\overset{|a+b|^2\leq2(|a|^2+|b|^2)}{\leq} \big(\frac{1-\gamma_2}{1-\gamma_1}\big)^2(|\mathbf{w}_{\gamma_1}\cdot \mathbf{y}_{\gamma_1}|^2 +(1-\gamma_2)^2(\gamma_1-\gamma_2)^2y_3^2(1-w_3)^2)\\
	&\leq 2 \left(\frac{1-\gamma_2}{1-\gamma_1}\right)^2 (|\mathbf{w}_{\gamma_1} \cdot \mathbf{y}_{\gamma_1}|^2 
	+ \left(\frac{\gamma_1-\gamma_2}{\gamma_2}\right)^2 |\mathbf{y}_{\gamma_2}|^2  (|\mathbf{w}|^2 - |\mathbf{w}_{\gamma_2}|^2))\\
	&\leq 2 \left(\frac{1-\gamma_2}{1-\gamma_1}\right)^2 (|\mathbf{w}_{\gamma_1} \cdot \mathbf{y}_{\gamma_1}|^2 
	+ \left(\frac{\gamma_1-\gamma_2}{\gamma_2}\right)^2 |\mathbf{y}_{\gamma_2}|^2  (1 - |\mathbf{w}_{\gamma_2}|^2))\\
	&\leq 2 \left(\frac{1-\gamma_2}{1-\gamma_1}\right)^2 (|\mathbf{w}_{\gamma_1} \cdot \mathbf{y}_{\gamma_1}|^2 
	+ \left(\frac{\gamma_1-\gamma_2}{\gamma_2}\right)^2 |\mathbf{y}_{\gamma_2}|^2  h(|\mathbf{w}_{\gamma_2}|^2))
\end{align*}

Where (similar to the dephasing channel), we used the following facts:
\begin{enumerate}
    \item $|\mathbf{y}_{\gamma_2}|^2 \geq (1-\gamma_2)^2 y_3^2$
    \item $1 - |\mathbf{w}_{\gamma_2}|^2 = 1 - |\mathbf{w}|^2 + \gamma_2(w_1^2 + w_2^2) + 2\gamma_2 w_3 (1 - w_3) + \gamma_2^2 (1-w_3)^2
    \geq 1 - |\mathbf{w}|^2 + \gamma_2^2 (1-w_3)^2$
    \item $h(x) \geq h_1(x) = \frac{4(1-x)}{4(1-x) + 4x}=1-x \quad \forall x \in [0,1]$
\end{enumerate}

\noindent $ $\\
%==
\noindent Define $c_1'(\gamma_1, \gamma_2) = 2\left(\frac{1 - \gamma_2}{1-\gamma_1}\right)^2$ and $c_2'(\gamma_1, \gamma_2) = 2\left( \frac{\gamma_1 - \gamma_2}{\gamma_2} \right)^2$, then:\\
\begin{align*}
	&\liminf_{\substack{\mathbf{w} \to \mathbf{e}_3,\\ |\mathbf{w}| \leq 1}} \inf_{|\mathbf{y}| = 1}
	\frac{h_{\mathcal{A}_{\gamma_1}}(\mathbf{w}, \mathbf{y})}{h_{\mathcal{A}_{\gamma_2}}(\mathbf{w}, \mathbf{y})}\\
	&\quad\quad= \liminf_{\substack{\mathbf{w} \to \mathbf{e}_3,\\ |\mathbf{w}| \leq 1}} \inf_{|\mathbf{y}| = 1}
	\frac{ |\mathbf{w}_{\gamma_1}\cdot \mathbf{y}_{\gamma_1}/|\mathbf{y}_{\gamma_1}||^2(1-h(|\mathbf{w}_{\gamma_1}|^2)) + |\mathbf{w}_{\gamma_1}|^2 h(|\mathbf{w}_{\gamma_1}|^2) }
	{ c_1'(\gamma_1, \gamma_2)|\mathbf{w}_{\gamma_1} \cdot \mathbf{y}_{\gamma_1}/|\mathbf{y}_{\gamma_2}||^2(1-h(|\mathbf{w}_{\gamma_2}|^2)) + (|\mathbf{w}_{\gamma_2}|^2 + c_2'(\gamma_1, \gamma_2)) h(|\mathbf{w}_{\gamma_2}|^2) }
	\\[2ex]
	&\geq \min \left\{\liminf_{\substack{\mathbf{w} \to \mathbf{e}_3,\\ |\mathbf{w}| \leq 1}} \inf_{|\mathbf{y}| = 1}\frac{1}{c_1'(\gamma_1, \gamma_2)} \cdot \frac{|\mathbf{y}_{\gamma_2}|^2}{|\mathbf{y}_{\gamma_1}|^2} \cdot
	\frac{1-h(|\mathbf{w}_{\gamma_1}|^2)}{1-h(|\mathbf{w}_{\gamma_2}|^2)}, \frac{1}{1+c_2'(\gamma_1, \gamma_2)}
	\liminf_{\substack{\mathbf{w} \to \mathbf{e}_3,\\ |\mathbf{w}| \leq 1}} \inf_{|\mathbf{y}| = 1}
	\frac{h(|\mathbf{w}_{\gamma_1}|^2)}{h(|\mathbf{w}_{\gamma_2}|^2)}\right\}\\
	&\geq \min \left\{
	\frac{1}{c_1'(\gamma_1, \gamma_2)},\, \frac{1}{1+c_2'(\gamma_1, \gamma_2)}
	\right\} > 0
\end{align*}

\noindent
Therefore, for $\varepsilon > 0$ sufficiently small,
\[
\inf_{\mathbf{w} \in B(\mathbf{e}_3, \varepsilon)}
\frac{h_{\mathcal{A}_{\gamma_1}}(\mathbf{w}, \mathbf{y})}{h_{\mathcal{A}_{\gamma_2}}(\mathbf{w}, \mathbf{y})} > 0
\]

\noindent
Now, we only have to show $\widecheck{\eta}^{\mathrm{Riem}}_{\kappa_{\max}}(\mathcal{A}_{\gamma_1}, \mathcal{A}_{\gamma_2}) > 0$:

\begin{align*}
	\widecheck{\eta}^{\mathrm{Riem}}_{\kappa_{\max}}(\mathcal{A}_{\gamma_1}, \mathcal{A}_{\gamma_2})
	&\ge \inf_{\mathbf{y}:|\mathbf{y}| = 1}
	\frac{|\mathbf{y}_{\gamma_1}|^2}{|\mathbf{y}_{\gamma_2}|^2}
	\inf_{\mathbf{w}:|\mathbf{w}| \leq 1} \frac{1 - |\mathbf{w}_{\gamma_2}|^2}{1 - |\mathbf{w}_{\gamma_1}|^2} \\
	&= \left(\frac{1-\gamma_1}{1-\gamma_2}\right)^2 \cdot \inf_{\mathbf{w}:|\mathbf{w}| \leq 1}
	\frac{1 - |\mathbf{w}_{\gamma_2}|^2}{1 - |\mathbf{w}_{\gamma_1}|^2} \\
	&= \left(\frac{1-\gamma_1}{1-\gamma_2}\right)^2\cdot \inf_{\mathbf{w}:|\mathbf{w}| \leq 1} \frac{(1-\gamma_2)(1 - |\mathbf{w}|^2 + \gamma_2(w_3-1)^2)}{(1-\gamma_1)(1 - |\mathbf{w}|^2 + \gamma_1(w_3-1)^2)} \\
	&\geq \left(\frac{1-\gamma_1}{1-\gamma_2}\right) \min\left\{1, \frac{\gamma_2}{\gamma_1}\right\} > 0
\end{align*}
\end{proof}

With Propositions~\ref{prop5} and \ref{prop6}, we now know that many qubit channels can have a positive expansion coefficient by restricting the domain appropriately; in other words, we can obtain positive relative expansion coefficients when we are comparing two channels that are similar. 

Due to the unitary invariance of the Riemannian semi-norm (or the $\chi^2$-divergence), we can deduce positive expansion coefficients for bit-flip and $Y$-error channels from the result for dephasing. All of the other Pauli channels are in fact strictly positive channels (they have full-rank output). This implies the following result for qubit Pauli channels by Proposition~\ref{prop5} and Lemma~\ref{lem:2norm}:
\begin{corollary}[Qubit Pauli's have Positive Expansion Coefficients]\label{cor:paulicoeffpos}
    Let $\Phi:\mathcal{B}(\mathbb{C}^2)\to\mathcal{B}(\mathbb{C}^2)$ be any qubit Pauli channel, then we have
    \begin{equation*}
        \widecheck\eta_\kappa^\mathrm{Riem}(\Phi;\mathrm{Im}\:\Phi)\equiv\widecheck\eta_\kappa^\mathrm{Riem}(\Phi^2,\Phi)>0
    \end{equation*}
\end{corollary}

\section{Conclusion and Open Problems}
In this work, we have investigated the intricate nature of relative expansion coefficients. Our first major result (Section~\ref{sec:nordpi}), showed that there is no reverse data processing inequality (over all states) for a large family of quantum channels, for any monotone quantum $f$-divergence (Theorem~\ref{nordpi}). However, the argument was so tied to the pigeonhole principle, that erasure channels do happen to have this reverse DPI. This property was further contrasted with cases where restricting the domain of the quantum channels can result in a positive expansion coefficient, notably in the case of primitive quantum channels for a reverse quantum Markov convergence theorem (Theorem~\ref{thm:primcoeffpos} and Corollary~\ref{cor:qmkv}) and qubit channels such as the Pauli channels (Corollary~\ref{cor:paulicoeffpos}). All of the demonstrated cases of positive expansion coefficients also allude to bounds (Corollaries~\ref{cor:BKMrecovery} and \ref{cor:alpharecover}) on the accuracy of recovery maps; this signifies a development in the interpretation of expansion coefficients as a measure of preserved information.

Connections between divergence relative expansion coefficients and Riemannian relative expansion coefficients are particularly interesting, because they teach us about the connection between the standard $f$-divergences and their local behaviour. We noted the importance of integral relationships between standard $f$-divergences and their induced Riemannian semi-norms, and constructed new cases of generic equality between divergence and Riemannian coefficients in Theorem~\ref{thm:equalitycases}. Further, we saw that these relations perfectly determine the standard $f$-divergence (Theorem~\ref{thm:classicspecialrel}). 

To study the relationship between different distinguishability measures comprehensively, we introduced an equivalence framework. This allowed us to construct even more cases where the local behaviour of standard $f$-divergences influences the preserved proportion of distinguishability, from the cases of generic equality (Theorem~\ref{thm:inherit}). However, we also saw that the bounded and unbounded cases behave quite differently. This was most strikingly observed in Theorem~\ref{thm:RiemInequiv}, which provided a family of primitive channels where the Riemannian contraction coefficients can be arbitrarily many orders of magnitude apart across the bounded and unbounded cases. This provided the first examples of inequivalence between relative expansion coefficients, and proved that the bounded-case Riemannian relative expansion coefficients form an entire equivalence class. This also has applications to the quantum Markov convergence theorem, since one therefore has to be very careful about the Riemannian contraction coefficient chosen to upper bound the mixing time. 

Another consequence of these new interpretations of relative expansion coefficients is a new series of open problems. First of all, can there be more cases of generic equality between divergence and Riemannian coefficients? In particular, are integral relations between bounded standard $f$-divergences and Riemannian semi-norms possible? Secondly, in the bounded $\kappa_f$ case, are the Riemannian relative expansion coefficients inequivalent (over all pairs of quantum channels $\mathcal{N},\mathcal{M}$) to the corresponding divergence relative expansion coefficients? Thirdly, based on a claim by \cite{Lesniewski_1999}, are there any cases where a divergence relative expansion coefficient is zero, but some Riemannian relative expansion coefficient for the same pair of channels is strictly positive? If the answer to this final question is negative, then the problem of showing that a proportion of information is preserved for a standard $f$-divergence or Riemannian semi-norm will reduce entirely, though we will not be able to deduce exactly the proportion (which Theorem~\ref{thm:RiemInequiv} already tells us is not always possible). Finally, thinking back to the convergence theorems for primitive channels (Theorems~\ref{thm:convergencethm} and \ref{cor:qmkv}), how different are our contraction and expansion coefficients from the coefficients that have the fixed point of the channel as a reference state?

\bibliographystyle{marcotomPB-appearance}
\bibliography{references}

\begin{thebibliography}{10}

\bibitem{jozsa1994fidelity}
R.~Jozsa.
\newblock {\em ``Fidelity for mixed quantum states''}.
\newblock Journal of modern optics {\bf 41}(12):\,2315--2323, (1994).

\bibitem{fuchs2002cryptographic}
C.~A. Fuchs and J.~Van De~Graaf.
\newblock {\em ``Cryptographic distinguishability measures for quantum-mechanical states''}.
\newblock IEEE Transactions on Information Theory {\bf 45}(4):\,1216--1227, (2002).

\bibitem{holevo1973statistical}
A.~S. Holevo.
\newblock {\em ``Statistical decision theory for quantum systems''}.
\newblock Journal of multivariate analysis {\bf 3}(4):\,337--394, (1973).

\bibitem{helstrom1969quantum}
C.~W. Helstrom.
\newblock {\em ``Quantum detection and estimation theory''}.
\newblock Journal of Statistical Physics {\bf 1}(2):\,231--252, (1969).

\bibitem{audenaert2007discriminating}
K.~M. Audenaert, J.~Calsamiglia, R.~Munoz-Tapia, E.~Bagan, L.~Masanes, A.~Acin, and F.~Verstraete.
\newblock {\em ``Discriminating states: The quantum Chernoff bound''}.
\newblock Physical review letters {\bf 98}(16):\,160501, (2007).

\bibitem{nussbaum2009chernoff}
M.~Nussbaum and A.~Szkoła.
\newblock {\em ``The Chernoff lower bound for symmetric quantum hypothesis testing''}.
\newblock \href{http://dx.doi.org/10.1214/08-aos593}{The Annals of Statistics {\bf 37}(2)} (2009).

\bibitem{nussbaum2011asymptotic}
M.~Nussbaum and A.~Szkoła.
\newblock {\em ``An asymptotic error bound for testing multiple quantum hypotheses''}.
\newblock \href{http://dx.doi.org/10.1214/11-aos933}{The Annals of Statistics {\bf 39}(6)} (2011).

\bibitem{muller2013quantum}
M.~M{\"u}ller-Lennert, F.~Dupuis, O.~Szehr, S.~Fehr, and M.~Tomamichel.
\newblock {\em ``On quantum R{\'e}nyi entropies: A new generalization and some properties''}.
\newblock Journal of Mathematical Physics {\bf 54}(12), (2013).

\bibitem{hiai1991proper}
F.~Hiai and D.~Petz.
\newblock {\em ``The proper formula for relative entropy and its asymptotics in quantum probability''}.
\newblock Communications in mathematical physics {\bf 143}(1):\,99--114, (1991).

\bibitem{ogawa2005strong}
T.~Ogawa and H.~Nagaoka.
\newblock {\em ``Strong converse and Stein’s lemma in quantum hypothesis testing, Asymptotic Theory Of Quantum Statistical Inference: Selected Papers, 28-42''}, (2005).

\bibitem{preskill1998lecture}
J.~Preskill.
\newblock {\em ``Lecture Notes for Physics 229: Quantum Information and Computation''}.
\newblock \url{https://www.lorentz.leidenuniv.nl/quantumcomputers/literature/preskill_1_to_6.pdf}, (1998).
\newblock Lecture Notes, Lorentz Institute for Theoretical Physics.

\bibitem{Wildebook}
M.~M. Wilde.
\newblock {\em Quantum Information Theory}.
\newblock Cambridge University Press (2017).

\bibitem{schumacher1996quantum}
B.~Schumacher and M.~A. Nielsen.
\newblock {\em ``Quantum data processing and error correction''}.
\newblock Physical Review A {\bf 54}(4):\,2629, (1996).

\bibitem{lloyd1997capacity}
S.~Lloyd.
\newblock {\em ``Capacity of the noisy quantum channel''}.
\newblock Physical Review A {\bf 55}(3):\,1613, (1997).

\bibitem{csiszar1963informationstheoretische}
I.~Csisz{\'a}r.
\newblock {\em ``Eine informationstheoretische Ungleichung und ihre Anwendung auf den Beweis der Ergodizit{\"a}t von Markoffschen Ketten''}.
\newblock A Magyar Tudom{\'a}nyos Akad{\'e}mia Matematikai Kutat{\'o} Int{\'e}zet{\'e}nek K{\"o}zlem{\'e}nyei {\bf 8}(1-2):\,85--108, (1963).

\bibitem{Petz84}
D.~Petz.
\newblock {\em ``Quasi-entropies for finite quantum systems''}.
\newblock Reports on mathematical physics {\bf 23}(1):\,57--65, (1986).

\bibitem{Petz86}
D.~Petz.
\newblock {\em ``Sufficient subalgebras and the relative entropy of states of a von Neumann algebra''}.
\newblock Communications in mathematical physics {\bf 105}(1):\,123--131, (1986).

\bibitem{csiszar1967topological}
I.~Csisz{\'a}r.
\newblock {\em ``On topological properties of f-divergence''}.
\newblock Studia Sci. Math. Hungar. {\bf 2}:\,330--339, (1967).

\bibitem{liese1987convex}
F.~Liese and I.~Vajda.
\newblock {\em ``Convex statistical distances''}.
\newblock (No Title) , (1987).

\bibitem{liese2006divergences}
F.~Liese and I.~Vajda.
\newblock {\em ``On divergences and informations in statistics and information theory''}.
\newblock IEEE Transactions on Information Theory {\bf 52}(10):\,4394--4412, (2006).

\bibitem{SV16}
I.~Sason and S.~Verdu.
\newblock {\em ``$f$ -Divergence Inequalities''}.
\newblock \href{http://dx.doi.org/10.1109/tit.2016.2603151}{IEEE Transactions on Information Theory {\bf 62}(11):\,5973–6006} (2016).

\bibitem{HM17}
F.~Hiai and M.~Mosonyi.
\newblock {\em ``Different quantum f-divergences and the reversibility of quantum operations''}.
\newblock \href{http://dx.doi.org/10.1142/s0129055x17500234}{Reviews in Mathematical Physics {\bf 29}(07):\,1750023} (2017).

\bibitem{schumacher2000relative}
B.~Schumacher and M.~D. Westmoreland.
\newblock {\em ``Relative entropy in quantum information theory''}.
\newblock arXiv preprint quant-ph/0004045 , (2000).

\bibitem{lindblad1975}
G.~Lindblad.
\newblock {\em ``Completely positive maps and entropy inequalities''}.
\newblock \href{http://dx.doi.org/https://doi.org/10.1007/BF01609396}{Communications in Mathematical Physics {\bf 40}:\,147--151} (1975).

\bibitem{Hiai_2011}
F.~Hiai, M.~Mosonyi, D.~Petz, and C.~B{\'e}ny.
\newblock {\em ``Quantum f-divergences and error correction''}.
\newblock \href{http://dx.doi.org/10.1142/s0129055x11004412}{Reviews in Mathematical Physics {\bf 23}(07):\,691–747} (2011).

\bibitem{HT24}
C.~Hirche and M.~Tomamichel.
\newblock {\em ``Quantum Rényi and f-Divergences from Integral Representations''}.
\newblock \href{http://dx.doi.org/10.1007/s00220-024-05087-3}{Communications in Mathematical Physics {\bf 405}(9)} (2024).

\bibitem{OP93}
M.~Ohya and D.~Petz.
\newblock {\em Quantum Entropy and Its Use}.
\newblock Springer Berlin Heidelberg (2004).

\bibitem{GT25}
I.~George and M.~Tomamichel.
\newblock {\em ``A Unified Approach to Quantum Contraction and Correlation Coefficients''}.
\newblock arXiv preprint arXiv:2505.15281 , (2025).

\bibitem{PR98}
D.~Petz and M.~B. Ruskai.
\newblock {\em ``Contraction of generalized relative entropy under stochastic mappings on matrices''}.
\newblock Infinite Dimensional Analysis, Quantum Probability and Related Topics {\bf 1}(01):\,83--89, (1998).

\bibitem{Matsumoto15}
K.~Matsumoto.
\newblock {\em ``A new quantum version of f-divergence''}.
\newblock In {\em Nagoya Winter Workshop: Reality and Measurement in Algebraic Quantum Theory}, pages 229--273, (2015).

\bibitem{H2024}
C.~Hirche.
\newblock {\em ``Quantum Doeblin Coefficients: A Simple Upper Bound on Contraction Coefficients''}.
\newblock In {\em 2024 IEEE International Symposium on Information Theory (ISIT)}, pages 557--562, (2024), \text{\href{http://dx.doi.org/10.1109/ISIT57864.2024.10619667}{DOI:\,10.1109/ISIT57864.2024.10619667}}.

\bibitem{beigi2025}
S.~Beigi, C.~Hirche, and M.~Tomamichel.
\newblock {\em ``Some properties and applications of the new quantum $f$-divergences''}, (2025).
\newblock Available online: \url{https://arxiv.org/abs/2501.03799}.

\bibitem{liu2025}
P.-C. Liu, C.~Hirche, and H.-C. Cheng.
\newblock {\em ``Layer Cake Representations for Quantum Divergences''}, (2025).
\newblock Available online: \url{https://arxiv.org/abs/2507.07065}.

\bibitem{LR99}
A.~Lesniewski and M.~B. Ruskai.
\newblock {\em ``{Monotone Riemannian metrics and relative entropy on noncommutative probability spaces}''}.
\newblock \href{http://dx.doi.org/10.1063/1.533053}{Journal of Mathematical Physics {\bf 40}(11):\,5702--5724} (1999).

\bibitem{HR15}
F.~Hiai and M.~B. Ruskai.
\newblock {\em ``{Contraction coefficients for noisy quantum channels}''}.
\newblock \href{http://dx.doi.org/10.1063/1.4936215}{Journal of Mathematical Physics {\bf 57}(1):\,015211} (2015).

\bibitem{HRS22}
C.~Hirche, C.~Rouz{\'{e}}, and D.~Stilck~Fran{\c{c}}a.
\newblock {\em ``On contraction coefficients, partial orders and approximation of capacities for quantum channels''}.
\newblock \href{http://dx.doi.org/10.22331/q-2022-11-28-862}{{Quantum} {\bf 6}:\,862} (2022).

\bibitem{GZB24}
I.~George, A.~Zheng, and A.~Bansal.
\newblock {\em ``Divergence inequalities with applications in ergodic theory''}.
\newblock arXiv preprint arXiv:2411.17241 , (2024).

\bibitem{RISB21}
N.~Ramakrishnan, R.~Iten, V.~B. Scholz, and M.~Berta.
\newblock {\em ``Computing Quantum Channel Capacities''}.
\newblock \href{http://dx.doi.org/10.1109/tit.2020.3034471}{IEEE Transactions on Information Theory {\bf 67}(2):\,946–960} (2021).

\bibitem{BGSW24}
P.~Belzig, L.~Gao, G.~Smith, and P.~Wu.
\newblock {\em ``Reverse-type Data Processing Inequality''}, (2024).
\newblock Available online: \url{https://arxiv.org/abs/2411.19890}.

\bibitem{Choi1994}
M.-D. Choi, M.~B. Ruskai, and E.~Seneta.
\newblock {\em ``Equivalence of certain entropy contraction coefficients''}.
\newblock Linear algebra and its applications {\bf 208}:\,29--36, (1994).

\bibitem{Cohen1993}
J.~E. Cohen, Y.~Iwasa, G.~Rautu, M.~B. Ruskai, E.~Seneta, and G.~Zbaganu.
\newblock {\em ``Relative entropy under mappings by stochastic matrices''}.
\newblock Linear algebra and its applications {\bf 179}:\,211--235, (1993).

\bibitem{Temme_2010}
K.~Temme, M.~J. Kastoryano, M.~B. Ruskai, M.~M. Wolf, and F.~Verstraete.
\newblock {\em ``The $\chi^2$-divergence and mixing times of quantum Markov processes''}.
\newblock \href{http://dx.doi.org/10.1063/1.3511335}{Journal of Mathematical Physics {\bf 51}(12)} (2010).

\bibitem{AE11}
K.~M.~R. Audenaert and J.~Eisert.
\newblock {\em ``{Continuity bounds on the quantum relative entropy — II}''}.
\newblock \href{http://dx.doi.org/10.1063/1.3657929}{Journal of Mathematical Physics {\bf 52}(11):\,112201} (2011).

\bibitem{GR22}
L.~Gao and C.~Rouzé.
\newblock {\em ``Complete Entropic Inequalities for Quantum Markov Chains''}.
\newblock \href{http://dx.doi.org/10.1007/s00205-022-01785-1}{Archive for Rational Mechanics and Analysis {\bf 245}(1):\,183–238} (2022).

\bibitem{WBCDT24}
L.~H. Wolff, P.~Belzig, M.~Christandl, B.~Durhuus, and M.~Tomamichel.
\newblock {\em ``Fundamental Limit on the Power of Entanglement Assistance in Quantum Communication''}.
\newblock \href{http://dx.doi.org/10.1103/PhysRevLett.134.020802}{Phys. Rev. Lett. {\bf 134}:\,020802} (2025).

\bibitem{GJLL22}
L.~Gao, M.~Junge, N.~LaRacuente, and H.~Li.
\newblock {\em ``Relative entropy decay and complete positivity mixing time''}, (2023).
\newblock Available online: \url{https://arxiv.org/abs/2209.11684}.

\bibitem{LS23}
N.~Laracuente and G.~Smith.
\newblock {\em ``Information Fragility or Robustness Under Quantum Channels''}, (2023).
\newblock Available online: \url{https://arxiv.org/abs/2312.17450}.

\bibitem{araiza2025transportationcostcontractioncoefficient}
R.~Araiza, M.~Junge, and P.~Wu.
\newblock {\em ``Transportation cost and contraction coefficient for channels on von Neumann algebras''}, (2025).
\newblock Available online: \url{https://arxiv.org/abs/2506.04197}.

\bibitem{gao2023sufficient}
L.~Gao, H.~Li, I.~Marvian, and C.~Rouz{\'e}.
\newblock {\em ``Sufficient statistic and recoverability via Quantum Fisher Information metrics''}.
\newblock arXiv preprint arXiv:2302.02341 , (2023).

\bibitem{junge2018universal}
M.~Junge, R.~Renner, D.~Sutter, M.~M. Wilde, and A.~Winter.
\newblock {\em ``Universal Recovery Maps and Approximate Sufficiency of Quantum Relative Entropy''}.
\newblock \href{http://dx.doi.org/10.1007/s00023-018-0716-0}{Annales Henri Poincaré {\bf 19}(10):\,2955–2978} (2018).

\bibitem{AliSilvey1966}
S.~M. Ali and S.~D. Silvey.
\newblock {\em ``A General Class of Coefficients of Divergence of One Distribution from Another''}.
\newblock \href{http://dx.doi.org/10.1111/j.2517-6161.1966.tb00626.x}{Journal of the Royal Statistical Society: Series B (Methodological) {\bf 28}(1):\,131--142} (1966).

\bibitem{csiszar1967information}
I.~Csisz{\'a}r.
\newblock {\em ``On information-type measure of difference of probability distributions and indirect observations''}.
\newblock Studia Sci. Math. Hungar. {\bf 2}:\,299--318, (1967).

\bibitem{morimoto1963markov}
T.~Morimoto.
\newblock {\em ``Markov processes and the H-theorem''}.
\newblock Journal of the Physical Society of Japan {\bf 18}(3):\,328--331, (1963).

\bibitem{HKPR13}
F.~Hiai, H.~Kosaki, D.~Petz, and M.~B. Ruskai.
\newblock {\em ``Families of completely positive maps associated with monotone metrics''}.
\newblock Linear Algebra and its Applications {\bf 439}(7):\,1749--1791, (2013).

\bibitem{davies1976quantum}
E.~B. Davies.
\newblock {\em Quantum Theory of Open Systems}.
\newblock Academic Press, London (1976).

\bibitem{SGWC10}
M.~Sanz, D.~Perez-Garcia, M.~M. Wolf, and J.~I. Cirac.
\newblock {\em ``A quantum version of Wielandt's inequality''}.
\newblock IEEE Transactions on Information Theory {\bf 56}(9):\,4668--4673, (2010).

\bibitem{kastoryano2012cutoff}
M.~J. Kastoryano, D.~Reeb, and M.~M. Wolf.
\newblock {\em ``A cutoff phenomenon for quantum Markov chains''}.
\newblock Journal of Physics A: Mathematical and Theoretical {\bf 45}(7):\,075307, (2012).

\bibitem{Lesniewski_1999}
A.~Lesniewski and M.~B. Ruskai.
\newblock {\em ``Monotone Riemannian metrics and relative entropy on noncommutative probability spaces''}.
\newblock \href{http://dx.doi.org/10.1063/1.533053}{Journal of Mathematical Physics {\bf 40}(11):\,5702–5724} (1999).

\end{thebibliography}
\end{document}